\documentclass[a4paper,onecolumn,unpublished,11pt]{quantumarticle}
\pdfoutput=1
\usepackage{graphicx} 
\usepackage{amsmath}
\usepackage{amssymb}
\usepackage[pagebackref, colorlinks = true, linkcolor = myred, urlcolor  = myred, citecolor = myred]{hyperref}
\usepackage{mathtools}
\usepackage{float}
\usepackage{amsthm}
\usepackage{physics}
\usepackage{dsfont}
\usepackage{longtable}
\usepackage{braket}
\usepackage{tikz}
\usepackage{enumitem} 
\usepackage[T1]{fontenc}

\usepackage{tcolorbox}
\tcbuselibrary{breakable,skins}

\usetikzlibrary{backgrounds,positioning,shapes.geometric,decorations.markings,arrows,knots,calc,decorations.pathmorphing,decorations.pathreplacing,calligraphy,calc,shapes,fit,decorations.text}
\tikzset{snake it/.style={
    decorate,
    decoration={snake,amplitude=0.07cm,post=lineto, post length=0.0cm},
}}
\tikzset{green_box/.style={
    very thick,
    text=mygreen!40!black,
    draw=mygreen!80!white,
    rounded corners,
    minimum height=1.1cm,
    minimum width=1.1cm,
    fill=mygreen!30,
}}
\tikzset{blue_box/.style={
    very thick,
    text=myblue!40!black,
    draw=myblue!80!white,
    rounded corners,
    minimum height=1.1cm,
    minimum width=1.1cm,
    fill=myblue!30,
}}

\newcommand{\gear}[6]{%
	(0:#2)
	\foreach \i [evaluate=\i as \n using {\i-1)*360/#1}] in {1,...,#1}{%
		arc (\n:\n+#4:#2) {[rounded corners=1.5pt] -- (\n+#4+#5:#3)
			arc (\n+#4+#5:\n+360/#1-#5:#3)} --  (\n+360/#1:#2)
	}%
	(0,0) circle[radius=#6] 
}

\definecolor{LightGray}{RGB}{220,220,220}
\definecolor{myred2}{RGB}{255, 19, 0}
\definecolor{myred}{RGB}{15, 122, 98}
\definecolor{myblue}{RGB}{14, 81, 167}
\definecolor{myorange}{RGB}{255, 129, 0}
\definecolor{mygreen}{RGB}{0, 146, 44}

\newtheorem{thm}{Theorem}[section]
\newtheorem*{thm*}{Theorem}
\newtheorem{lem}[thm]{Lemma}
\newtheorem{cor}[thm]{Corollary}

\newtheorem{defi}[thm]{Definition}

\newtheorem{remark}[thm]{Remark}

\newtheorem{condition}[thm]{Condition}

\newcommand{\cA}{\mathcal{A}}
\newcommand{\cB}{\mathcal{B}}
\newcommand{\cC}{\mathcal{C}}
\newcommand{\cE}{\mathcal{E}}
\newcommand{\cF}{\mathcal{F}}
\newcommand{\cI}{\mathcal{I}}
\newcommand{\cK}{\mathcal{K}}
\newcommand{\cM}{\mathcal{M}}
\newcommand{\cN}{\mathcal{N}}
\newcommand{\cR}{\mathcal{R}}
\newcommand{\cS}{\mathcal{S}}
\newcommand{\cT}{\mathcal{T}}

\newcommand{\cX}{\mathcal{X}}
\newcommand{\cY}{\mathcal{Y}}
\newcommand{\Efin}{E_\mathrm{fin}}
\newcommand{\opt}{\mathrm{opt}} 

\begin{document}

\title{Computing key rates for one-sided\newline device-independent quantum key distribution}

\author{Andreas Bluhm}
\email{andreas.bluhm@univ-grenoble-alpes.fr}
\affiliation{Univ. Grenoble Alpes, CNRS, Grenoble INP, LIG, 38000 Grenoble, France}
\author{Gereon Koßmann}
\email{gereonkossmann@gmail.com}
\affiliation{Institute for Quantum Information, RWTH Aachen, 52062 Aachen, Germany}
\author{Martin Sandfuchs}
\email{martisan@phys.ethz.ch}
\affiliation{Institute for Theoretical Physics, ETH Zürich, 8093 Zurich, Switzerland}
\author{René Schwonnek}
\email{rene.schwonnek@itp.uni-hannover.de}
\affiliation{Institute for Theoretical Physics, Leibniz Universität Hannover, 30167 Hannover, Germany}
\author{Giuseppe Viola}
\email{giuseppe.viola@uni-siegen.de}
\affiliation{Naturwissenschaftliche-Technische Fakultät, Universität Siegen, 57068 Siegen, Germany}
\author{Ramona Wolf}
\email{ramona.wolf@uibk.ac.at}
\affiliation{Institute for Theoretical Physics, Universität Innsbruck, 6020 Innsbruck, Austria}

\maketitle

\fontfamily{lmr}\selectfont

\begin{abstract}
     The defining feature of one-sided device-independent quantum key distribution is its asymmetric trust model in which only one party is characterized. This scenario sets an interesting middle ground between high key rates achievable by characterizing devices and the security of full device-independence. Here, we provide  new tools, methods, and benchmarks for calculating key rates in this setting.
    To achieve this, we develop and compare two extensions of the NPA hierarchy and derive finite-size security bounds against general attacks with arbitrary device memory. The latter is based on the Generalized Entropy Accumulation Theorem. 
    We then investigate the performance of various protocols: the BB84 protocol both with and without losses, a qutrit mutually unbiased bases protocol, and protocols based on Bell inequalties such as CHSH and $I_{3322}$.
     We find that the choice of which  party is characterized can strongly affects the key rate, surprisingly without a universal ordering. Our work thus offers a general toolbox for calculating key rates for one-sided device-independent QKD protocols.
\end{abstract}

\vspace{-1cm}
\newpage
\section{Introduction}

Compared with many other technologies, cryptographic technologies are distinguished by the fact that a functioning device is characterized not only by what it can do, but also by what it cannot allow. The former can be demonstrated directly in operation, whereas the latter concerns the absence of successful attacks and cannot be established by testing alone. This verification asymmetry becomes particularly pronounced in quantum key distribution (QKD) \cite{BennettBrassard1984, Ekert1991}.  A successful implementation must both distribute cryptographic keys over a quantum channel and support a mathematically formulated guarantee of their secrecy against all admissible adversaries with quantum resources. Without the device \cite{diamanti2016practical}, no useful keys are generated. Without the proof \cite{Pirandola_2020,primaatmaja2023security}, the generated keys carry no well-defined claim of secrecy.

QKD consequently relies on a close interplay between experiment and theory \cite{diamanti2016practical,Pirandola_2020,lo2014secure,PortmannRenner2022}. Only their combination yields a technology that is both functional and supported by a rigorous security guarantee \cite{Renner05}. The strength of this guarantee, however, depends on how accurately the mathematical model underlying the proof captures the relevant properties of the physical implementation.
Every security model must therefore specify which properties of the implementation are trusted and which are inferred from observed data. Different QKD paradigms correspond to different choices of these trust assumptions, as depicted in Figure~\ref{fig:characterise}. Device-dependent QKD \cite{BennettBrassard1984,Ekert1991,BennettBrassardMermin1992} directly incorporates detailed models of the devices, whereas device-independent QKD (DIQKD) \cite{Pironio2009,primaatmaja2023security} seeks to certify security from the observed statistics alone, typically through a Bell test \cite{Bell1964}. \sloppy{This minimizes implementation-specific assumptions but requires demanding experimental conditions \cite{Zhang2021,Liu2021a,Nadlinger2021,Lu2026}. Semi-device-independent approaches \cite{Branciard2012,Tan2021_key_rates,Masini_2024,Mothsara_2026,ZengEtAl2026} retain selected device assumptions while treating the remaining components as untrusted.}

In this work we study protocols in which trust is distributed asymmetrically between the two parties. One party’s devices are characterised, while those of the other are treated as a black box, a setting known as one-sided DIQKD \cite{Branciard2012}. It provides a natural middle ground between detailed device modelling and fully device-independent certification, combining substantially reduced device assumptions with less demanding experimental requirements. This setting poses new challenges for security proofs: For device-dependent protocols, a variety of numerical techniques have been developed based on convex optimisation and semidefinite programming \cite{Winick2018_key_rates,Hu2022,Arajo2023,HamzaFawzi_IPM_relentro,kossmann2025optimisingrelativeentropysemidefinite,navarro2026finitesizequantumkeydistribution}. Similarly, device-independent security proofs can be treated numerically using relaxations of the set of quantum correlations, most prominently through the NPA hierarchy \cite{Navascus2007_npa1,Navascus2008_npa2,Tan2021_key_rates,brown2024device, kossmann2025reliableentropyestimationobserved}. Semi-device-independent protocols, however, combine trusted and untrusted components within the same security analysis, leading to optimisation problems that differ structurally from both settings.

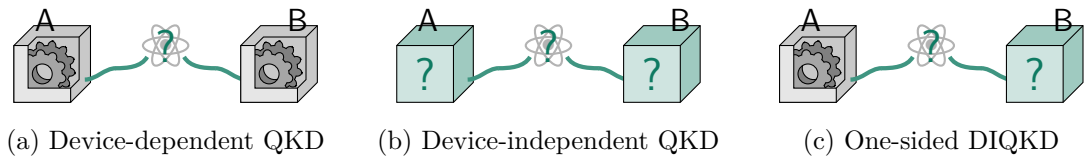
\begin{figure}[t]
	\centering
	\begin{tikzpicture}[scale=0.75]
		\begin{scope}[xshift=-6.5cm]
			\node at (2.675,-.7) {\small (a) Device-dependent QKD};
			\draw[fill=gray!42] (1,0) -- (1.35,0.35) -- (1.35,1.35) -- (0.35,1.35) -- (0,1) -- (0,0) -- cycle;
			\begin{scope}[xshift=0.6cm, yshift=0.6cm,scale=0.2]
				\filldraw[fill=gray,even odd rule] \gear{12}{2}{2.4}{10}{2}{1};
			\end{scope}
			\begin{scope}[xshift=0.5cm, yshift=0.5cm,scale=0.175]
				\filldraw[fill=gray!70,even odd rule] \gear{10}{2}{2.4}{10}{2}{1};
			\end{scope}
			\draw[fill=gray!20] (0,0) -- (1,0) -- (1,0.25) -- (0.25,0.25) -- (0.25,1) -- (0,1) -- cycle;
			\draw (1,0.25) -- (1.15,0.4) -- (1.15,1.15) -- (1.35,1.35);
			\draw (1.15,1.15) -- (0.4,1.15) -- (0.25,1);
			\node at (0.55,1.45) {\large \textsf{A}};
			\begin{scope}[xshift=4cm]
				\begin{scope}[xscale=-1,xshift=-1.35cm]
					\begin{knot}
						\strand [ultra thick, color=myred!80, looseness=0.9] (2.35,0.8)
						to [out=240,in=right] (1.79,0.6);
						\strand [ultra thick, color=myred!80, looseness=0.9] (1.8,0.6)
						to [out=left, in=right] (1.25,0.4);
					\end{knot}
				\end{scope}
				\draw[fill=gray!42] (1,0) -- (1.35,0.35) -- (1.35,1.35) -- (0.35,1.35) -- (0,1) -- (0,0) -- cycle;
				\begin{scope}[xshift=0.6cm, yshift=0.6cm,scale=0.2]
					\filldraw[fill=gray,even odd rule] \gear{12}{2}{2.4}{10}{2}{1};
				\end{scope}
				\begin{scope}[xshift=0.5cm, yshift=0.5cm,scale=0.175]
					\filldraw[fill=gray!70,even odd rule] \gear{10}{2}{2.4}{10}{2}{1};
				\end{scope}
				\draw[fill=gray!20] (0,0) -- (1,0) -- (1,0.25) -- (0.25,0.25) -- (0.25,1) -- (0,1) -- cycle;
				\draw (1,0.25) -- (1.15,0.4) -- (1.15,1.15) -- (1.35,1.35);
				\draw (1.15,1.15) -- (0.4,1.15) -- (0.25,1);
				\node at (1,1.45) {\large \textsf{B}};
			\end{scope}
			\node[color=black!30] at (2.675,1) {$\bullet$};
			\draw[color=black!30,thick] (2.675,1) ellipse (0.4cm and 0.15cm);
			\draw[color=black!30,thick,rotate around={60:(2.675,1)}] (2.675,1) ellipse (0.4cm and 0.15cm);
			\draw[color=black!30,thick,rotate around={120:(2.675,1)}] (2.675,1) ellipse (0.4cm and 0.15cm);
			\node[color=myred] at (2.675,1) {\LARGE \textsf{?}};
			\begin{knot}
				\strand [ultra thick, color=myred!80, looseness=0.9] (2.35,0.8)
				to [out=240,in=right] (1.79,0.6);
				\strand [ultra thick, color=myred!80, looseness=0.9] (1.8,0.6)
				to [out=left, in=right] (1.25,0.4);
			\end{knot}
		\end{scope}
		\begin{scope}[xshift=0.25cm]
			\node at (2.675,-.7) {\small (b) Device-independent QKD};
			\draw[fill=myred!40] (1,0) -- (1.35,0.35) -- (1.35,1.35) -- (0.35,1.35) -- (0,1);
			\draw[fill=myred!20] (0,0) -- (1,0) -- (1,1) -- (0,1) -- cycle;
			\draw (1,1) -- (1.35,1.35);
			\node[color=myred] at (0.5,0.5) {\LARGE \textsf{?}};
			\node at (0.55,1.45) {\large \textsf{A}};
			\begin{scope}[xshift=4cm]
				\begin{scope}[xscale=-1,xshift=-1.35cm]
					\begin{knot}
						\strand [ultra thick, color=myred!80, looseness=0.9] (2.35,0.8)
						to [out=240,in=right] (1.79,0.6);
						\strand [ultra thick, color=myred!80, looseness=0.9] (1.8,0.6)
						to [out=left, in=right] (1.25,0.4);
					\end{knot}
				\end{scope}
				\draw[fill=myred!40] (1,0) -- (1.35,0.35) -- (1.35,1.35) -- (0.35,1.35) -- (0,1);
				\draw[fill=myred!20] (0,0) -- (1,0) -- (1,1) -- (0,1) -- cycle;
				\draw (1,1) -- (1.35,1.35);
				\node[color=myred] at (0.5,0.5) {\LARGE \textsf{?}};
				\node at (1,1.45) {\large \textsf{B}};
			\end{scope}
			\node[color=black!30] at (2.675,1) {$\bullet$};
			\draw[color=black!30,thick] (2.675,1) ellipse (0.4cm and 0.15cm);
			\draw[color=black!30,thick,rotate around={60:(2.675,1)}] (2.675,1) ellipse (0.4cm and 0.15cm);
			\draw[color=black!30,thick,rotate around={120:(2.675,1)}] (2.675,1) ellipse (0.4cm and 0.15cm);
			\node[color=myred] at (2.675,1) {\LARGE \textsf{?}};
			\begin{knot}
				\strand [ultra thick, color=myred!80, looseness=0.9] (2.35,0.8)
				to [out=240,in=right] (1.79,0.6);
				\strand [ultra thick, color=myred!80, looseness=0.9] (1.8,0.6)
				to [out=left, in=right] (1.25,0.4);
			\end{knot}
		\end{scope}
		\begin{scope}[xshift=7cm]
			\node at (2.675,-.7) {\small (c) One-sided DIQKD};
			\draw[fill=gray!42] (1,0) -- (1.35,0.35) -- (1.35,1.35) -- (0.35,1.35) -- (0,1) -- (0,0) -- cycle;
			\begin{scope}[xshift=0.6cm, yshift=0.6cm,scale=0.2]
				\filldraw[fill=gray,even odd rule] \gear{12}{2}{2.4}{10}{2}{1};
			\end{scope}
			\begin{scope}[xshift=0.5cm, yshift=0.5cm,scale=0.175]
				\filldraw[fill=gray!70,even odd rule] \gear{10}{2}{2.4}{10}{2}{1};
			\end{scope}
			\draw[fill=gray!20] (0,0) -- (1,0) -- (1,0.25) -- (0.25,0.25) -- (0.25,1) -- (0,1) -- cycle;
			\draw (1,0.25) -- (1.15,0.4) -- (1.15,1.15) -- (1.35,1.35);
			\draw (1.15,1.15) -- (0.4,1.15) -- (0.25,1);
			\node at (0.55,1.45) {\large \textsf{A}};
			\begin{scope}[xshift=4cm]
				\begin{scope}[xscale=-1,xshift=-1.35cm]
					\begin{knot}
						\strand [ultra thick, color=myred!80, looseness=0.9] (2.35,0.8)
						to [out=240,in=right] (1.79,0.6);
						\strand [ultra thick, color=myred!80, looseness=0.9] (1.8,0.6)
						to [out=left, in=right] (1.25,0.4);
					\end{knot}
				\end{scope}
				\draw[fill=myred!40] (1,0) -- (1.35,0.35) -- (1.35,1.35) -- (0.35,1.35) -- (0,1);
				\draw[fill=myred!20] (0,0) -- (1,0) -- (1,1) -- (0,1) -- cycle;
				\draw (1,1) -- (1.35,1.35);
				\node[color=myred] at (0.5,0.5) {\LARGE \textsf{?}};
				\node at (1,1.45) {\large \textsf{B}};
			\end{scope}
			\node[color=black!30] at (2.675,1) {$\bullet$};
			\draw[color=black!30,thick] (2.675,1) ellipse (0.4cm and 0.15cm);
			\draw[color=black!30,thick,rotate around={60:(2.675,1)}] (2.675,1) ellipse (0.4cm and 0.15cm);
			\draw[color=black!30,thick,rotate around={120:(2.675,1)}] (2.675,1) ellipse (0.4cm and 0.15cm);
			\node[color=myred] at (2.675,1) {\LARGE \textsf{?}};
			\begin{knot}
				\strand [ultra thick, color=myred!80, looseness=0.9] (2.35,0.8)
				to [out=240,in=right] (1.79,0.6);
				\strand [ultra thick, color=myred!80, looseness=0.9] (1.8,0.6)
				to [out=left, in=right] (1.25,0.4);
			\end{knot}
		\end{scope}
	\end{tikzpicture}
\caption{\label{fig:characterise}\textbf{Different levels of device characterisation in QKD protocols.} Device-dependent QKD (a) assumes fully characterised measurement devices for both parties \textsf{A} and \textsf{B}, whereas DIQKD (b) requires no assumptions about the devices and relies only on observed correlations. One-sided DIQKD (c) represents an intermediate scenario, where one measurement device is characterised while the other is treated as a black box.}
\end{figure}

The aims of this work are (i) to clarify a rigorous mathematical model for this setting and to introduce numerical methods for solving the corresponding asymptotic key rate problem, (ii) to provide a self-sustained finite key rate analysis secure against general attacks, and (iii) to apply our methods to a variety of relevant protocols.   

We begin with (i) by placing the one-sided device-independent setting in the representation independent language of $C^*$-algebras \cite{Blackadar2006}. An uncharacterized measurement device is described by the universal $C^*$-algebra generated by its measurement projectors and their defining relations, whereas a characterized finite-dimensional implementation is represented by a full matrix algebra. Physical preparations are then modelled as states on the $C^*$-tensor product of the two local algebras, from which concrete Hilbert-space realizations can be recovered when needed. Crucially, the one-sided setting has additional structure: since full matrix algebras are nuclear, this tensor product is unique and can be identified with a matrix algebra over the black-box algebra. This removes tensor-product ambiguities present in more general operator-algebraic models and provides a particularly clean mathematical foundation for the security analysis and numerical methods developed below.

Building on our model, we formulate the asymptotic key-rate problem for one-sided DIQKD as a non-commutative polynomial optimization and develop two complementary extensions of the NPA hierarchy \cite{Navascus2007_npa1,Navascus2008_npa2} that incorporate the fixed matrix algebra of the characterized device. The first encodes this algebra through additional relations within an otherwise standard NPA construction, while the second keeps the characterized subsystem explicit and replaces the scalar moments of the usual hierarchy with matrix-valued ones. Both approaches yield convergent hierarchies of semidefinite programs and are compatible with the conditional entropy approximations in \cite{brown2024device} and \cite{kossmann2025reliableentropyestimationobserved}. We compare the numerical performance of the resulting methods and use them to compute asymptotic key rates for several representative one-sided device-independent protocols.

For (ii) we provide a finite-size security analysis against general attacks by adapting the Generalized Entropy Accumulation Theorem (GEAT) \cite{Metger_2022} to the one-sided device-independent setting. The uncharacterized devices may possess arbitrary internal memory and exhibit correlations across rounds, subject only to the operational requirement that the public communication does not reveal information directly from this memory. The resulting proof reduces the finite-key problem to the same single-round conditional-entropy optimization used for the asymptotic rate. Its numerical solution determines a min-tradeoff function and thereby connects our SDP methods directly to composable finite-key bounds.

Finally, for (iii) we apply our methods to a range of representative protocols, including CHSH-based QKD \cite{Pironio2009}, lossless and lossy variants of entanglement-based BB84 \cite{Branciard2012,BennettBrassardMermin1992}, a qutrit protocol based on mutually unbiased bases (MUBs), and an $I_{
3322}$-based protocol \cite{froissart1981constructive,collins2004relevant}. For each example, we compare the two possible one-sided trust assignments---either Alice or Bob is characterised---while the raw key is always generated by Alice. We find that this choice can have a substantial impact on the achievable key rate, but no universal ordering emerges: characterising the key-generating party is advantageous in several of our examples, whereas the CHSH protocol displays the opposite behaviour. Our results can therefore only provide a first exploration of this largely open design question. The methods developed here now make it possible to investigate a much broader range of protocols and to identify the structural principles governing the optimal placement of trust.

\section{The setup}

We consider an entanglement-based QKD protocol in the one-sided device-in\-de\-pen\-dent setting, depicted in Figure~\ref{fig:setup}. A high-level description of the protocol is given in Box~\hyperref[box:simpleprotocol]{1} and a detailed description is given in Box~\hyperref[box:protocol]{2}.
An untrusted source prepares a quantum system in the tripartite state $\rho_{Q_AQ_BQ_E}$ and sends part $Q_A$ to Alice and part $Q_B$ to Bob, while part $Q_E$ is kept by the adversary. Alice and Bob perform measurements on their respective local systems and store the measurement outcomes $A$ and $B$, respectively. In the one-sided DIQKD setting, one party's measurement device is treated as a black box (i.e., only the input/output correlations are used in the analysis and nothing is assumed about the inner workings of the device), whereas the other party's measurements are fully characterised. Which party is characterised (or, in other words, trusted) is part of the protocol specification.

\begin{figure}[t]
	\centering
	\begin{tikzpicture}
		\draw[draw=myred,thick, fill=myred!15, rounded corners] (0,-0.6) rectangle (4,3);
		\node[color=myred] at (0.75,3.25) {\textsf{Alice}};
		\node (W) at (1,2.1) {\includegraphics[width=0.08\textwidth]{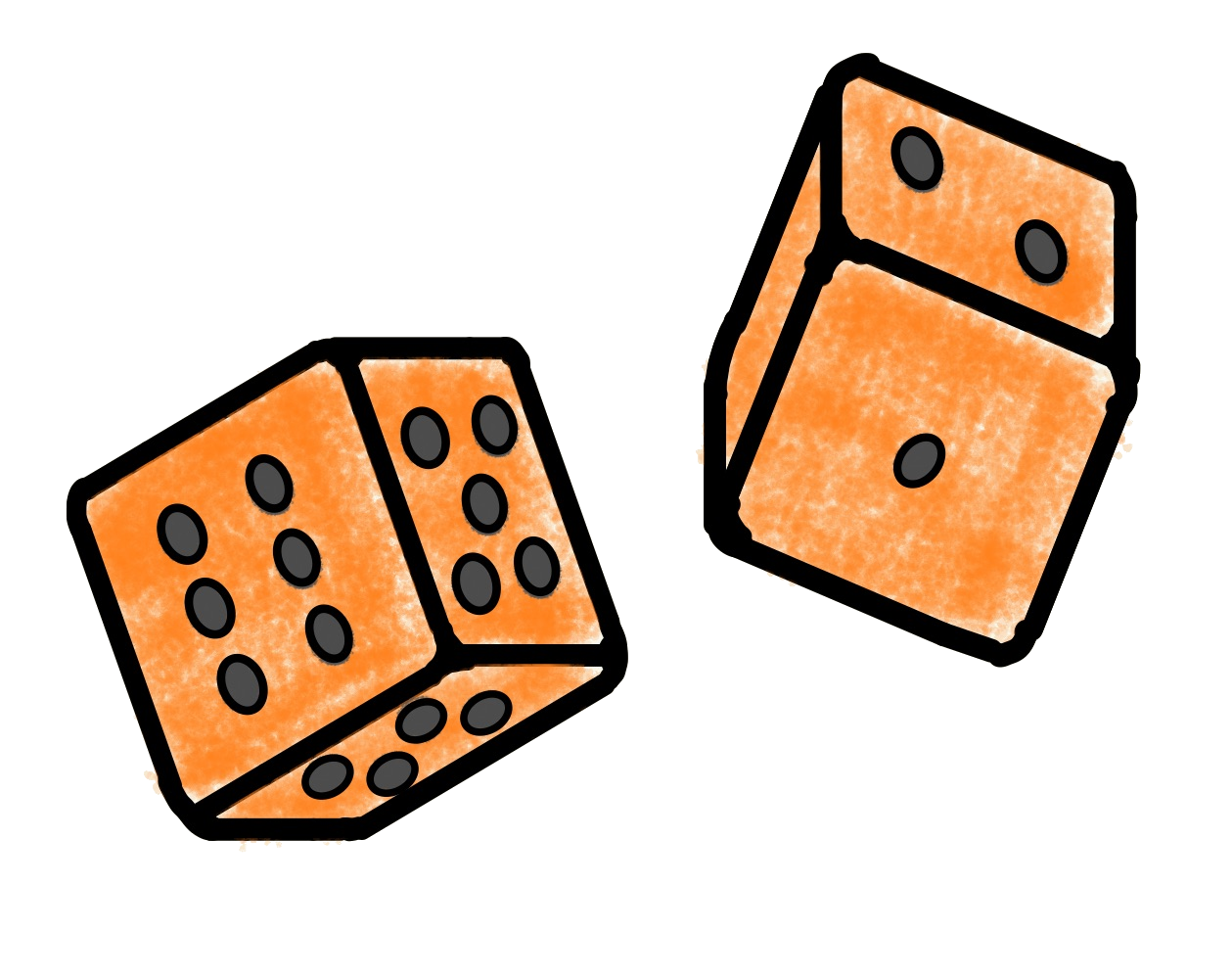}};
		\node[color=black!70] at (0.7,2.5) {\footnotesize \textsf{RNG}};
		\node (P) at (2,0.75) {\includegraphics[width=0.08\textwidth]{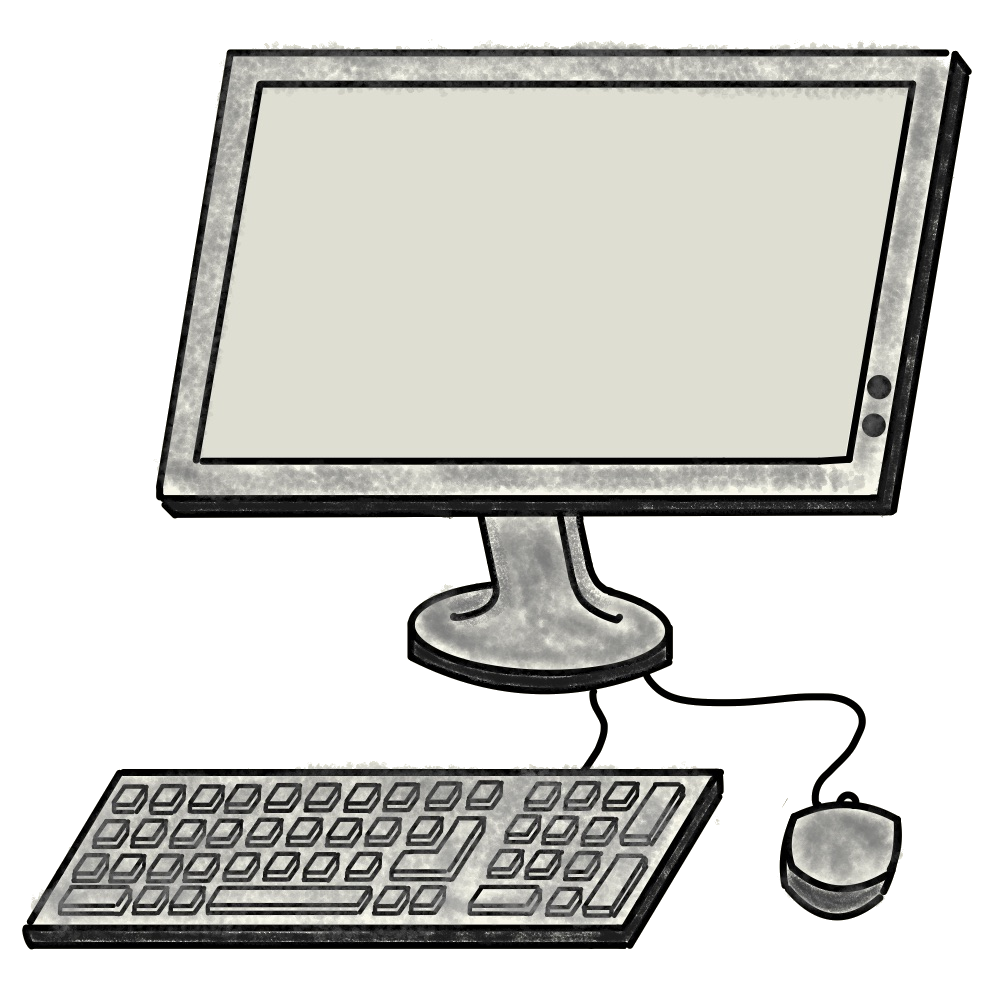}};
        \node at (2.05,1.05) {\footnotesize $S$};
		\draw[draw=black, fill=black!20] (2.6,1.6) rectangle (3.6,2.6);
		\node[color=black!70] at (3.1,2.75) {\tiny \textsf{measurement}};
		\node at (2.95,2.35) {\footnotesize $\mathcal{M}$};
		\draw[draw=black,fill=myred2!30] (2.6,1.6) -- (3.6,2.6) -- (3.6,1.6) -- cycle;
		\node[color=myred2] at (3.3,1.9) {\small \textsf{?}};
		\node at (2,-0.15) {\footnotesize $H(S|IQ_E)$};
		\draw[->,>=stealth,thick, color=black] (2,-0.45) -- (2,-1.1);
		\node at (2,-1.4) {\small $K_A$};
		\draw[->,>=stealth, thick] (W) -- (2.5,2.1);
		\draw[->,>=stealth, thick] (0.8,1.5) to [bend right] (1.4,0.75);
		\draw[->,>=stealth, thick] (3.2,1.5) to [bend left] node[right] {\footnotesize $A$} (2.6,0.75);
		\begin{scope}[xshift=12cm,xscale=-1]
			\draw[draw=myred,thick, fill=myred!15, rounded corners] (0,-0.6) rectangle (4,3);
			\node[color=myred] at (0.75,3.25) {\textsf{Bob}};
			\node[xscale=-1] (W2) at (1,2.1) {\includegraphics[width=0.08\textwidth]{Wuerfel.png}};
			\node[color=black!70] at (0.7,2.5) {\footnotesize \textsf{RNG}};
			\node (P2) at (2,0.75) {\includegraphics[width=0.08\textwidth]{PCblack.png}};
			\draw[draw=black, fill=black!20] (2.6,1.6) rectangle (3.6,2.6);
			\node[color=black!70] at (3.1,2.75) {\tiny \textsf{measurement}};
			\node at (2.95,1.85) {\footnotesize $\mathcal{N}$};
			\draw[draw=black,fill=myred2!30] (3.6,2.6) -- (3.6,1.6) -- (2.6,2.6) -- cycle;
			\node[color=myred2] at (3.3,2.3) {\small \textsf{?}};
			\draw[->,>=stealth,thick, color=black] (2,0) -- (2,-1.1);
			\node at (2,-1.4) {\small $K_B$};
			\draw[->,>=stealth, thick] (W2) -- (2.5,2.1);
			\draw[->,>=stealth, thick] (0.8,1.5) to [bend right] (1.4,0.75);
			\draw[->,>=stealth, thick] (3.2,1.5) to [bend left] node[left] {\footnotesize $B$} (2.6,0.75);
		\end{scope}
		\begin{scope}[xshift=6cm, yshift=3cm,scale=1]
			\draw[draw=black,fill=black!20] (0,0) circle(0.6cm);
			\node[color=black!70] at (0,0) {$\bullet$};
			\draw[color=black!70,very thick] (0,0) ellipse (0.4cm and 0.15cm);
			\draw[color=black!70,very thick,rotate=60] (0,0) ellipse (0.4cm and 0.15cm);
			\draw[color=black!70,very thick,rotate=120] (0,0) ellipse (0.4cm and 0.15cm);
			\draw[decoration={text along path,text color=black!70,text={|\small\sffamily| source},text align={center}},decorate] (-0.6,0.35) to [bend left=70] (0.6,0.35);
			\draw[->,>=stealth,thick,decorate, decoration={snake, segment length=3.7mm, amplitude=0.7mm}] (-0.7,-0.2) -- (-2.2,-1);
			\draw[->,>=stealth,thick,decorate, decoration={snake, segment length=3.7mm, amplitude=0.7mm}] (0.7,-0.2) -- (2.2,-1);
		\end{scope}
		\draw[<->,>=stealth,double,thick] (2.75,0.35) to node[above] {\small $I$} node[below] {\small \textsf{classical channel}} (9.25,0.35);
	\end{tikzpicture}
	\caption{
        \label{fig:setup}\textbf{Setup for one-sided DIQKD.} An untrusted source distributes a quantum state to Alice and Bob, who perform local measurements and exchange classical information $I$. One measurement device is treated as a black box, indicated by the question mark, while the other is fully characterised, denoted by $\mathcal{M}$ or $\mathcal{N}$ for Alice or Bob, respectively. Throughout this work, Alice generates the raw key $S$. The choice of which party's measurement device is characterised is part of the protocol specification.
    }
\end{figure}

Through public communication (for example, comparison of the measurement bases to discard rounds without correlations), Alice generates a raw key $S$. The classical information exchanged via the public channel during this process is denoted $I$. To turn the raw key into a secure key, Alice and Bob then have to perform classical post-processing. In addition, since in a one-sided DIQKD protocol the untrusted device can have memory, we require that $I$ does not depend on the memory of the device:
\begin{condition} \label{cond:non_signalling}
    The protocol is such that there is no signalling from the device memory to the public information $I$.
\end{condition}
    
The achievable secret key rate is determined by two quantities: the amount of uncertainty about Alice's raw key from the adversary's perspective and the information that has to be revealed during error correction. More precisely, as shown in Corollary~\ref{cor:asymptotic_key_rate}, the asymptotic key rate $r_\infty$ is given by the optimization problem 
\begin{equation} 
\label{eq:asymptotic_key_rate}
    r_\infty = \inf_{\rho} (H(S|IQ_E)_\rho - H(S|IB)_\rho).
\end{equation}
Here, $H$ denotes the conditional von Neumann entropy. It is evaluated over the state $\rho_{SIBQ_E}$, which is the state at the end of a single round (for a full expression see Theorem~\ref{thm:sdi_qkd_security}). The minimization is over all states compatible with the parameter estimation constraints, i.e., states which produce the same statistics as the honest implementation.

The above key rate formula (and the protocol in Box~\hyperref[box:simpleprotocol]{1}) assumes that the raw key is taken from Alice's measurement results (this is also called direct reconciliation). For some protocols it can be beneficial to use Bob's measurement results to generate the raw key instead. However, many protocols are symmetric when swapping the roles of Alice and Bob, and hence it does not matter which party generates the key; the only thing that matters is whether or not the key is taken from the trusted or from the untrusted party. In other words, the key rate when Bob is trusted and the key is generated by Alice is the same as when Alice is trusted and the key is generated by Bob. Therefore, for ease of presentation, we will always assume that the key is generated from Alice's measurement and either Alice's or Bob's measurements are characterized.\footnote{For the numerical implementation, however, it is easier to always trust Alice and generate the key from either Alice or Bob. By the argument above, for symmetric protocols this does not matter.}

\begin{tcolorbox}[
    title=Box 1: One-sided DIQKD protocol (simplified),
    title filled=false, 
    colback=black!5!white,
    colbacktitle=black!5!white,
    colframe=black!40!white,
    coltitle=black,
    enhanced,
    breakable,
]
\label{box:simpleprotocol}
\begin{enumerate}[itemsep=-0.2em]
    \item For a total of $n$ rounds:
    \begin{enumerate}[itemsep=-0.2em, topsep=-0.5em]
        \item Eve produces a tripartite state $\rho_{Q_AQ_BQ_E}$. She distributes system $Q_A$ to Alice and system $Q_B$ to Bob.
        \item Alice and Bob perform measurements on their respective parts of the state to obtain  outcomes $A$ and $B$, respectively.
        \item Alice and Bob communicate publicly to establish Alice's raw key $S$. The public communication transcript is stored in the register $I$. 
    \end{enumerate}
    \item Alice and Bob communicate such that Bob obtains a guess $\hat{S}$ for the raw key $S$.
    \item Alice and Bob perform parameter estimation using $\hat{S}$, $B$, and $I$. If parameter estimation fails, they abort the protocol. 
    \item Alice chooses a two-universal hash function $F$ at random and sends it to Bob. Alice and Bob apply $F$ to their raw keys $S$ and $\hat{S}$ to obtain the final secret key.
\end{enumerate}
\end{tcolorbox}

\begin{remark}
    The protocol in Box~\hyperref[box:simpleprotocol]{1} considers an entanglement based one-sided DIQKD protocol. Using the source replacement technique \cite{BennettBrassardMermin1992}, our methods can also be adapted to prepare-and-measure protocols.
\end{remark}

\section{Methods} \label{sec:methods}

To turn measured statistical data in a QKD protocol into estimates of the amount of extractable key, the remaining numerical challenge is to provide reliable bounds on a conditional entropy. In the following, we outline how to compute bounds on the conditional von Neumann entropy in a (one-sided) device-independent manner. We assume in the following that Alice, Bob, and Eve are described by quantum systems $\mathcal H_{Q_A} \otimes \mathcal H_{Q_BQ_E}$. In particular we assume that we know a concrete Hilbert space representation of the operators $\{M(a\vert x)\}_{a \in \mathcal A}$ for each $x \in \mathcal X$ on Alice's side. In contrast, we do not include any concrete knowledge about Bob's operators $\{N(b\vert y)\}_{b\in \mathcal B}$ for each $y \in \mathcal Y$ such that the only information we have about Bob is that for each $y \in \mathcal  Y$, $\{N(b\vert y)\}_{b\in \mathcal B}$ is a POVM. In a setting in which Bob is trusted instead of Alice, the systems $Q_A$, $A$ and $Q_B$, $B$ need to be switched. Under these assumptions, the optimization problem for the conditional von Neumann entropy $H(A|X=x_0, Q_E)_\rho$ of Alice's measurement outcome conditioned on measurement setting $x=x_0$ and Eve's quantum memory reads as follows:
\begin{equation} \label{eq:key-rate-opti}
\begin{aligned}
    \textrm{minimize} \quad & \quad H(A|X=x_0, Q_E)_\rho  \\
    \textrm{such~that} \quad & \quad \sum_{a \in \mathcal A, b \in \mathcal B, x \in \mathcal X, y \in \mathcal Y} c_{abxyi} \Tr[(M(a|x) \otimes N(b|y)) \rho_{Q_AQ_B}] \geq q_i \quad \forall i \in \mathcal C  \\
    \quad & \quad \sum_{b} N(b|y) = \mathds{1} \quad \forall y \in \mathcal  Y \\
    \quad & \quad N(b|y) \geq 0 \quad \forall b \in \mathcal  B, y \in \mathcal  Y  \\
    \quad & \quad \rho \in \mathcal S(\mathcal H_{Q_A} \otimes \mathcal H_{Q_BQ_E}) \mathrm{~pure}. 
\end{aligned}
\end{equation}
Here, $\mathcal C$ is a set of constraints and $c_{abxyi}$, $q_i \in \mathbb C$ are coefficients implementing them, usually coming from the observed experimental outcomes.

\begin{remark}
    In the optimization problem \eqref{eq:key-rate-opti}, Alice's measurement outcome is directly used to generate the raw key (i.e. $S = A$) and there is no classical side information $I$. These classical post-processing steps can easily be incorporated (see Appendix~\ref{sec:vonNeumann_relaxation}), but we consider the simplified problem \eqref{eq:key-rate-opti} for clarity of the exposition.
    
\end{remark}

The optimization problem \eqref{eq:key-rate-opti} includes an optimization with respect to all possible quantum systems $Q_B$ for Bob and we remark that the purification from Eve does not need to be in tensor product relation with respect to Bob (cf.~for~details~\cite{ChengDIQKD}). In \eqref{eq:key-rate-opti}, the coefficients $c_{abxyi}$ and $q_i$ can be chosen in different ways to represent how Alice and Bob use their measurement results in order to constrain the knowledge Eve could have obtained. For example, the $c_{abxyi}$ and $q_i$ could be 
chosen such that 
\begin{equation}
    \sum_{a \in \mathcal A, b \in \mathcal B, x \in \mathcal X, y \in \mathcal Y} c_{abxyi} \Tr[(M(a|x) \otimes N(b|y)) \rho_{Q_AQ_B}] \geq q_i
\end{equation}
becomes the violation of a suitable Bell inequality. Alternatively, we could use the full statistics as constraints of the form 
\begin{equation}
    q_i-\delta \leq \Tr[(M(a|x) \otimes N(b|y)) \rho_{Q_AQ_B}] \leq q_i + \delta \,.
\end{equation}
There are different ways of approximating the conditional entropy by a polynomial in non-commutative variables. One way is to use the Gauss-Radau quadrature to discretize the integral representation of the logarithm as in \cite{brown2024device}. We will refer to this approximation simply as BFF for short. Another way is to use the integral representation for the relative entropy due to Frenkel \cite{frenkel2023integral} and then again to use a quadrature rule to discretize it \cite{kossmann2025reliableentropyestimationobserved,kossmann2025optimisingrelativeentropysemidefinite}. We will refer to this approximation simply as KS for short. Either way, this approximation converts the optimization problem we want to solve into a non-commutative polynomial optimization problem. If one is interested in computing key rates for fully DIQKD, one can use the NPA hierarchy \cite{Navascus2007_npa1, Navascus2008_npa2, pironio2010convergent} to approximate \eqref{eq:key-rate-opti} by a hierarchy of semidefinite programs (SDPs), a type of convex optimization problem for which powerful solvers exist. The SDPs yield increasingly better lower bounds on \eqref{eq:key-rate-opti}. However, since we are interested in one-sided DIQKD, we want to fix the observables that Alice measures. This prevents us from using the NPA hierarchy directly. In the following, we will consider two modifications of the NPA hierarchy that will allow us to compute bounds on \eqref{eq:key-rate-opti}.

\subsection{Finite-size key rate}
At this point we also note that, using the Entropy Accumulation Theorem~\cite{Dupuis_2020, Dupuis_2019, Metger_2022}, the optimization problem in \eqref{eq:key-rate-opti} can be used to calculate finite-size key rates as we will sketch now. As shown in Appendix~\ref{sec:finite-size-proof}, the GEAT \cite{Metger_2022} can be used to show that, under Condition~\ref{cond:non_signalling}, any affine lower-bound $\textsc{ca}$ on $H(S|I Q_E)$ gives a bound
\begin{equation} \label{eq:key-length}
    l \geq n\left[ \min_{p \in \Omega_{\mathrm{acc}}} \textsc{ca}(p) - H(S|IB) \right] - O(\sqrt{n})
\end{equation}
on the key-length $l$. The set $\Omega_{\mathrm{acc}}$ contains all probability distributions for which the parameter estimation step of the protocol does not abort. Hence, one is left with the problem of obtaining a good affine lower-bound $\textsc{ca}$. Since $\textsc{ca}$ is affine, we may parametrize it as $\textsc{ca}(p) = c + \lambda^T p$ for some arbitrarily chosen offset $c$ and gradient $\lambda$. The condition that\footnote{The set $\Sigma(p)$ here are all the states which are compatible with the parameter-estimation probability $p$. For details, see Appendix~\ref{sec:finite-size-proof}.} $\textsc{ca}(p) \leq \inf_{\nu \in \Sigma(p)} H(S|I Q_E)_{\nu}$ can then be recast as the optimization problem
\begin{equation} \label{eq:min-tradeoff-opti}
    c \leq \inf_{\nu} (H(S|I Q_E)_\nu - \lambda^T \nu_C)
\end{equation}
which is essentially equivalent to \eqref{eq:key-rate-opti}. For any choice of $\lambda$, the optimization problem \eqref{eq:min-tradeoff-opti} then provides a valid min-tradeoff function $\textsc{ca}$. Since this works for all $\lambda$, one can then use heuristic optimization procedures to find the gradient $\lambda$ (and hence the min-tradeoff function $\textsc{ca}$) which maximizes the key length in \eqref{eq:key-length}.

\subsection{NPA hierarchy with matrix algebra constraints}\label{subsec:npa_hierarchy_with_matrix_constraints}

A powerful perspective on the optimization problem \eqref{eq:key-rate-opti} for the purpose of optimizing conditional von Neumann entropy is the commuting operator framework \cite{ChengDIQKD}. The commuting operator framework yields a clean description of what is meant by \emph{optimizing over all quantum systems at Bob and Eve's disposal}. Indeed, taking the discussion from \cite[Sec.~5]{ChengDIQKD} into account yields that Bob can be just described as the universal $C^\ast$-algebra of POVMs and, since Alice is assumed to be trusted in here, she has just a full matrix algebra at her disposal. Thus, Alice's algebra is defined by elementary matrices
\begin{align}
    \{\Theta_{i,j} \ \vert \ 1\leq i,j\leq m \}, 
\end{align}
with the rule 
\begin{align}\label{eq:rule_elementary_matrices}
    \Theta_{i,j} \Theta_{k,l} = \delta_{j,k} \Theta_{i,l}.
\end{align}
Thus, the $\Theta_{i,j}$ should be thought of as matrix units, i.e., matrices with $1$ in the $(i,j)$ entry and zeroes everywhere else. In contrast, Bob's algebra is generated by symbols, also called \emph{generators}, $\mathcal{G}\coloneqq \{N(b\vert y)\}_{b \in \mathcal B, \ y \in \mathcal Y}$ and relations
\begin{equation}\label{eq:rules_universalCstar_algebra}
    \begin{aligned}
         \mathcal{R}_1 \coloneqq &\{N(b\vert y)^2  = N(b\vert y)  \ \vert \ b \in \mathcal B, \  y \in \mathcal Y\}  \\
         \mathcal{R}_2  \coloneqq &\{\sum_{b \in \mathcal B} N(b\vert y)  = \mathds{1}  \ \vert \ \  y \in \mathcal Y\} \\
         \mathcal{R}_3 \coloneqq  &\{N(b\vert y)^*  = N(b\vert y)  \ \vert \ b \in \mathcal  B, \  y \in \mathcal Y\}
    \end{aligned}
\end{equation}
with $\mathcal{R} \coloneqq \mathcal{R}_1 \cup \mathcal{R}_2 \cup  \mathcal{R}_3$. Then the universal $C^\ast$-algebra generated by $\mathcal{G}$ and relations $\mathcal{R}$ is denoted by $\mathfrak{B} \coloneqq C^*\!\left(\mathcal{G}\vert \mathcal{R}\right)$ \cite[Def.~II.8.3.1]{Blackadar2006}. In other words, Bob's algebra is indeed generated by projective measurements of arbitrary dimension. Since Alice's algebra is nuclear, the $C^*$-tensor product of Alice's and Bob's algebra is unambiguous and can be identified with (cf.~\cite[Exam.~II.9.4.2]{Blackadar2006})
\begin{align}\label{eq:whole-algebra}
    \mathfrak{A} \otimes \mathfrak{B} \cong \mathbb{M}_m\!\left(\mathfrak{B}\right), 
\end{align}
i.e, $m \times m$ matrices with coefficients in the algebra $\mathfrak B$. As shown in \cite{ChengDIQKD}, a concrete realization of an experiment yielding statistics from a state $\rho \in \mathcal{S}\!\left( \mathfrak{A} \otimes \mathfrak{B}\right)$ can be considered without loss of generality with the GNS-representation \cite[Sec.~II.6.4]{Blackadar2006} as purification. For the purpose of optimization theory, we need to find a unified representation in terms of generators and relations of the whole algebra in  \eqref{eq:whole-algebra} describing the Alice-Bob system, which is sufficient. Even though there may be more efficient ways to describe the whole algebra in  \eqref{eq:whole-algebra}, we prefer the generic way resulting from \eqref{eq:rule_elementary_matrices} and \eqref{eq:rules_universalCstar_algebra} in the following and we refer to Appendix \ref{sec:more_efficient_presentation} for details on the efficiency in terms of the amount of generators. A crucial ingredient for the later efficiency of the NPA hierarchy will be the dilation argument that Bob's system can be chosen to be the universal $C^\ast$-algebra of projective measurements instead of general effects. The argument is done in \cite[Cor.~5.6]{ChengDIQKD}, which we can apply here. Thus, we define
\begin{align}
    \eta_{i,j,b\vert y} \coloneqq \Theta_{i,j} \otimes N(b\vert y),
    \quad 1\leq i,j\leq m,\ b\in \mathcal B,\ y \in \mathcal Y,
\end{align}
as elementary generators with the following relations for the situation that $\{N(b\vert y)\}_{b\in \mathcal B}$ are projective measurements for each $y \in \mathcal Y$ on Bob's side: 
\begin{enumerate}
    \item[(R1)] \emph{$*$-structure:}
    \begin{align}
        \eta_{i,j,b\vert y}^\ast = \eta_{j,i,b\vert y}
        \quad\text{for all }i,j,b,y.
    \end{align}
    \item[(R2)] \emph{Matrix-unit relations for a fixed input-output $(b,y)$:}
    \begin{align}
        \eta_{i,j,b\vert y}\,\eta_{k,\ell,b\vert y}
        = \delta_{j,k}\,\eta_{i,\ell,b\vert y}
        \quad\text{for all }i,j,k,\ell, b,y.
    \end{align}
    \item[(R3)] \emph{Orthogonality of different outcomes for the same input $y$ on Bob's side:}
    \begin{align}
        \eta_{i,j,b\vert y}\,\eta_{k,\ell,b^\prime\vert y} = 0
        \quad\text{for all }i,j,k,\ell, y\text{ and }b\neq b'.
    \end{align}
    \item[(R4)] \emph{Completeness for each input $y$:}
    \begin{align}
        \sum_{i=1}^m \sum_{b\in \mathcal B} \eta_{i,i,b\vert y} = \mathds{1}
        \quad\text{for all }y\in \mathcal Y.
    \end{align}
    \item[(R5)] \emph{Independence of the input label for Alice's part:}
    For all $i,j$ and all $y,y'\in \mathcal Y$,
    \begin{align}
        \sum_{b\in \mathcal B} \eta_{i,j,b\vert y} = \sum_{b\in \mathcal B} \eta_{i,j,b\vert y'}.
    \end{align}
\end{enumerate}
By the fact that all operators $\eta_{i,j,b\vert y}$ are bounded by construction and the relations can be fulfilled within Hilbert spaces, \cite[II.8.3.1]{Blackadar2006} yields that together they generate a well-defined universal $C^\ast$-algebra. Showing that the generators $\{\eta_{i,j,b\vert y}\}_{i,j,b,y}$, together with (R1)-(R5), indeed generate  \eqref{eq:whole-algebra} requires a standard proof for a $\ast$-isomorphism, which we omit here and refer instead to Appendix \ref{sec:more_efficient_presentation}. 

In a next step we use the, by assumption, known representation of the operators $M(a|x)$ in terms of the matrix units $\Theta_{i,j}$ given by
\begin{align}
    M(a\vert x) = \sum_{i,j=1}^m \lambda_{i,j}^{(a\vert x)} \Theta_{i,j}
\end{align}
for some coefficients $\lambda_{i,j}^{(a\vert x)} \in \mathbb{C}$. Thus, $M(a\vert x) \otimes N(b\vert y)$ can be identified up to isomorphism of $C^\ast$-algebras with
\begin{align}
     \left(\sum_{i,j=1}^m \lambda_{i,j}^{(a\vert x)} \Theta_{i,j} \otimes \mathds{1}\right)
       (\mathds{1}\otimes N(b\vert y)) &= \sum_{i,j=1}^m \lambda_{i,j}^{(a\vert x)} (\Theta_{i,j}\otimes N(b\vert y)) \\
     &\cong \sum_{i,j=1}^m \lambda_{i,j}^{(a\vert x)} \,\eta_{i,j,b\vert y}.
\end{align}
Thus the optimization problem \eqref{eq:key-rate-opti} can be written as 
\begin{equation}\label{eq:Cstar_SDP_eta}
\begin{aligned}
    \mathrm{minimize}\quad \ &H(A\vert X=x_0, Q_E)_\psi \\
    \mathrm{such~that}\quad \ 
    &\sum_{a \in \mathcal A, b \in \mathcal B, x \in \mathcal X, y \in \mathcal Y} c_{abxyi}  \sum_{j,\ell=1}^m \lambda_{j,\ell}^{(a\vert x)} \,\psi(\eta_{j,\ell,b\vert y}) \geq q_i \quad \forall i \in \mathcal C \\
    &\psi \in \mathcal{S}\!\left( \mathfrak{A} \otimes \mathfrak{B}\right) \\
    &\{\eta_{j,\ell,b\vert y}\}\ \text{satisfy (R1)-(R5).}
\end{aligned}
\end{equation}
The optimization problem in \eqref{eq:Cstar_SDP_eta} is in a form which is compatible with a relaxation by the NPA hierarchy, since all relations (R1)-(R5) are relations in the monomials $\{\eta_{i,j,b\vert y}\}_{i,j,b,y}$. We are left with a relaxation of the conditional von Neumann entropy, which is a non-linear functional on the state, to a polynomial. For details on this relaxation we refer to \cite{kossmann2025reliableentropyestimationobserved,kossmann2025optimisingrelativeentropysemidefinite}. The resulting optimization problem is given in   \eqref{eq:optimization_problem_eta} in Appendix~\ref{sec:vonNeumann_relaxation} and can be relaxed using the NPA hierarchy.

While the approach using matrix units works for any dimension $m$, for $m=2$ it is alternatively possible to impose matrix algebra constraints based on the algebraic relations of the Pauli operators. See Appendix \ref{sec:more_efficient_presentation}, in particular \eqref{eq:pauli_generator_entropy_program}, for details. 

\subsection{NPA hierarchy with matrix-valued polynomials} \label{sec:matrix-NPA-methods}
Another way to tackle problem \eqref{eq:key-rate-opti} is to extend the NPA hierarchy to allow for matrix-valued polynomials, corresponding to a subsystem of fixed dimension. Concretely, in Appendix \ref{sec:matrix-NPA} we give a hierarchy of SDPs which yield a converging sequence of lower bounds to the following polynomial optimization problem 
\begin{align*}
    \mathrm{minimize} \quad & \quad \bra{\psi} p(X) \ket{\psi} \\
    \mathrm{such~that} \quad & \quad q^{(j)}(X) \geq 0 \qquad \forall j \in \mathcal Q \\
    \quad & \quad \bra{\psi}r^{(\ell)}(X) \ket{\psi} \geq 0  \qquad \forall \ell \in \mathcal R \tag{\textbf{P}} \label{eq:problem-P-main} \\
    \quad & \quad \langle \psi | \psi \rangle = 1 \\
    \quad & \quad \ket{\psi} \in \mathbb C^m \otimes \mathcal H \\
    \quad & \quad X \in \mathcal B(\mathcal H)^g \\
    \quad & \quad \mathcal H \mathrm{~Hilbert~space}
\end{align*}
Here, $\mathcal Q$, $\mathcal R$ are index sets,  $p$ and $r^{(\ell)}$ are polynomials in non-commutative variables with coefficients in $\mathcal B(\mathbb C^m)$ and the $q^{(j)}$ are polynomials in non-commutative variables with coefficients in $\mathcal B(\mathbb C^{m_j})$. There are $g$ non-commutative variables $x=(x_1, \ldots, x_g)$ and the polynomials can both involve $x_i$ and their adjoint $x_i^\ast$. The polynomials are required to be Hermitian, i.e., to be invariant if we take both the Hermitian conjugate of the coefficients and the adjoint of the variables. For example,
\begin{equation}
    p_1(x) = \begin{pmatrix}
        0 & 1 \\ 0 & 0
    \end{pmatrix} x_1^\ast x_2 x_1 + \begin{pmatrix}
        0 & 0 \\ 1 & 0
    \end{pmatrix} x_1^\ast x_2^\ast x_1 \label{eq:p1example}
\end{equation}
is Hermitian, whereas 
\begin{equation}
    p_2(x) = \begin{pmatrix}
        1 & 0 \\ 0 & 0
    \end{pmatrix} x_1^\ast x_2 x_1
\end{equation}
is not. If we evaluate the polynomial at a tuple of matrices $X \in \mathcal B(\mathcal H)^g$, we are taking tensor products between the coefficients and the matrices $X$. For example,
\begin{equation}
    p_1(X) = \begin{pmatrix}
        0 & 1 \\ 0 & 0
    \end{pmatrix} \otimes (X_1^\ast X_2 X_1) + \begin{pmatrix}
        0 & 0 \\ 1 & 0
    \end{pmatrix} \otimes (X_1^\ast X_2^\ast X_1) \,.
\end{equation}
For $m=1$ and $m_j=1$ for all $j \in \mathcal Q$, we recover the problem that the NPA hierarchy gives a sequence of converging lower bounds for.

To connect to problem \eqref{eq:key-rate-opti}, we have non-commutative variables $x=(x_{b|y})_{b \in \mathcal B\setminus \{|\mathcal B|\}, y \in \mathcal Y}$, where we evaluate at $X_{b|x} = N(b|y)$. Moreover, as we prefer inequality constraints, we substitute $N(|\mathcal B| | y) = \mathds{1}- \sum_{b \in \mathcal B \setminus \{|\mathcal B|\}} N(b|y)$ for all $y \in \mathcal Y$ in order to take care of the constraint that the $ N(b|y)$ have to sum to the identity. Here, we assumed without loss of generality the index set to have the form $\mathcal B = \{1, \ldots, |\mathcal B|\}$. For $b \in \mathcal B \setminus \{|\mathcal B|\}$ we obtain 
\begin{equation}
q^{(b|y)}(x) = x_{b|x}    
\end{equation}
with $m_{b|y} = 1$. Moreover,
\begin{equation}
q^{(|\mathcal B||y)}(x) = 1- \sum_{b \in |\mathcal B|} x_{b|x} 
\end{equation}
with $m_{|\mathcal B||y} = 1$. We could have also written the equality constraint $\sum_{b} N(b|y) =\mathds{1}$ as two inequalities, but this method is less efficient and typically leads to numerical instabilities when using the solver. Moreover, we set
\begin{equation}
    r^{(i)}(x) = \sum_{a \in \mathcal A, b \in \mathcal B \setminus \{|\mathcal B|\}, x \in \mathcal X, y \in \mathcal Y} (c_{abxyi}-c_{a|\mathcal B|xyi}) M(a|x) x_{b|y} + \sum_{a \in \mathcal A, x \in \mathcal X, y \in \mathcal Y}  c_{a|\mathcal B|xyi} M(a|x)- q_i \mathds{1} \,.
\end{equation}
As objective function $p$ we choose a polynomial approximation of $H(A|X=x_0, Q_E)_\rho$ using the methods in \cite{brown2024device, kossmann2025reliableentropyestimationobserved,kossmann2025optimisingrelativeentropysemidefinite} as described above.

The SDPs that need to be solved in order to compute an increasing sequence of lower bounds to \eqref{eq:problem-P-main} which converges to the optimal solution look as follows: For $k\geq k_0\coloneq\lceil \max_{j \in \mathcal Q, \ell \in \mathcal R}\{\operatorname{deg}(p), \operatorname{deg}(q^{(j)}), \operatorname{deg}(r^{(\ell)})\}/2\rceil$, they are
\begin{align*}
        \mathrm{minimize} \quad & \quad  \sum_{|w| \leq 2k} \Tr[p_w y_w^T] \\
    \mathrm{such~that} \quad & \quad M_k(y) \geq 0 \tag{$\mathbf{R_k}$} \label{eq:problem-Rk-main}  \\
    \quad & \quad M_{k-d_j}(q^{(j)}y) \geq 0  \qquad \forall j \in \mathcal Q \\
    \quad & \quad \sum_{|w| \leq 2k} \Tr[r^{(\ell)}_w y_w^T] \geq 0 \qquad \forall \ell \in \mathcal R\\ 
    \quad & \quad \sum_{i=1}^m y_{\emptyset}^{ii} = 1 \\
    \quad & \quad y_w =(y_w^{ij})_{i, j \in [m]} \in \mathcal B(\mathbb C^m) \qquad \forall |w| \leq 2k \,.
\end{align*}
Here, $d_j = \lceil \deg(q^{(j)})/2 \rceil$ and we have written the matrix-valued polynomials in terms of the words $w$ appearing as
\begin{equation}
    p = \sum_w p_w w \,,
\end{equation}
where $p_w \in \mathcal B(\mathbb C^m)$ is the the coefficient belonging to word $w$. For example, for $p_1$ as in  \eqref{eq:p1example}, there are two words $w_1\coloneq x_1^\ast x_2x_1$ and $w_2\coloneq x_1^\ast x_2x_1$ with
\begin{equation}
    p_{w_1} = \begin{pmatrix}
        0 & 1 \\ 0 & 0
    \end{pmatrix}, \qquad  p_{w_2} = \begin{pmatrix}
        0 & 0 \\ 1 & 0
    \end{pmatrix}.
\end{equation}
By $|w|$ we mean the length of the word. In our example, $|w_1|=|w_2|=3$. To each word, we associate an $m \times m$-matrix $y_w = (y_w^{ij})_{i, j \in [m]}$. These matrices are the variables we are optimizing over. The \emph{moment matrix} $M_k(y)$ is a block matrix consisting of the different $y_w$, where the entries are indexed by words up to length $k$ as $M_k(y)(v,w) = y_{v^\ast w}$, where $|v|$, $|w| \leq k$. Here, the $\ast$-operation acts on $v=v_1 \ldots v_k$ by reversing the letters and replacing the variables by their adjoints, i.e., $v^\ast = v_k^\ast \ldots v_1^\ast$. Writing $\ket{\psi} = \sum_{i \in [m]} \ket{i} \otimes \ket{\psi_i}$ for the state in \eqref{eq:problem-P-main}, the $y_w^{ij}$ should be thought of as being $\bra{\psi_i} w(X) \ket{\psi_j}$, where $w(X)$ is the word $w$ evaluated at $X$, i.e., where we replaced $x_i$ and $x_i^\ast$ by $X_i$ and $X_i^\ast$.

The localization matrix $M_{k-d_j}(q^{(j)}y)$ has a similar structure, but with blocks of size $m_j m  \times m_j m $ and entries
\begin{equation}
    M_{k-d_j}(q^{(j)}y)(v,w) = \sum_{|s|\leq \deg(q^{(j)})} q^{(j)}_s \otimes y_{v^\ast s w} 
\end{equation}
for $|v|$, $|w| \leq k-d_j$.

For $m=1$ and $m_j=1$ for all $j \in \mathcal Q$, the \eqref{eq:problem-Rk-main} coincide with the SDPs from the NPA hierarchy as expected. The main difference between \eqref{eq:problem-Rk-main} and the NPA hierarchy is that we replaced the entries $y_w$ of the moment matrix, which are scalars in the NPA hierarchy, by $m \times m$ matrices.

For formal definitions of the objects appearing in this section and a proof that the solutions $p^{(k)}$ to \eqref{eq:problem-Rk-main} form an increasing sequence of lower bounds to the optimal solution $p^\opt$ of \eqref{eq:problem-P-main} with $\lim_{k \to \infty} p^{(k)} = p^\opt$ we refer the reader to Appendix \ref{sec:matrix-NPA}.

\subsection{Comparison of the methods} \label{sec:comparison}
In this section, we presented two methods to solve the optimization problem~\eqref{eq:key-rate-opti}, the NPA hierarchy with matrix algebra constraints (NPA-AC) when using matrix units (and NPA-AC(Pauli) when using Pauli operators as in Appendix \ref{sec:more_efficient_presentation}) in Section \ref{subsec:npa_hierarchy_with_matrix_constraints} and the NPA hierarchy with matrix-valued polynomials (NPA-MP) in Section \ref{sec:matrix-NPA-methods}. A natural question to ask is how these two approaches are related, which is what we focus on in this section.

 In abstract terms, the NPA hierarchy constructs for a sequence of elements $\{ \gamma_1, ..., \gamma_n \} \subseteq \mathbb{M}_n\!\left(\mathfrak{B}\right)$ a linear, continuous and completely positive map (cf.~\cite[Lem.~5]{kossmann2023hierarchiessemidefiniteoptimizationmathcalcstaralgebras})
    \begin{equation}
        \begin{aligned}
            \Phi: \mathbb{C}^{n\times n} ~&\longrightarrow~ \mathbb{M}_n\!\left(\mathfrak{B}\right) \\
            M = ( m_{ij} ) ~&\longmapsto~ \sum_{i,j = 1}^n m_{ij} \gamma_i \gamma_j^*.
        \end{aligned}
    \end{equation}
From this map one can construct the NPA relaxation \cite{Navascus2007_npa1,Navascus2008_npa2} by an identification of the kernel of $\Phi$ as shown in \cite[Thm.~7]{kossmann2023hierarchiessemidefiniteoptimizationmathcalcstaralgebras}. The crucial observation for the following discussion is that \cite[Thm.~7]{kossmann2023hierarchiessemidefiniteoptimizationmathcalcstaralgebras} does not take into account that $\Phi$ is completely positive instead of just positive. The usual NPA-relaxation (and hence also the NPA-AC hierarchy) does not use this fact as a constraint on the resulting SDP. The NPA-MP hierarchy, however, takes these complete positivity constraints into account, which becomes apparent in the case of one-sided DIQKD, since a natural $2$-positivity is part of the problem. However, since the map $\Phi$ is completely positive, we could enforce $n$-positivity for all $n\in \mathbb{N}$ as constraints in the NPA-AC hierarchy as well. This would, in general, yield stronger bounds.     

The main appeal of the approach using the NPA-AC hierarchy is that it can be run on available implementations for the NPA hierarchy, which have been optimized over the last two decades, whereas for the NPA-MP hierarchy we had to write our own implementation. It turns out that in our examples using the NPA-AC hierarchy was slightly faster, whereas the NPA-MP hierarchy  needed less memory, see Table \ref{tab:runtime_comparison}. While we found the  NPA-MP hierarchy and the NPA-AC hierarchy to perform similarly for the provlems in this article, one or the other might be better suited for other problems. We leave the exploration of this possibility for future work.

\begin{table}[t!]
    \centering
    \begin{tabular}{c|c|c|c|c}
        Objective & BFF & KS & KS & KS \\
        and Method & NPA-MP & NPA-MP & NPA-AC(Pauli) & NPA-AC \\ \hline
        Runtime (s) & 2.8 & 1.6 & 31 & 1.3 \\
        Memory (MiB) & 294 & 201 & 645 & 245 \\
    \end{tabular}
    \caption{
        Runtime comparison between the different methods for the one-sided qubit BB84 protocol. All simulations were run on an Intel Ultra 9 285H CPU running Julia version 1.12.6 on Linux.
    }
    \label{tab:runtime_comparison}
\end{table}

Regarding the different relaxations of the objective function, for the small problems we consider in this article, the performance is very similar. For larger problems however, the KS approximation in \cite{kossmann2025reliableentropyestimationobserved,kossmann2025optimisingrelativeentropysemidefinite} is in general much faster than the BFF approximation in  \cite{brown2024device}. However, we observed that the KS relaxation was sometimes yielding slightly less tight bounds, but this could be due to a suboptimal choice of interpolation points. Here we should also mention that for both types of objective functions (BFF and KS) we use the ``block'' hierarchy from \cite[Appendix B]{Lanore_2026}. This improved runtime (and potentially accuracy) for all methods but the impact was especially drastic for the the method based on the NPA-AC hierarchy (last column in Table~\ref{tab:runtime_comparison}), where it improved runtime by roughly three orders of magnitude.
Finally, it may be surprising that the NPA-AC(Pauli) method performs worse than the NPA-AC method with matrix units. The explanation for this is that for the BB84 protocol considered in Table~\ref{tab:runtime_comparison} it suffices to consider real-valued optimization variables, whereas the Pauli commutation constraints require complex numbers (see Appendix~\ref{sec:more_efficient_presentation}). We expect that due to symmetries in the protocol it may be possible to speed up the NPA-AC(Pauli) method but leave this for future work.
In summary, the picture we see here for the different relaxations of the objective function is similar to what has been observed in the context of calculating key rates in fully DIQKD \cite{brown2024device, kossmann2025reliableentropyestimationobserved,kossmann2025optimisingrelativeentropysemidefinite}. 

\section{Numerical results}

In this section, we apply our methods to two protocols of interest. For each of the protocols, we will describe the honest measurements for both parties and the statistics collected during the protocol. Note that while we provide the honest measurements for both Alice and Bob, in our analysis we will only fix the measurements for one of the two parties while the other one remains untrusted. Like most QKD protocols, our protocols distinguish between key rounds and test rounds. That is, for each round, Alice and Bob independently decide whether to designate it as a key or as a test round. A given round is then considered a key round if both parties chose it to be a key round, and a test round if both parties chose it to be a test round. For key rounds, both parties perform measurements in a fixed basis, whereas for test rounds they randomly choose among a set of predefined bases. Only the measurement results of the key rounds are used to generate the final raw keys $S$ and $\hat{S}$.  Asymptotically, the fraction of test rounds goes to zero and hence almost all rounds are used for key generation. This means that in (\ref{eq:asymptotic_key_rate}), it suffices to consider the entropy of the measurement in the key basis for states which satisfy the constraints in the test bases.

We illustrate the achievable key rate in a one-sided device-independent setting for a few relevant protocols. For brevity, we focus here on two specific protocols: The first one is the well-known (one-sided) DIQKD protocol based on the  Clauser-Horne-Shimony-Holt (CHSH) inequality. The second protocol is the entanglement-based BB84 protocol with losses. Several other protocols can be found in Appendix~\ref{sec:protocols}.
More concretely, we consider protocols based on the $I_{3322}$ inequality, MUBs, and the standard qubit BB84 protocol. For all protocols we consider both the scenario where only Alice's measurements are characterized and the case where only Bob's measurements are characterized. 

For all the numerical calculations that follow, we use a Gauss-Radau quadrature with $m=8$ points in the BFF approximation and $r=22$ grid points in the KS approximation.

\subsection{CHSH protocol}

The CHSH protocol \cite{clauser1969proposed} is based on two binary-outcome measurements per party. We denote by $x \in  \{0, 1\}$ Alice's measurement setting and by $a \in \{0, 1\}$ her measurement result. Similarly, we write $y \in \{0, 1, 2\}$ for Bob's measurement setting and $b \in \{0, 1\}$ for his measurement result.

Alice randomly chooses between measuring the observables
\begin{equation}
    A_0 = \sigma_z
    \qquad \text{and} \qquad
    A_1 = \sigma_x,
\end{equation}
where $\sigma_x$ and $\sigma_z$ denote the Pauli matrices (acting on a qubit space).
This can be described by the POVMs $\{M_{Q_A}(a|x)\}_{a\in\mathcal{A},x\in\mathcal
X}$ given by
\begin{equation}
\begin{aligned} \label{eq:chsh_meas_alice}
    M_{Q_A}(0|0)  &= \ketbra{0}, &
    M_{Q_A}(1|0)&= \ketbra{1}, \\
    M_{Q_A}(0|1) &= \ketbra{+}, &
    M_{Q_A}(1|1) &= \ketbra{-}.
\end{aligned}
\end{equation}
Similarly, Bob chooses between the three binary-outcome observables
\begin{equation}
    B_0 = \frac{1}{\sqrt{2}}(\sigma_x + \sigma_z),
    \qquad
    B_1 = \frac{1}{\sqrt{2}}(\sigma_x - \sigma_z),
    \qquad \text{and} \qquad
    B_2 = \sigma_z,
\end{equation}
which leads to the POVMs $\{N_{Q_B}(b|y)\}_{b\in\mathcal{B},y\in\mathcal
Y}$ given by
\begin{equation}
\begin{aligned}
    N_{Q_B}(0|0) &= \frac{1}{2}(\mathds{1} + B_0), & N_{Q_B}(1|0) &= \frac{1}{2}(\mathds{1} - B_0), \\
    N_{Q_B}(0|1) &= \frac{1}{2}(\mathds{1} + B_1), & N_{Q_B}(1|1) &= \frac{1}{2}(\mathds{1} - B_1), \\
    N_{Q_B}(0|2) &= \ketbra{0}, & N_{Q_B}(1|2) &= \ketbra{1}.
\end{aligned}
\end{equation}
Any attack strategy, given by a state $\rho_{Q_AQ_BQ_E}$ and measurements $M_{Q_A}(a|x)$ and $N_{Q_B}(b|y)$, gives rise to a conditional probability distribution
\begin{equation}
    p(ab|xy) = \Tr[\rho_{Q_AQ_B} (M_{Q_A}(a|x) \otimes N_{Q_B}(b|y))].
\end{equation}
We consider a protocol where the inputs $x = 0$ and $y = 2$ are used for the key rounds, and a CHSH game with randomly chosen inputs $x \in \{0, 1\}$ and $y \in \{0, 1\}$ is played in the test rounds. The winning probability of the CHSH game is then computed as

\begin{figure}[t!]
    \centering
    \includegraphics[width=1.0\linewidth]{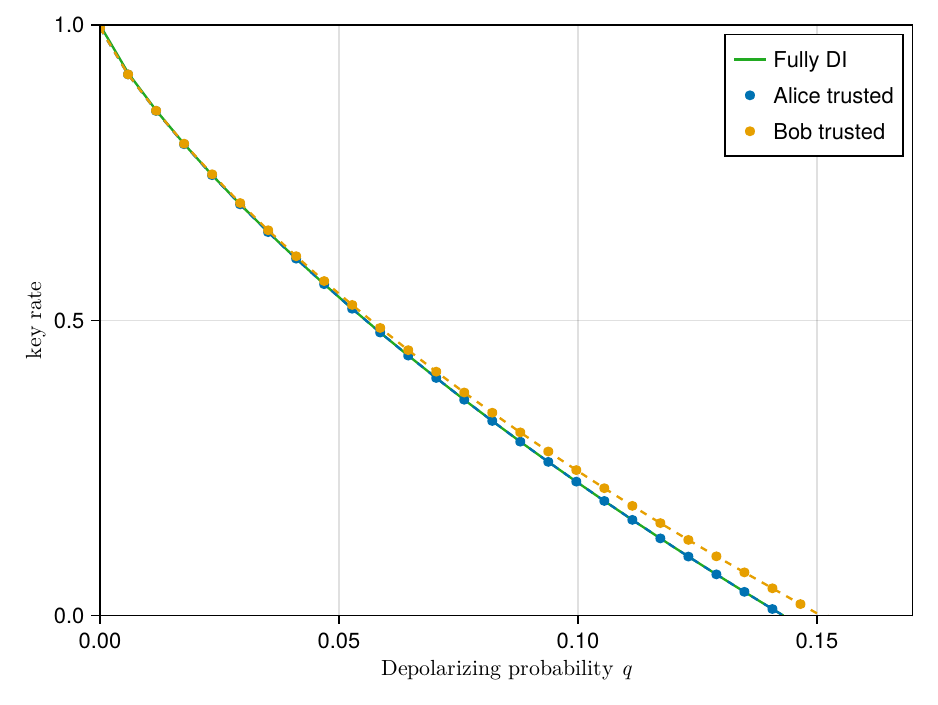}
    \caption{ \label{fig:fig1}
        Plot of the asymptotic key rate as a function of the depolarizing probability $q$. Shown are the cases when Bob is trusted (and Alice is not), the case when Alice is trusted (and Bob is not), and when neither party is trusted (fully device-independent). In all cases, the key is extracted from Alice's measurement results.
    }
    \label{fig:CHSH_key_rates}
\end{figure}

\begin{equation}
    \omega_{\mathrm{CHSH}}
    = \frac{1}{4}
    \sum_{\substack{a,b,x,y \in \{0, 1\} \\ a\oplus b = xy}}
    p(ab|xy),
\end{equation}
which corresponds to the probability that Alice's and Bob's outcomes satisfy the CHSH winning condition $a \oplus b = xy$. Classical correlations satisfy $\omega_{\mathrm{CHSH}}\leq 3/4$, while quantum correlations can achieve values up to $\omega_{\mathrm{CHSH}}=\cos^2(\pi/8)\approx 0.8536$.

To study the impact of experimental imperfections, we consider distributions $p(ab|xy)$ arising from the honest measurements given above and the joint state
\begin{equation}
    \rho_{Q_AQ_B} = (1-q) \ketbra{\psi}_{Q_AQ_B} + q \frac{\mathds{1}_{Q_AQ_B}}{4},
\end{equation}
where $\ket{\psi} = (\ket{00} + \ket{11})/\sqrt{2}$, and $q \in [0, 1]$ denotes the depolarizing probability. A plot showing the asymptotic key rate as a function of the depolarizing probability $q$ is shown in Figure~\ref{fig:CHSH_key_rates}.

We observe that the key rate is identical between the fully device-independent and the one-sided device-independent setting when Alice is trusted. However, when Bob's measurements are characterized, we achieve better key rates than in the fully device-independent setting. This is somewhat counter-intuitive since one would expect that characterizing the key-generating party leads to better key rates. This behavior can at least partially be explained by looking at the optimal attack on the device-independent protocol presented in \cite{Acin_2007}. This attack performs the honest measurements (i.e., the measurements in (\ref{eq:chsh_meas_alice})) on Alice's side but a different measurement on Bob's side.\footnote{In \cite{Acin_2007} they extract the key from Bob's side. Since we extract the key from Alice, we need to swap the parties in the attack of~\cite{Acin_2007}.} This means that the attack remains feasible even when Alice's measurements are trusted. Hence, the one-sided device-independent setting with trusted Alice cannot achieve a better key rate than the fully device-independent setting. However, since Bob's measurements in the attack from~\cite{Acin_2007} differ from the trusted measurement, this argument does not apply to Bob.

\subsection{BB84 with losses}
We consider a photonic implementation of an entanglement-based BB84 protocol with polarization encoding and active basis choice. In each round, Alice and Bob randomly choose between measuring in the rectilinear or diagonal basis, which we label with $0$ and $1$, respectively. After squashing \cite{Beaudry_2008,Gittsovich_2014}, the honest measurements act on the single-photon subspace $\mathcal{H} = \mathrm{span}\{\ket{0}, \ket{1}, \ket{\varnothing} \}$ and are given by 
\begin{equation}
\begin{aligned} \label{eq:lossy_bb84_mmts}
    M_{Q_A}(0|0) &= \ketbra{0}, \quad& M_{Q_A}(1|0) &= \ketbra{1}, \quad& M_{Q_A}(\varnothing| 0) &= \ketbra{\varnothing}, \\
    M_{Q_A}(0|1) &= \ketbra{+}, \quad& M_{Q_A}(1 | 1) &= \ketbra{-}, \quad& M_{Q_A}(\varnothing| 1) &= \ketbra{\varnothing}, \\
\end{aligned}
\end{equation}
where $\ket{+} = \tfrac{1}{\sqrt{2}}(\ket{0} + \ket{1})$ and $\ket{-} = \tfrac{1}{\sqrt{2}}(\ket{0} - \ket{1})$ are diagonally and anti-diagonally polarized photons. Rounds where both Alice and Bob measure in the rectilinear basis are used to generate the raw key and rounds where both parties measure in the diagonal basis are used for testing.

Just like in the previous protocol, we can study the situation when either Alice or Bob's measurement device is trusted while the other device remains uncharacterized. To achieve a non-zero key rate in the presence of high losses, we would like to post-select on the rounds where a detection occurs, i.e., the rounds with outcome not equal to $\varnothing$. However, since one of the two measurement devices is untrusted, we cannot post-select on its measurement results without violating the non-signalling condition in Condition~\ref{cond:non_signalling}: publicly announcing the no-detection rounds would create a signalling channel from the untrusted device's internal memory to the public transcript. Hence, we only post-select on the trusted device's measurement results and assign non-detection outcomes of the untrusted device to a fixed outcome, say $0$. Since we cannot perform post-selection on the untrusted device's outcomes, we place the source close to the untrusted party to minimize losses.

Let us illustrate how to rigorously deal with this post-selection in our security proof. 
For ease of presentation, we consider the setting where Alice is trusted, i.e., she performs the measurements $\{M_{Q_A}(a|x)\}_{a \in \mathcal A,x \in \mathcal X}$ given in \eqref{eq:lossy_bb84_mmts} and Bob performs some uncharacterized measurements $\{N_{Q_B}(b|y)\}_{b\in \mathcal B,y\in \mathcal Y}$. 
Consider a key generation round, i.e., a round where both Alice and Bob measure in the rectilinear basis. Then, their post-measurement state is given by
\begin{equation}
    \rho_{AB Q_E} = \sum_{\substack{a \in \{0, 1, \varnothing\}, \\ b \in \{0, 1\}}} \ketbra{a, b}_{AB} \otimes \Tr_{Q_AQ_B}[M_{Q_A}(a|0) \otimes N_{Q_B}(b|0) \rho_{Q_AQ_BQ_E}].
\end{equation}
After performing her measurement, Alice announces whether or not she received the outcome $A = \varnothing$, i.e., she sends $I = 0$ if $A = \varnothing$ and $I = 1$ otherwise.\footnote{In the complete protocol, $I$ also includes the basis choices and whether or not the round is chosen as a test round. Since we focus on a key round here, these announcements are already fixed.}
Alice sets her raw key to $S = A$.
According to (\ref{eq:asymptotic_key_rate}), the asymptotic key rate is given by
\begin{equation}
    r_\infty = \inf_\rho ( H(S|IQ_E)_\rho - H(S|IB)_\rho ).
\end{equation}
Focusing on the first term (the second term only depends on the expected statistics and is therefore easy to compute), we have that
\begin{equation}
\begin{aligned}
    H(S|IQ_E)_\rho
    &= H(S|Q_E,I=0)_\rho \mathrm{Pr_\rho}[I = 0] + H(S|Q_E,I=1)_\rho \mathrm{Pr_\rho}[I = 1] \\
    &= H(S|Q_E,I=1)_\rho \mathrm{Pr_\rho}[I = 1],
\end{aligned}
\end{equation}
where $\Pr_\rho[\cdot]$ denotes the probability evaluated with respect to the state $\rho$, and where we used that if $I = 0$ then $S = \varnothing$ deterministically and hence $H(S|Q_E,I=0)_\rho = 0$. Note that in the expression above we have the post-selected entropy $H(S|Q_E,I=1)_\rho$ instead of simply $H(S|Q_E)_\rho$. That is, to correctly handle post-selection, it is in general not sufficient to compensate for losses by simply adding a prefactor of $\mathrm{Pr_\rho}[I=1]$ in front of the key-rate formula.
This is a similar issue to the one recently highlighted in \cite{Lobo2026} where a gap in the security proof of \cite{Branciard2012} was identified.
The expression above can be further simplified as\footnote{There is some ambiguity in how to define the entropy of sub-normalized states. Here, we use the definition $H(A|B)_\rho = -D(\rho_{AB}, \mathds{1}_{A} \otimes \rho_{B})$ with $D(\rho, \sigma) = \Tr[\rho(\log \rho - \log \sigma)]$.}
\begin{equation}
    H(S|Q_E,I=1)_\rho \mathrm{Pr_\rho}[I = 1] = H(S|Q_E)_{\rho_{\land I = 1}},
\end{equation}
where we introduced the sub-normalized state $\rho_{SBQ_E \land I = 1}$ given by
\begin{equation}
    \rho_{SB Q_E \land I = 1} = \sum_{s, b \in \{0, 1\}} \ketbra{s, b}_{SB} \otimes \Tr_{Q_AQ_B}[M_{Q_A}(s|0) \otimes N_{Q_B}(b|0) \rho_{Q_AQ_BQ_E}].
\end{equation}
Note that in the expression above, we do not include the outcome $S=\varnothing$ since we are post-selecting on $I=1$. 
The entropy $H(S|Q_E)_{\rho_{\land I = 1}}$ can then be evaluated using the numerical techniques presented in Section~\ref{sec:methods}.

\begin{figure}[t!]
    \centering
    \includegraphics[width=\linewidth]{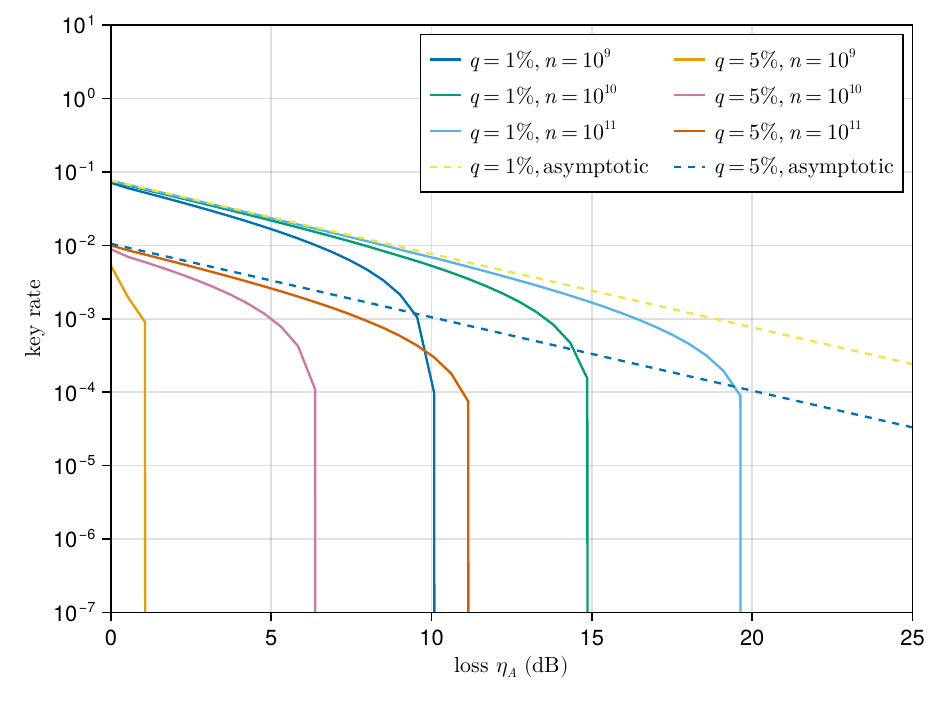}
    \caption{
        Key rates for the one-sided device-independent BB84 protocol with trusted Alice as a function of Alice's detection efficiency $\eta_A$ for different values of the depolarizing probability $q$ and the number of rounds $n$. Bob's detection efficiency is fixed at $\eta_B = 0.8$ and he assigns non-detection events to the outcome $b=0$. 
        The testing probability is fixed at $\gamma = 0.1$. The soundness parameter is set to $\varepsilon_{\mathrm{snd}} = 10^{-12}$ and the completeness parameter is set to $\varepsilon_{\mathrm{comp}}=0.01$.
    }
    \label{fig:lossy_bb84_finite_size}
\end{figure}

The behavior of the untrusted devices is restricted to producing the same statistics as the honest implementation. For this, we impose upper-bounds on the bit-error probability in the diagonal basis as well as an upper-bound on the detection probability in the diagonal basis. Formally, we impose 
\begin{align} 
    \begin{split}
    \Pr[\mathrm{err}|x=1] &= \Tr[\rho_{Q_AQ_B}M_{Q_A}(0|1) \otimes N_{Q_B}(1|1)]+ \Tr[\rho_{Q_AQ_B} M_{Q_A}(1|1) \otimes N_{Q_B}(0|1)] \\
    &\leq \mathrm{Pr}_{\mathrm{hon}}[\mathrm{err}|x=1], 
    \end{split}\\
    \begin{split}
    \Pr[\varnothing|x=1] &= \Tr[\rho_{Q_AQ_B} M_{Q_A}(\varnothing|1) \otimes \mathds{1}_{Q_B}]\\
    &\leq \mathrm{Pr}_\mathrm{hon}[\varnothing|x=1],
    \end{split}
\end{align}
where $\mathrm{Pr}_{\mathrm{hon}}[\mathrm{err}|x=1]$ and $\mathrm{Pr}_\mathrm{hon}[\varnothing|x=1]$ are the corresponding quantities when evaluated on the honest states and measurements. For evaluating the performance of the protocol, we consider the honest measurements from (\ref{eq:lossy_bb84_mmts}) and the honest state
\begin{equation}
    \rho_{Q_AQ_B}^\mathrm{hon} = \Big(\cE_{Q_A|Q_A'}^{(\eta_A)} \otimes \cE^{(\eta_B)}_{Q_B|Q_B'}\Big)\left[(1 - q) \ketbra{\psi}_{Q_A' Q_B'} + q \frac{\mathds{1}_{Q_A' Q_B'}}{4}\right],
\end{equation}
where $\ket{\psi} = (\ket{00} + \ket{11})/\sqrt{2}$, $q \in [0, 1]$ is the depolarizing probability, and the loss channels $\cE^{(\eta_A)}, \cE^{(\eta_B)}$ are given by
\begin{equation}
    \cE^{(\eta)} = \eta\,\mathrm{id} + (1 - \eta)\ketbra{\varnothing} \otimes \Tr,
\end{equation}
where $\mathrm{id}$ is the identity map.

The resulting key rates (for trusted Alice) are shown in Figure~\ref{fig:lossy_bb84_finite_size}. We see that, just like in the device-dependent setting, the asymptotic key rate scales as $O(\eta_A)$ with $\eta_A$ denoting Alice's detection probability. This behaviour was already observed previously in~\cite{Branciard2012,Masini_2024}. However, no explicit finite-size key rates against general attacks were computed in the prior works.

\section{Discussion and outlook}

In this paper, we have considered the problem of calculating key rates for one-sided DIQKD protocols. The main technical contributions that enabled us to solve this problem were two ways to incorporate finite dimensionality into the NPA hierarchy: first the NPA-AC hierarchy using a $C^*$-algebra approach and adding matrix algebra constraints, second the NPA-MP hierarchy generalizing the NPA hierarchy to allow for matrix-valued polynomials. While we only used these new SDP hierarchies in the context of one-sided DIQKD, they are applicable in other areas of quantum information theory as well, for example in quantum steering. 

The possibility of an extension of the NPA hierarchy to matrix-valued polynomials as in the NPA-MP hierarchy has been mentioned in a comment \cite{pironio2010convergent}, pointing to the fact that the Positivstellensatz in \cite{helton2004positivstellensatz} also holds for matrix-valued polynomials. Special cases of the NPA-MP hierarchy we give in the present paper have been proposed in \cite{johnston2016extended} and later in \cite{escola2025lossy}. Both works were motivated by applications to extended non-local games; \cite{escola2025lossy} extending the results in \cite{johnston2016extended} to lossy-and-constrained extended non-local games. Both works prove that their respective hierarchies converge. While the proofs do not extend straightforwardly to our setting, we were in particular inspired by the treatment in \cite{johnston2016extended}. More recently, the block matrix structure of the moment matrix was also used in \cite{dalessandro2026semidefinite}, but without proving convergence. The exact relation to, e.g., \cite{johnston2016extended}, is not made explicit in \cite{dalessandro2026semidefinite}. Other works on similar questions include \cite{navascues2014characterization}, which studies correlations for dimensionally-constrained systems, resulting in a provably converging modified version of the NPA hierarchy that, however, does not cover the problem we consider in the present paper. Moreover, \cite{navascues2015characterizing} included dimension constraints in the NPA hierarchy using matrix polynomial identities, but the results again do not cover the type of problem we are interested in here. Finally, works such as \cite{dalessandro2025semidefinite} have studied polynomial optimization methods for quantum steering scenarios, but without giving convergence guarantees \cite{kogias}.

Equipped with these new tools, we computed and compared key rates for different protocols for one-sided DIQKD and studied the behavior of these rates under different amounts of loss. Prior work essentially falls into two categories: On the one hand, many works were tailored towards specific protocols \cite{Branciard2012, Tomamichel_2013, Roy2026, Masini_2024}. On the other hand, \cite{Tan2021_key_rates} can handle general algebraic linear NPA-constraints and entropy optimization. Our method goes beyond that in that it can handle generic discrete-variable one-sided DIQKD protocols with matrix-valued NPA-constraints. Compared to~\cite{Mothsara_2026}, which is based on SDP relaxations of the min-entropy, our method directly computes the von Neumann entropy. In addition, to obtain an SDP, \cite{Mothsara_2026} relaxes some of the underlying tensor product structure inherent to the optimization problem, opting instead to only impose weaker non-signaling constraints. Hence, compared to \cite{Mothsara_2026}, our method obtains superior key rates and noise tolerance (see the MUB protocol in Appendix~\ref{sec:protocols}). 
Additionally, we provide a generic security proof against general adversaries (see Appendix~\ref{sec:finite-size-proof}). While some prior works also provided a security proof against coherent attacks, this was either restricted to specific protocols \cite{Tomamichel_2013} (using entropic uncertainty relations), or only provided asymptotic key rates \cite{Masini_2024} (using entropy accumulation).
We also note that some prior works studied one-sided continuous-variable protocols \cite{Walk_2016, Gehring_2015}, which we do not cover in this paper.

Using our new method to calculate key rates for different protocols, we observed a phenomenon which has already been observed in prior literature: Depending on whether the party that extracts the key is trusted or not, the key rate differs \cite{Walk_2016, Mothsara_2026}.
Intuitively, one would expect that the key rate is higher if the key generating party is trusted, which is consistent with the results in \cite{Walk_2016, Mothsara_2026}. However, while this turns out to be true in many of the examples we considered, in the one-sided CHSH protocol the behavior is reversed. To our knowledge, this is the first time that extracting the secret key from the untrusted side has been shown to improve performance. Understanding when and why this behavior arises merits future work.

Another open question is which numerical method is best suited for which type of problem: In our numerical calculations, we observed that the NPA-AC hierarchy is slightly faster than the NPA-MP hierarchy (as shown in Table~\ref{tab:runtime_comparison}), whereas the latter uses less memory. It would be interesting to do a more thorough comparison of the two methods.

Several conceptual questions for the operator-algebraic model underlying device-in\-de\-pen\-dent security remain. In fully device-independent scenarios, assuming a tensor-product decomposition between the parties is mathematically convenient, but does not describe the most general commuting-operator model. The distinction is nontrivial as highlighted by Tsirelson's problem \cite{Tsirelson1993,Kirchberg1993}. Commuting-operator models can in general admit correlations that cannot be reproduced within the usual tensor-product framework \cite{ji2021mip}. More broadly, phenomena associated with unbounded or infinite-dimensional entanglement, including embezzlement-type behaviour, show that finite-dimensional intuition can be misleading \cite{van2024embezzlement, van2024schmidt}. Whether such distinctions can actually be exploited by an adversary in QKD, and whether they can affect achievable key rates, remains largely open. Here, the one-sided setting has an interesting additional asymmetry. A characterized system described by a matrix algebra  induces a unique tensor product with the algebra of the uncharacterized device. Consequently, characterizing Alice or Bob is not merely a relabeling once the raw key and reconciliation direction are fixed. In a more general commuting-operator description, it determines whether the uncharacterized subsystems associated with key generation or with error correction is the one whose correlations with the adversary are least constrained. It would be interesting to determine whether this algebraic asymmetry can partly explain the differences between the two trust assignments observed in our numerical results. 

A closely related question addresses the assumptions required for finite-size security. Our present analysis uses the GEAT in a setting where the relevant systems admit an effective finite-dimensional description. It would be valuable to understand precisely what this assumption means algebraically and which parts of the proof survive when fewer restrictions are imposed on the correlation structure or on the devices' internal memories. Beyond these foundational questions, an experimental comparison of the different trust assignments would be particularly informative. Different assumptions remove different engineering requirements, and it is not clear whether the asymmetry found at the level of idealized key rates persists once losses, detector imperfections, calibration requirements, and implementation complexity are taken into account. Finally, one may consider protocols in which Alice and Bob themselves assign different trust models to the devices. This raises a rather operational question: under which conditions can the protocol produce a key that both parties regard as secure, and when can security be certified only from the viewpoint of one of them? Such asymmetric notions of security may provide a useful framework for understanding QKD networks in which the parties have genuinely different levels of control over, or confidence in, the underlying hardware. 

In summary, our work provides a general toolbox for calculating key rates for one-sided DIQKD protocols. This provides the starting point for investigating new protocols in one-sided DIQKD.

\section*{Code availability}
A stand-alone Julia implementation of the NPA-MP hierarchy (see Appendix~\ref{sec:matrix-NPA} for details) is available in \url{https://github.com/MartinSandfuchs/NPA-MP}.
Code to compute key-rates is available in \url{https://github.com/MartinSandfuchs/semi-DI}.
An alternative Python implementation, reproducing a subset of the CHSH and BB84 results in this work using the NPA-MP hierarchy and the BFF entropy bound, is available at \url{https://github.com/GiuseppeViola/npa-mp-bff-diqkd}.

\section*{Acknowledgements}
AB would like to thank Igor Klep for discussions regarding Positivstellensätze. AB was supported by the French National Research Agency in the framework of the “France 2030” program (ANR-11-LABX-0025-01) for the LabEx PERSYVAL and by the ANR project PraQPV, grant number ANR-24-CE47-3023. GK acknowledges support from the Excellence Cluster - Matter and Light for Quantum Computing (ML4Q-2) and by the European Research Council (ERC Grant Agreement No. 948139). MS acknowledges funding from the Swiss National Science Foundation via project No.~20CH21\_218782, the National Centre of Competence in Research \emph{SwissMAP}, and the ETH Zurich Quantum Center. RS is supported by the DFG under Germany’s Excellence Strategy
- EXC-2123 QuantumFrontiers-2 - 390837967 and SFB
1227 (DQ-mat), the Quantum Valley Lower Saxony, and
the German Federal Ministry of Research, Technology and Space (BMFTR) via the projects ATIQ, SEQUIN, Quanda, Quics and CBQD. GV acknowledges support from the Deutsche Forschungsgemeinschaft
(DFG, German Research Foundation, project number 563437167), the Sino-German Center for Research Promotion (Project M-0294), and the German Federal Ministry of Research, Technology and Space (Project QuKuK, Grant No. 16KIS1618K and Project BeRyQC, Grant No. 13N17292). RW acknowledges financial support from the Ministry of Culture and Science of North Rhine-Westphalia via the NRW-Rückkehrprogramm. This research was funded in part by the Austrian Science Fund (FWF) 10.55776/PIN1152626. AB and RW would like to thank the SwissMAP Research Station for its hospitality during their visit in the summer of 2023 as part of the Short Research Stays program, where this project was conceived. Moreover, AB and RW would like to thank the Banff International Research Station for its hospitality during the Research in Teams program (25rit933), where part of this work was carried out. 

\bibliographystyle{halpha}
\bibliography{lit}

\newpage
\appendix

\section{Notation}
Here, we will give an overview of the notation used in this article.

 \renewcommand{\arraystretch}{1.2}
 \begin{longtable}[H]{|c||p{12cm}|}
         \hline
         \textbf{Notation} & \textbf{Description} \\
         \hline
         \hline
         $[n]$ & The set $\{1,\dots,n\}$ \\
         \hline
         $\mathbb{P}_\cC$ & Set of probability distributions over the alphabet $\cC$ \\
         \hline
         $\mathcal H, \mathcal{H}_A, \mathcal{H}_B$, ... & Hilbert spaces (belonging to different systems, $A$, $B$, ...) \\
         \hline
         $|A|$ & Dimension of $\mathcal{H}_A$  \\
         \hline
          $\mathcal{B}(\mathcal{H}_A, \mathcal{H}_B)$ & Bounded linear operators  mapping from $\mathcal{H_A}$ to $\mathcal{H_B}$. If $\mathcal H_A = \mathcal H_B = \mathcal H$, we write $\mathcal B(\mathcal{H})$ \\
         \hline
         $\mathcal{B}(\mathcal{H})_\mathrm{sa}$ & Self-adjoint bounded linear operators on $\mathcal{H}$ \\
         \hline
         $S \geq T$ & $S - T \geq 0$, i.e., $S - T$ is positive semidefinite \\
         \hline
         $\mathds{1}_A$ & Identity operator on $\mathcal{H}_A$ \\
         \hline
         $\|S\|_p$ & Schatten $p$-norm of $S$, given by $\|S\|_p = \Tr[(\sqrt{S^* S})^p]^{1/p}$. $\|S\|_\infty$ is the operator norm of $S$ \\
         \hline
         $\mathcal{S}(\mathcal{H})$ & Set of normalised density operators (also referred to as quantum states of just states) on $\mathcal{H}$: $\mathcal{S}(\mathcal{H}) = \{{\rho\in \mathcal{B}(\mathcal{H})} : \Tr[\rho] = 1, \rho \geq 0\}$ \\
         \hline
         $\rho_{XA}$ & Classical-quantum state describing a random variable $X$ correlated with a quantum system $Q_A$: $\rho_{XA}=\sum_{x \in \mathcal{X}} p_X(x) \ketbra{x}_X \otimes \rho_{A|x}$, with $\rho_{A|x}$ the state of $A$ conditioned on $X=x$ \\
         \hline
         $\mathrm{Pr}[\Omega]$ & Probability of an event $\Omega\subseteq \mathcal{X}$: $\mathrm{Pr}[\Omega]=\sum_{x \in \Omega} p_X(x)$ \\
         \hline
         $\rho_{XA|\Omega}$ & State conditioned on $\Omega \subseteq \mathcal{X}$: $\rho_{XA|\Omega}=\frac{1}{\mathrm{Pr}[\Omega]} \sum_{x \in \Omega} p_X(x) \ketbra{x}_X \otimes \rho_{A|x}$ (Note that $\rho_{XA|\Omega}$ is only well-defined if $\mathrm{Pr}[\Omega] > 0$. This is unproblematic if $\rho_{XA|\Omega}$ is accompanied by a compensating factor of $\mathrm{Pr}[\Omega]$.) \\
         \hline
         $\rho_{XA \land \Omega}$ & Subnormalized conditional state given by $\rho_{XA \land \Omega} = \mathrm{Pr}[\Omega]\,\rho_{XA | \Omega}$ \\
         \hline
         $\cE_{B|A}$ & Completely positive and trace preserving (CPTP) map or channel $\mathcal{E}_{B|A}:\mathcal{B}(\mathcal{H}_A)\to\mathcal{B}(\mathcal{H}_B)$\\
         \hline
         $\mathbb C[x, x^*]$, $\mathbb C[x, x^*]_d$ & The algebra of complex-valued polynomials in non-commutative variables $x = (x_1, \ldots, x_g)$, $x^* = (x_1^*, \ldots, x_g^*)$ (of degree at most $d$)\\
         \hline
         \begin{tabular}{c} $\mathcal B(\mathcal H_A, \mathcal H_B)[x, x^*]$, \\ $\mathcal B(\mathcal H_A, \mathcal H_B)[x, x^*]_d$ \end{tabular} & The set of polynomials with coefficients in $\mathcal B(\mathcal H_A, \mathcal H_B)$ and in non-commutative variables $x$, $x^*$ (of degree at most $d$) \\
         \hline
         $\mathcal{W}_d$ & Set of words $w$ of length $|w|\le d$ \\
         \hline
         $\mathcal{W}_\infty$ & Set of all words (of any length)\\
         \hline
         $\mathfrak{A}\otimes_{\min}\mathfrak{B}$ & Minimal tensor product of two $C^*$-algebras $\mathfrak{A}$ and $\mathfrak{B}$\\
         \hline
         $\mathbb M_n(\mathfrak{A})$ & Algebra of $n\times n$ matrices with entries in the  $C^*$-algebra $\mathfrak{A}$ \\
         \hline
         \begin{tabular}{c} $H_{\min}(A|B)_\rho$, \\ $H_{\max}(A|B)_\rho$ \end{tabular} &  Conditional min/max-entropy \\ \hline
         \begin{tabular}{c} $H_{\min}^\varepsilon(A|B)_\rho$, \\ $H_{\max}^\varepsilon(A|B)_\rho$ \end{tabular} & The smooth conditional min/max-entropy in terms of purified distance \\ \hline
     \caption{
         \textbf{Summary of notation.}
     }
     \label{tab:notation}
 \end{longtable}

\section{Additional protocols}

\label{sec:protocols}
Here we present some additional protocols that, for readability, were not included in the main text. We will compare three different protocols: The lossless qubit BB84 protocol, the MUB protocol, and a protocol based on the $I_{3322}$ Bell-inequality. For each protocol, we will discuss the measurements performed by the honest parties. It is understood that for numerical evaluations we will only characterize one of the two parties while leaving the measurements of the other party unspecified. We defer plots of key rates until the end of this section (Figure~\ref{fig:plot_all}) where a comparison between all the protocols in this section (and the CHSH protocol) is shown.

\subsection{Lossless qubit BB84}
In the BB84 protocol, the honest parties randomly choose between measuring the $X$ and the $Z$ basis. This can be described by the POVM
\begin{equation}
\begin{aligned}
    M_{Q_A}(0|0) &= \ketbra{0}, & M_{Q_A}(1|0) &= \ketbra{1}, \\
    M_{Q_A}(0|1) &= \ketbra{+}, & M_{Q_A}(1|1) &= \ketbra{-},
\end{aligned}
\end{equation}
and similarly for Bob. During parameter estimation, the honest parties estimate the error rates in the $X$ and $Z$ bases, which are given by
\begin{equation}
\begin{aligned}
    \mathrm{err}_X &\coloneqq \Tr[(M_{Q_A}(0|0) \otimes N_{Q_B}(1|0) \rho_{Q_AQ_B}]+ \Tr[M_{Q_A}(1|0) \otimes N_{Q_B}(0|0) \rho_{Q_AQ_B}], \\
    \mathrm{err}_Z &\coloneqq \Tr[(M_{Q_A}(0|1) \otimes N_{Q_B}(1|1) \rho_{Q_AQ_B}]+ \Tr[M_{Q_A}(1|1) \otimes N_{Q_B}(0|1) \rho_{Q_AQ_B}].
\end{aligned}
\end{equation}

\subsection{MUB protocol}
A set of of orthonormal bases $\cB_1, \ldots \cB_k$ of $\mathbb{C}^d$ is called mutually unbiased if
\begin{equation}
    \abs{\langle u | v \rangle}^2 = \frac{1}{d} \quad \forall u \in \cB_i, v \in \cB_j, i \neq j.
\end{equation}
The most well-known example of MUBs are the eigenbases of the Pauli operators $\sigma_x, \sigma_y, \sigma_z$ in dimension $d=2$. More generally, it is known that $d + 1$ MUBs exist whenever $d$ is a prime power \cite{Bengtsson_2007}, while the case of general $d$ remains an open problem.

Given a set of MUBs, a QKD protocol can be built by Alice and Bob randomly choosing one of the bases for their measurements. For $d=2$, this leads to the well-known six-state protocol \cite{Bruss1998,BPG1999}. Here, we consider the case of $d=3$ with the MUBs
\begin{equation}
\begin{aligned}
    \cB_0 =& \left\{ \ket{0}, \ket{1}, \ket{2} \right\}, \\
    \cB_1 =& \left\{ \tfrac{1}{\sqrt{3}}(\ket{0} + \ket{1} + \ket{2}), \tfrac{1}{\sqrt{3}}(\ket{0} + \xi^2\ket{1} + \xi\ket{2}), \tfrac{1}{\sqrt{3}}(\ket{0} + \xi \ket{1} + \xi^2 \ket{2}) \right\}, \\
    \cB_2 =& \left\{ \tfrac{1}{\sqrt{3}}(\ket{0} + \ket{1} + \xi \ket{2}), \tfrac{1}{\sqrt{3}}(\ket{0} + \xi^2\ket{1} + \xi^2\ket{2}), \tfrac{1}{\sqrt{3}}(\ket{0} + \xi \ket{1} + \ket{2}) \right\}, \\
    \cB_3 =& \left\{ \tfrac{1}{\sqrt{3}}(\ket{0} + \ket{1} + \xi^2 \ket{2}), \tfrac{1}{\sqrt{3}}(\ket{0} + \xi^2\ket{1} + \ket{2}), \tfrac{1}{\sqrt{3}}(\ket{0} + \xi \ket{1} + \xi \ket{2}) \right\}, \\
\end{aligned}
\end{equation}
where $\xi = e^{2\pi i / 3}$. This means that we get a QKD protocol with four measurement settings $x, y \in \{0, 1, 2, 3\}$ and three measurement outcomes $a, b \in \{0, 1, 2\}$.

For the honest implementation, we assume that Alice and Bob share a bipartite state
\begin{equation}
    \rho_{Q_AQ_B} = (1 - q)\ketbra{\psi}_{Q_AQ_B} + q \frac{1}{d^2} \mathds{1}_{Q_AQ_B},
\end{equation}
where $\ket{\psi}_{Q_A Q_B} \coloneqq \sum_{i=1}^{d} \tfrac{1}{\sqrt{d}} \ket{i i}_{Q_A Q_B}$ is the maximally entangled state in $d$ dimensions and $q \in [0, 1]$ is the depolarizing probability. For parameter estimation, we characterize the full distribution $p(ab|xy)$. We also note that the MUB protocol has been previously studied in~\cite{Mothsara_2026}. However, in~\cite{Mothsara_2026} the authors lower-bound the von Neumann entropy $H(A|IQ_E)$ in \eqref{eq:asymptotic_key_rate} by the min-entropy $H_{\min}(A|IQ_E)$ (which can be expressed as an SDP), leading to overly pessimistic key rates as illustrated in Figure~\ref{fig:plot_MUB}.

\begin{figure}[t!]
    \centering
    \includegraphics[width=\linewidth]{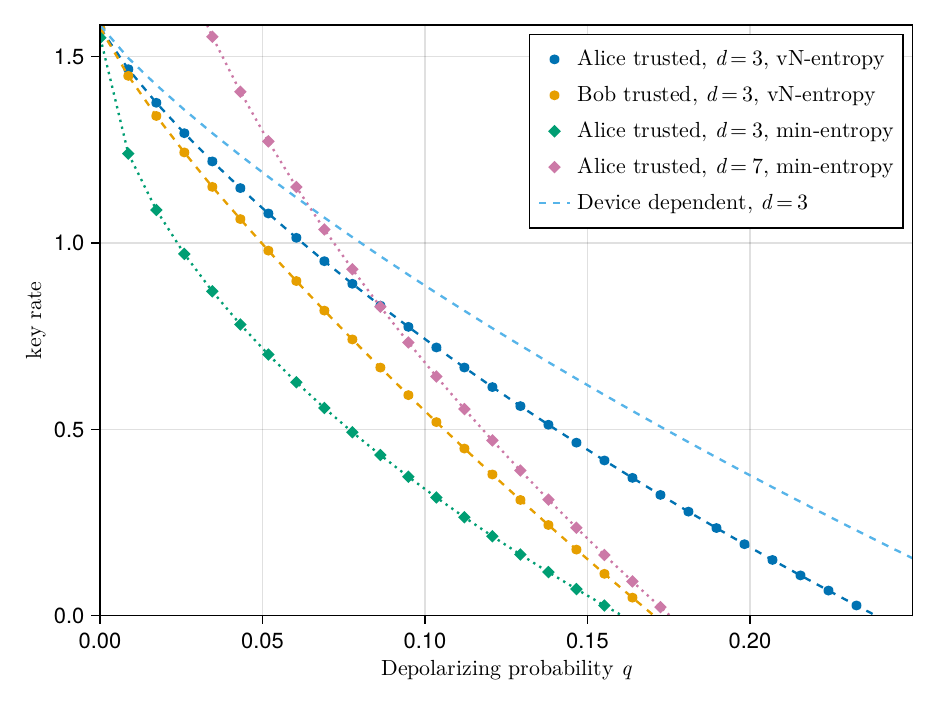}
    \caption{
        Key rates for the one-sided DIQKD protocol based on $d + 1$ mutually unbiased bases in $d$-dimensional Hilbert spaces. We compare bounds derived via the von Neumann entropy with bounds based on the min-entropy. The former are calculated using our method based on the KS objective function and the NPA-MP hierarchy, whereas for the latter we used the SDP relaxation of the min-entropy from \cite{Mothsara_2026}. We observe that even in dimension $d=3$, the von Neumann entropy provides better noise tolerance than the min-entropy at $d=7$.
    }
    \label{fig:plot_MUB}
\end{figure}

\subsection{$I_{3322}$ protocol}

The $I_{3322}$ protocol \cite{froissart1981constructive, sliwa2003symmetries, collins2004relevant, pal2010maximal} is based on three binary-outcome measurements per party. Alice randomly chooses between measuring the observables
\begin{equation}
    A_i = \cos(\alpha_i) \sigma_z + \sin(\alpha_i) \sigma_x,
\end{equation}
with the angles $\alpha_0 = 0$, $\alpha_1 = \pi/3$, and $\alpha_2 = 2\pi/3$. This can be described by the POVM
\begin{equation}
\begin{aligned}
    M_{Q_A}(0|i) = \frac{1}{2}(\mathds{1}_{Q_A} + A_i), \quad M_{Q_A}(1|i) = \frac{1}{2}(\mathds{1}_{Q_A} - A_i).
\end{aligned}
\end{equation} 
Similarly, a characterized Bob measures the observables
\begin{equation}
    B_i = \cos(\beta_i) \sigma_z + \sin(\beta_i) \sigma_x,
\end{equation}
with the measurement angles $\beta_0 = 4\pi/3$, $\beta_1 = \pi$, and $\beta_2 = 2\pi/3$.
 This can be described by the POVM
\begin{equation}
\begin{aligned}
    N_{Q_B}(0|i) = \frac{1}{2}(\mathds{1}_{Q_B} + B_i), \quad N_{Q_B}(1|i) = \frac{1}{2}(\mathds{1}_{Q_B} - B_i).
\end{aligned}
\end{equation}

Again, the measurement devices of Alice and Bob are characterized by the conditional joint probability distribution $p(ab|xy)$, where $x,y\in\{0,1,2\}$ denote the measurement settings chosen by Alice and Bob, respectively, and $a,b\in\{0,1\}$ are the corresponding outcomes. Experimentally, this distribution is estimated from the observed frequencies of outcomes for each pair of inputs.

The security analysis of the protocol is based on the violation of the
$I_{3322}$ Bell inequality. In terms of the expectation values of the
measurement observables, the corresponding Bell expression is

\begin{equation}
\begin{aligned}
    I_{3322} &\coloneqq
    -\langle A_1 \rangle - \langle B_0 \rangle - 2\langle B_1 \rangle + \langle A_0 B_0 \rangle + \langle A_0 B_1 \rangle + \langle A_1 B_0 \rangle \\
    &\hspace{10pt}+ \langle A_1 B_1 \rangle - \langle A_0 B_2 \rangle  + \langle A_1 B_2 \rangle - \langle A_2 B_0 \rangle + \langle A_2 B_1 \rangle \\
    &\leq \, 0,
\end{aligned}
\end{equation}
where the inequality holds for any local hidden-variable model. Quantum theory
allows for a violation of this bound; the maximum value achievable with
a pair of qubits is $0.25$, while the true quantum maximum, attained
only in the limit of infinite-dimensional Hilbert spaces, is
$\approx 0.250\,875\,38$ \cite{pal2010maximal}.

\begin{figure}[ht!]
    \centering
    \includegraphics[width=\linewidth]{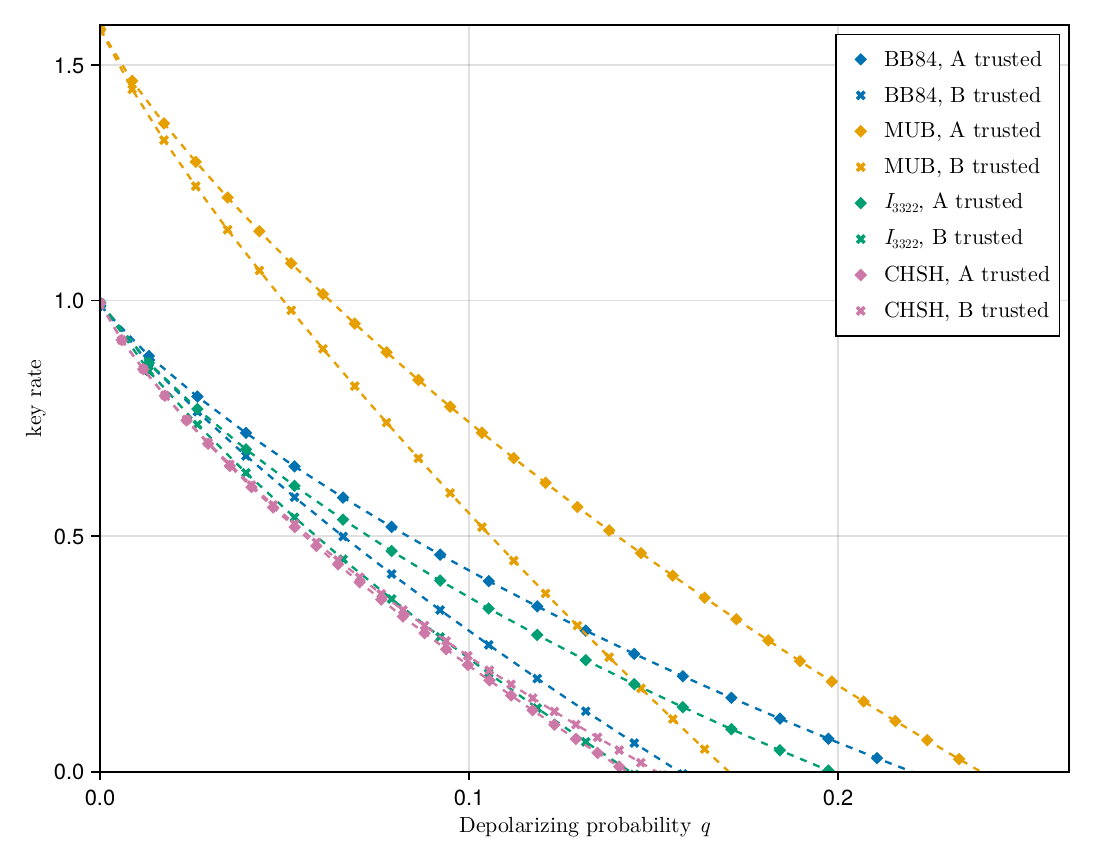}
    \caption{
        The asymptotic key rate of the various one-sided DIQKD protocols as a function of the depolarizing probability $q$. We remark that for BB84 with Alice trusted we obtain the same key rate as the device-independent case.
    }
    \label{fig:plot_all}
\end{figure}

\section{Technical background for finite-size security proofs}
\label{sec:technical_background}

In this appendix, we collect some technical preliminaries for the finite-size security proof in Appendix \ref{sec:finite-size-proof}.

\subsection{Security of QKD}

\begin{defi}[Completeness]
\label{def:completeness}
    The protocol is called $\varepsilon_\mathrm{comp}$-complete if there exists an honest implementation of the protocol such that
        \begin{equation}
            \mathrm{Pr}[\mathrm{abort}]\le\varepsilon_\mathrm{comp}.
        \end{equation}
\end{defi}

\begin{defi}[Soundness] \label{def:soundness}
    The protocol is called $\varepsilon_\mathrm{snd}$-sound if
        \begin{equation}
            \label{eq:soundness}
            \frac{1}{2} \| \rho_{K_A^l K_B^l \Efin \land \Omega} - \tau_{K_A^l K_B^l}\otimes \rho_{\Efin \land \Omega}\|_1 \leq \varepsilon_\mathrm{snd},
        \end{equation}
    where $\tau_{K_A^l K_B^l}=\frac{1}{2^l}\sum_{k\in\{0,1\}^l}\ketbra{kk}_{K_A^l K_B^l}$.
\end{defi}

\subsection{Entropic quantities}
Let $\rho_A$ be a density operator acting on the finite-dimensional Hilbert space
$\mathcal{H}_A$. The \emph{von Neumann entropy} \cite{nielsen2010quantum} of $\rho_A$ is defined as
\begin{equation}
H(A)_\rho
\coloneqq
-\Tr\!\left[\rho_{A} \log_2 \rho_{A}\right].
\end{equation}
Here the logarithm is taken to base $2$, so entropy is measured in bits.

Let $\rho_{AB}$ be a bipartite density operator on
$\mathcal{H}_A\otimes\mathcal{H}_B$, and let $\rho_{B}=\Tr_{A}[\rho_{AB}]$
denote its reduced density operator on subsystem $\mathcal{H}_B$. The
\emph{conditional von Neumann entropy} \cite{nielsen2010quantum} of $A$ conditioned on $B$ is
\begin{equation}
H(A|B)_\rho
\coloneqq
H(AB)_\rho-H(B)_\rho .
\end{equation}

For the finite-size security proof, we will also require the smooth min- and max-entropy~\cite{Renner05,Tomamichel2016}, which are defined as follows:
\begin{equation}
\begin{aligned}
    H_{\min}(A|B)_{\rho} \coloneqq& \max_{\lambda, \sigma} \{ \lambda : \rho_{AB} \leq 2^{-\lambda} \mathds{1}_{A} \otimes \sigma_{B} \}, \\
    H_{\max}(A|B)_{\rho} \coloneqq& \max_{\sigma} \log F(\rho_{AB}, \mathds{1}_{A} \otimes \sigma_{B}), \\
\end{aligned}
\end{equation}
where the maximization is over sub-normalized states $\sigma_B$ and $F(\rho, \sigma) = \left(\Tr\left[\abs{\sqrt{\rho} \sqrt{\sigma}}\right]\right)^2$ denotes the fidelity. Let $\varepsilon \geq 0$, we define the smooth min- and max-entropy by
\begin{equation}
\begin{aligned}
    H_{\min}^\varepsilon(A|B)_{\rho} &\coloneqq \max_{\tilde{\rho} \in \mathcal{B}_{\rho}^\varepsilon} H_{\min}(A|B)_{\tilde{\rho}}, \\
    H_{\max}^\varepsilon(A|B)_{\rho} &\coloneqq \min_{\tilde{\rho} \in \mathcal{B}_{\rho}^\varepsilon} H_{\max}(A|B)_{\tilde{\rho}},
\end{aligned}
\end{equation}
where $\mathcal{B}_{\rho}^{\varepsilon} \coloneqq \{ \tilde{\rho} : P(\rho, \tilde{\rho}) \leq \varepsilon,\, \tilde{\rho} \geq 0, \,\Tr(\tilde{\rho})\leq 1 \}$ denotes the $\varepsilon$-ball with respect to the purified distance $P(\rho, \tilde{\rho})$ (for a definition of the purified distance, see, for instance, \cite{Tomamichel2016}).

\subsection{Leftover Hashing}
Proving soundness (Definition~\ref{def:soundness}) requires showing an upper-bound on the trace-distance between the final state of the QKD protocol and the state describing the ideal key. Since during the execution of the QKD protocol the adversary may have obtained some side-information about the raw key, one needs to perform some additional classical post-processing to obtain a private key. A common approach to achieve this, called privacy amplification~\cite{Renner05}, is to apply a random hash-function from a two-universal family of hash functions (see Definition~\ref{def:two_universal} below). The Leftover Hashing Lemma (Lemma~\ref{lem:leftover_hashing}) then shows that this produces uniform and independent randomness when the raw key was sufficiently random (quantified by the smooth min-entropy).

\begin{defi}[Two-universal hash functions] \label{def:two_universal}
    We call a family of functions $\mathcal{F}$ from $\mathcal{X}$ to $\mathcal{Z}$ \emph{two-universal} if, for every $x \neq x'$ it holds that
    \begin{equation}
        \Pr_{f \in \mathcal{F}}[f(x) = f(x')] \leq \frac{1}{\abs{\mathcal{Z}}},
    \end{equation}
    where $f$ is chosen uniformly at random from $\mathcal{F}$.
\end{defi}

\begin{lem}[{Leftover Hashing Lemma \cite[Prop.~9]{Tomamichel_2017}}] \label{lem:leftover_hashing}
    Let $\sigma_{XE}$ be a sub-normalized state which is classical on $X$ and let $\varepsilon \in [0, \sqrt{\Tr(\sigma)}]$. Let $\cF$ be a two-universal family of hash functions from $\cX = \{0, 1\}^n$ to $\cK = \{0, 1\}^l$. Then
    \begin{equation}
        \frac{1}{2} \norm{\rho_{KFE} - \tau_{K} \otimes \rho_{FE}}_1 \leq \frac{1}{2} \sqrt{2^{l - H_{\min}^\varepsilon(X|E)_{\sigma}}} + 2\varepsilon,
    \end{equation}
    where $\tau_{K} = 2^{-l} \mathds{1}_{K}$ is the maximally mixed state and $\rho_{KFE} = \frac{1}{\abs{\cF}} \sum_{f \in \cF} f_{K|X}[\sigma_{XE}] \otimes \ketbra{f}_F$.
\end{lem}

\subsection{Entropy accumulation}

In this section, we collect all definitions and theorems required for the GEAT \cite{Metger_2022}. The setup is depicted in Figure~\ref{fig:GEAT_setup}.

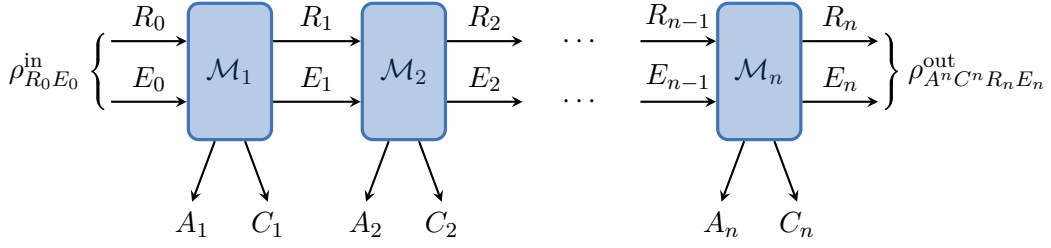
\begin{figure}[t!]
    \centering
    \begin{tikzpicture}[
        >=stealth,
    ]
        \node[blue_box, minimum height=1.8cm] (M1) at (0, 0) {$\cM_1$};
        \node[blue_box, minimum height=1.8cm] (M2) at (2.3, 0) {$\cM_2$};
        \node[blue_box, minimum height=1.8cm] (Mn) at (7, 0) {$\cM_n$};
    
        \draw[thick,->] ([xshift=-1.0cm,yshift=+0.4cm] M1.west) -- node[above]{$R_0$} ([yshift=+0.4cm] M1.west);
        \draw[thick,->] ([xshift=-1.0cm,yshift=-0.4cm] M1.west) -- node[above]{$E_0$} ([yshift=-0.4cm] M1.west);
        \draw[very thick,decorate,decoration={calligraphic brace,raise=0.1cm,amplitude=0.15cm}] ([xshift=-1.0cm,yshift=-0.5cm] M1.west) -- node[left=0.25cm]{$\rho^\mathrm{in}_{R_0 E_0}$} ([xshift=-1.0cm,yshift=+0.5cm] M1.west);
    
        \draw[thick,->] ([xshift=-0.2cm] M1.south) -- ([xshift=-0.5cm,yshift=-0.8cm] M1.south) node[below]{$A_1$};
        \draw[thick,->] ([xshift=+0.2cm] M1.south) -- ([xshift=+0.5cm,yshift=-0.8cm] M1.south) node[below]{$C_1$};
    
        \draw[thick,->] ([yshift=+0.4cm] M1.east) -- node[above]{$R_1$} ([yshift=+0.4cm] M2.west);
        \draw[thick,->] ([yshift=-0.4cm] M1.east) -- node[above]{$E_1$} ([yshift=-0.4cm] M2.west);
    
        \draw[thick,->] ([yshift=+0.4cm] M2.east) -- node[above]{$R_2$} ([xshift=+1.0cm, yshift=+0.4cm] M2.east);
        \draw[thick,->] ([yshift=-0.4cm] M2.east) -- node[above]{$E_2$} ([xshift=+1.0cm, yshift=-0.4cm] M2.east);
    
        \draw[thick,->] ([xshift=-0.2cm] M2.south) -- ([xshift=-0.5cm,yshift=-0.8cm] M2.south) node[below]{$A_2$};
        \draw[thick,->] ([xshift=+0.2cm] M2.south) -- ([xshift=+0.5cm,yshift=-0.8cm] M2.south) node[below]{$C_2$};
    
        \node at (4.65, +0.4) {$\cdots$};
        \node at (4.65, -0.4) {$\cdots$};
    
        \draw[thick,->] ([xshift=-1.0cm,yshift=+0.4cm] Mn.west) -- node[above]{$R_{n-1}$} ([yshift=+0.4cm] Mn.west);
        \draw[thick,->] ([xshift=-1.0cm,yshift=-0.4cm] Mn.west) -- node[above]{$E_{n-1}$} ([yshift=-0.4cm] Mn.west);
    
        \draw[thick,->] ([yshift=+0.4cm] Mn.east) -- node[above]{$R_n$} ([xshift=+1.0cm, yshift=+0.4cm] Mn.east);
        \draw[thick,->] ([yshift=-0.4cm] Mn.east) -- node[above]{$E_n$} ([xshift=+1.0cm, yshift=-0.4cm] Mn.east);
    
        \draw[thick,->] ([xshift=-0.2cm] Mn.south) -- ([xshift=-0.5cm,yshift=-0.8cm] Mn.south) node[below]{$A_n$};
        \draw[thick,->] ([xshift=+0.2cm] Mn.south) -- ([xshift=+0.5cm,yshift=-0.8cm] Mn.south) node[below]{$C_n$};

        \draw[very thick,decorate,decoration={calligraphic brace,raise=0.1cm,amplitude=0.15cm}] ([xshift=+1.0cm,yshift=+0.5cm] Mn.east) -- node[right=0.25cm]{$\rho^\mathrm{out}_{A^n C^n R_n E_n}$} ([xshift=+1.0cm,yshift=-0.5cm] Mn.east);
    \end{tikzpicture}
    \caption{Setup of the GEAT.}
    \label{fig:GEAT_setup}
\end{figure}

\begin{defi}[Non-signalling channel]
    Let $\cN : A B \rightarrow C D$ be a channel. We say that \emph{$\cN$ does not signal from $A$ to $D$} if there is a channel $\cR: B \rightarrow D$ such that
    \begin{equation}
        \Tr_{C} \circ\,\mathcal{N} = \cR \circ \Tr_{A}.
    \end{equation}
\end{defi}

\begin{defi}[GEAT channel]
    We call a sequence of channels $\{\cM_i\}_i$ with $\cM_i: E_{i-1}R_{i-1} \rightarrow A_{i} C_{i} E_{i} R_{i}$ a sequence of \emph{GEAT channels} if $\cM_i$ does not signal from $R_{i-1}$ to $E_{i}$ and there exist channels $\cT_i : A_i E_i \rightarrow A_i C_i E_i$ such that $\cM_i = \cT_i \circ \Tr_{C_i} \circ \cM_i$ and $\cT_i$ has the form
    \begin{equation} \label{eq:geat_projective_stats}
        \cT_{i}(\omega_{A_{i} E_{i}}) = \sum_{y,z} (\Pi_{A_i}^{(y)} \otimes \Pi_{E_i}^{(z)}) \omega_{A_{i}E_{i}} (\Pi_{A_i}^{(y)} \otimes \Pi_{E_i}^{(z)}) \otimes \ketbra{r(y,z)}_{C_i},
    \end{equation}
    where $\{\Pi_{A_i}^{(y)}\}_y$ and $\{\Pi_{E_i}^{(z)}\}_z$ are projective measurements, and $r:\cY\times\mathcal{Z}\to\cC$ is a deterministic function.
\end{defi}

\begin{defi}[Min-tradeoff function]
    Let $\{\cM_i\}_i$ be a sequence of GEAT channels. We call $f: \mathbb{P}_{\cC} \rightarrow \mathbb{R}$ a \emph{min-tradeoff function} if
    \begin{equation}
        \inf_{\omega \in \Sigma_i(p)} H(A_i|E_i\tilde{E}_{i-1})_{\cM_i[\omega]} \geq f(p) \quad \forall i, p \in \mathbb{P}_{\cC},
    \end{equation}
    with
    \begin{equation}
        \Sigma_i(p) = \{ \omega_{R_{i-1}E_{i-1}\tilde{E}_{i-1}} : \mel{c}{\Tr_{A_i E_i R_i \tilde{E}_{i-1}} \circ \cM_{i}[\omega]}{c}_{C_i} = p(c) \quad \forall c \in \cC \}.
    \end{equation}
\end{defi}

\begin{defi}[Frequency distribution]
    Let $\cC$ be a finite set and $c^n \in \cC^n$. Define the distribution $\mathrm{freq}_{c^n} \in \mathbb{P}_\cC$ by
    \begin{equation}
        \mathrm{freq}_{c^n}(c) = \frac{\abs{\{i : c_i = c\}}}{n}.
    \end{equation}
\end{defi}

\begin{defi}[Min, Max, Var]
    Let $f$ be a min-tradeoff function for $\{\cM_i\}_i$. Define
    \begin{equation}
    \begin{aligned}
        \mathrm{Max}(f) =& \max_{p \in \mathbb{P}_{\cC}} f(p), \\
        \mathrm{Min}_{\Sigma}(f) =& \min_{p : \Sigma(p) \neq \emptyset} f(p), \\
        \mathrm{Var}_{\Sigma}(f) =& \max_{p : \Sigma(p) \neq \emptyset} \sum_{c \in \cC} p(c) f(\delta_c)^2 - \left( \sum_{c} p(c) f(\delta_c) \right)^2,
    \end{aligned}
    \end{equation}
    where $\Sigma(p) = \cup_i \Sigma_i(p)$ and $\delta_c$ is the distribution with all the weight on element $c$: $\delta_c(c') = \delta_{c,c'}$. 
\end{defi}

\begin{thm}[Generalised Entropy Accumulation \cite{Metger_2022}] \label{thm:GEAT}
Let $\{\cM_i\}_i$ be a sequence of GEAT channels and $f$ be an affine min-tradeoff function. Then, for any $\varepsilon \in (0, 1)$, $\alpha \in (1, \tfrac{3}{2})$, initial state $\rho^{\mathrm{in}}_{R_0 E_0}$ and any convex set $\Omega \subseteq \cC^n$, it holds that
\begin{equation}
\begin{aligned}
    H_{\min}^\varepsilon (A^n|E_n)_{\rho^{\mathrm{out}}_{|\Omega}} &\geq \, nh - n \frac{\alpha - 1}{2 - \alpha}V^2 - \frac{g(\varepsilon) + \alpha \log(1/\Pr[\Omega])}{\alpha - 1} - n\left(\frac{\alpha - 1}{2 - \alpha}\right)^2 K(\alpha) \\
    &\eqqcolon \, n h - \kappa^{\textsc{geat}}(n, \alpha, \varepsilon, \Pr[\Omega], f),
\end{aligned}
\end{equation}
where $\rho^{\mathrm{out}}_{A^nC^nR_nE_n} = \cM_n \circ \ldots \circ \cM_1[\rho^{\mathrm{in}}]$, and
\begin{equation}
\begin{aligned}
    h &= \min_{c^n \in \Omega} f(\mathrm{freq}_{c^n}) \\
    g(\varepsilon) &= \log \frac{1}{1 - \sqrt{1 - \varepsilon^2}}, \\
    V &= \log (2 \abs{A}^2 + 1) + \sqrt{2 + \mathrm{Var}(f)}, \\
    K(\alpha) &= \frac{(2 - \alpha)^3}{6(3 - 2\alpha)^3 \ln 2} 2^{\frac{\alpha-1}{2-\alpha}(2 \log \abs{A} + \mathrm{Max}(f) - \mathrm{Min}_{\Sigma}(f))} \cdot \ln^3\left(2^{2 \log \abs{A} + \mathrm{Max}(f) - \mathrm{Min}_{\Sigma}(f)} + e^2\right).
\end{aligned}
\end{equation}
\end{thm}

\section{Finite-size security proof} \label{sec:finite-size-proof}

Based on the preliminaries given in Appendix~\ref{sec:technical_background}, we can now formulate the full finite-size security proof for a general one-sided DIQKD protocol as stated in Box~\hyperref[box:protocol]{2}. The final security statement is given in Theorem~\ref{thm:sdi_qkd_security}.

\subsection{Protocol description}\label{subsec:protocol}

The following protocol closely follows the general QKD protocol presented in \cite{Metger_2023}.
\begin{tcolorbox}[
    title=Box 2: One-sided DIQKD protocol (detailed),
    title filled=false, 
    colback=black!5!white,
    colbacktitle=black!5!white,
    colframe=black!40!white,
    coltitle=black,
    enhanced,
    breakable,
]
\label{box:protocol}
\textbf{Protocol Arguments} \\
{\setlength{\tabcolsep}{2pt}
\begin{tabular}{rcp{8.5cm}}
    $n \in \mathbb{N}$ & : & Number of protocol rounds \\
    $l \in \mathbb{N}$ & : & Length of the final secret key \\
    $\Xi^{(E)} \subseteq \mathrm{CPTP}(Q_E, Q_E Q_A Q_B)$ & : & Set of attacks available to Eve. Here $E$ denotes Eve's side-information and $Q_A$ and $Q_B$ are quantum systems distributed to Alice and Bob \\
    $\Xi^{(A)} \subseteq \mathrm{CPTP}(Q_A R_A, A R_A)$ & : & Set of possible instruments for Alice with $A$ being a classical register with alphabet $\cA$. $R_A$ is Alice's memory register. \\
    $\Xi^{(B)} \subseteq \mathrm{CPTP}(Q_B R_B, B R_B)$ & : & Set of possible instruments for Bob with $B$ being a classical register with alphabet $\cB$. $R_B$ is Bob's memory register. \\
    $\textsc{pd} : \cA \times \cB \rightarrow \cI$ & : & Function describing public discussion \\
    $\textsc{rk} : \cA \times \cI \rightarrow \cS$ & : & Function describing Alice's raw key generation \\
    $\textsc{ev} : \cS \times \cB \times \cI \rightarrow \cC$ & : & Function ``evaluating'' each round by assigning a label from the alphabet $\cC$ \\
    $\Omega_{\textsc{pe}} \subseteq \cC^n$ & : & Set describing accepted statistics. \\
    $\lambda_{\textsc{ec}} \in \mathbb{N}$ & : & Number of communicated bits during error correction \\
    $\varepsilon_{\textsc{kv}} \in (0, 1)$ & : & Tolerated error for key verification \\
    $\varepsilon_{\textsc{pa}} \in (0, 1)$ & : & Tolerated error for privacy amplification \\
    $\varepsilon_{\mathrm{snd}} \in (0, 1)$ & : & Desired security parameter \\
\end{tabular}
}

\textbf{Protocol steps}
\begin{enumerate}[itemsep=-0.1em]
    \item \textit{Data generation.} For each $i \in [n]$:
    \begin{enumerate}[itemsep=-0.1em]
        \item Eve applies an attack channel $\cE \in \Xi^{(E)}$ to produce quantum systems $Q_A$ and $Q_B$ which are distributed to Alice and Bob respectively.
        \item Alice and Bob apply their instruments $\cM \in \Xi^{(A)}$ and $\cN \in \Xi^{(B)}$ to the state that they received plus their own personal memory registers $R_A$ and $R_B$. They record their outcomes in the registers $A_i$ and $B_i$.
        \item Alice and Bob perform classical communication to exchange the registers $I_i = \textsc{pd}(A_i, B_i)$.
        The register $I_i$ becomes available to Eve.
        \item Alice sets her raw key bit to $S_i = \textsc{rk}(A_i, I_i)$.
    \end{enumerate}
    
    \item \textit{Error correction.} Alice and Bob communicate at most $\lambda_{\mathrm{EC}}$ bits to perform error correction. At the end, Bob ends up with a guess for Alice's raw key in the registers $\hat{S}^n$.
    
    \item \textit{Key validation.} Alice picks a two-universal hash function $\textsc{Hash}$ hashing to $\lceil \log \frac{1}{\varepsilon_{\textsc{KV}}} \rceil$ bits. Alice sends her choice of hash function and $\textsc{Hash}(S^n)$ to Bob. Bob computes $\textsc{Hash}(\hat{S}^n)$. If their hashes disagree, they abort the protocol.  A transcript of the error correction and key validation steps is stored in the register $O$.
    
    \item \textit{Parameter estimation.} Bob computes $\hat{C}_i = \textsc{ev}(\hat{S}_i, B_i, I_i)$. If $\mathrm{freq}_{\hat{C}^n} \notin \Omega_{\textsc{pe}}$, then Alice and Bob abort the protocol.
    \item \textit{Privacy amplification.} Alice and Bob apply a random function $F$ from a two-universal set of hash functions to $S^n$ and $\hat{S}^n$ to obtain their final secret keys $K_A^l = F(S^n)$ and $K_B^l = F(\hat{S}^n)$.
\end{enumerate}
\end{tcolorbox}

A single round of the quantum phase of the protocol is shown in Figure~\ref{fig:GEAT_singleround}.

Let $K_A^l$ and $K_B^l$ be the respective keys for Alice and Bob, and let $(\rho_{|\Omega})_{K_A^lK_B^lE}$ be the quantum state that describes the final keys and Eve's knowledge (both quantum and classical) at the end of the protocol, conditioned on the event $\Omega$ that the protocol did not abort. Proving security of the protocol encompasses showing two properties, namely \emph{completeness} (Definition \ref{def:completeness}) and \emph{soundness} (Definition \ref{def:soundness}).

\subsection{Completeness}
To show that the protocol is complete, we have to bound the probability of aborting for an honest implementation, i.e., in a setting where the adversary is passive. To achieve this, we first observe that the protocol can abort in two places, in the key validation step and in the parameter estimation step. Let us bound both cases separately.

The key validation step can only fail if error correction was not successful. This probability can be made small by making $\lambda_\textsc{ec}$ sufficiently large. How large depends on the exact noise model of the honest implementation and the performance of the error-correcting code. The following bound was shown in~\cite{Metger_2023} for the error correction procedure from~\cite{Renes_2012}.
\begin{lem}[{\cite[Supplementary Lemma 1]{Metger_2023}}] \label{lem:comp_ec}
    Let $\varepsilon_\textsc{kv}^\mathrm{comp} \in (0, 1)$ and choose $\lambda_\textsc{ec}$ such that
    \begin{equation}
        \lambda_\textsc{ec} \geq n H(S|IB)_{p^\mathrm{hon}} + 2\sqrt{n}\sqrt{1 + 2 \log \frac{2}{\varepsilon_\textsc{kv}^\mathrm{comp}}} \log(1 + 2\abs{\cS}) + 2 \log \frac{2}{\varepsilon_\textsc{kv}^\mathrm{comp}},
    \end{equation}
    where $p^\mathrm{hon}$ is the distribution describing the honest distribution of $S$, $B$, and $I$.
    Then there exists an error correcting procedure which, for the honest implementation, produces an incorrect key (i.e. $\hat{S} \neq S$) with probability at most $\varepsilon_\textsc{kv}^\mathrm{comp}$.
\end{lem}

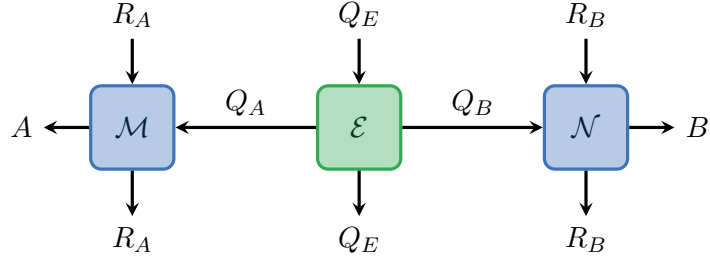
\begin{figure}[t!]
    \centering
    \begin{tikzpicture}
        \node[green_box] (XiE) at (0, 0) {$\cE$};
        \node[blue_box] (XiA) at (-3, 0) {$\cM$};
        \node[blue_box] (XiB) at (+3, 0) {$\cN$};

        \draw[->,>=stealth, very thick] (XiE.west) -- node[above]{$Q_A$} (XiA.east);
        \draw[->,>=stealth, very thick] (XiE.east) -- node[above]{$Q_B$} (XiB.west);

        \draw[->,>=stealth,very thick] ([yshift=0.6cm] XiA.north) node[above]{$R_A$} -- (XiA.north);
        \draw[->,>=stealth,very thick] (XiA.south) -- ([yshift=-0.6cm] XiA.south) node[below]{$R_A$};
        
        \draw[->,>=stealth,very thick] ([yshift=0.6cm] XiB.north) node[above]{$R_B$} -- (XiB.north);
        \draw[->,>=stealth,very thick] (XiB.south) -- ([yshift=-0.6cm] XiB.south) node[below]{$R_B$};

        \draw[->,>=stealth,very thick] ([yshift=0.6cm] XiE.north) node[above]{$Q_E$} -- (XiE.north);
        \draw[->,>=stealth,very thick] (XiE.south) -- ([yshift=-0.6cm] XiE.south) node[below]{$Q_E$};

        \draw[->,>=stealth,very thick] (XiA.west) -- ([xshift=-0.6cm] XiA.west) node[left]{$A$};
        \draw[->,>=stealth,very thick] (XiB.east) -- ([xshift=+0.6cm] XiB.east) node[right]{$B$};
    \end{tikzpicture}
    \caption{
        A single round of the one-sided DIQKD protocol. Eve distributes some shared state to Alice and Bob (the map $\cE \in \Xi^{(E)}$) who then apply instruments $\cM \in \Xi^{(A)}$ and $\cN \in \Xi^{(B)}$ to obtain the outcomes $A$ and $B$.
    }
    \label{fig:GEAT_singleround}
\end{figure}

The other part where the protocol may abort is during parameter estimation, i.e., when $\mathrm{freq}_{\hat{C}^n} \notin \Omega_{\textsc{pe}}$.

\begin{lem} \label{lem:comp_pe}
    Let $p^\mathrm{hon} \in \mathbb{P}_\cC$ denote the honest distribution of the register $C$. Assume that $\Omega_\textsc{pe}$ is chosen as $\Omega_\textsc{pe} = \{ p \in \mathbb{P}_\cC : |p(c) - p^\mathrm{hon}(c)| \leq \delta \quad \forall c \in \cC \}$. If
    \begin{equation}
        \delta \geq \sqrt{\frac{1}{2n} \ln \frac{2\abs{\cC}}{\varepsilon_\textsc{pe}^\mathrm{comp}}},
    \end{equation}
    then the honest implementation aborts with probability at most $\varepsilon_\textsc{pe}^\mathrm{comp}$.
\end{lem}
\begin{proof}
For each $c \in \cC$, let $X^c_i$ be the indicator variable which is equal to $1$ if $C_i = c$ and $0$ otherwise. By Hoeffding's inequality \cite{Hoeffding1963}, we have that
\begin{equation}
    \Pr[\abs{\frac{1}{n} \sum_i X_i^c - p^{\mathrm{hon}}(c)} > \delta] \leq 2 e^{-2n \delta^2}.
\end{equation}
Hence, by the union bound, we have that
\begin{equation}
    \Pr[\mathrm{freq}_{C^n} \notin \Omega_{\textsc{pe}}] \leq 2\abs{\cC} e^{-2n\delta^2},
\end{equation}
where we noted that $\frac{1}{n} \sum_i X_i^c = \mathrm{freq}_{C^n}(c)$. By the given choice of $\delta$, this is less than $\varepsilon_{\textsc{pe}}^\mathrm{comp}$.
\end{proof}

\subsection{Soundness}

For the following theorem we consider a channel $\cE \in \Xi^{(E)}$ which distributes a state which is then measured by the instruments $\cM \in \Xi^{(A)}$ and $\cN \in \Xi^{(B)}$ (see Figure~\ref{fig:GEAT_singleround}). We will also denote by $M_{Q_A R_A}^a$ and $N_{Q_B R_B}^b$ the POVM elements of $\cM_{A R_B|Q_A R_A}$ and $\cN_{B R_B|Q_B R_B}$ when observing the outcomes $a$ and $b$ respectively.

\begin{thm} \label{thm:sdi_qkd_security}
    Assume that for every $\cM \in \Xi^{(A)}$ and $\cN \in \Xi^{(B)}$, the channel 
    \begin{equation} \label{eq:non_signalling}
        \textsc{pd}_{I|AB} \circ (\cM_{A R_A|Q_A R_A} \otimes \cN_{B R_B|Q_B R_B})
    \end{equation}
    does not signal from $R_A R_B$ to $I$ and denote by $\Efin$ Eve's quantum side-information at the end of the protocol.
    Let $\textsc{ca}$ be an arbitrary min-tradeoff function, i.e., an affine function satisfying
    \begin{equation}
        \textsc{ca}\big(\nu^{\rho,\cE,\cM,\cN}_C\big) \, \leq \, H(S|I Q_E)_{\nu^{\rho,\cE,\cM,\cN}} \quad \forall \rho,\,\cE \in \Xi^{(E)},\,\cM \in \Xi^{(A)},\,\cN \in \Xi^{(B)}, \\
    \end{equation}
    with
    \begin{equation}
    \begin{aligned}
        \nu^{\rho,\cE,\cM,\cN}_{SIC B Q_E} \coloneqq \sum_{a, b} & [\textsc{rk}(a, i)]_{S} \otimes [\textsc{pd}(a, b)]_I \otimes [\textsc{ev}(s, b, i)]_C \otimes \ketbra{b}_{B} \\
        &\otimes \Tr_{Q_A R_A Q_B R_B}[(M^a_{Q_A R_A} \otimes N^b_{Q_B R_B}) \circ \cE[\rho_{R_A R_B Q_E}]],
    \end{aligned}
    \end{equation}
    where $i$ and $s$ are shorthands for $\textsc{pd}(a, b)$ and $\textsc{rk}(a, i)$ respectively.

    Then, the one-sided DIQKD protocol in Box~\hyperref[box:protocol]{2} is $\varepsilon_\mathrm{snd}$-sound when instantiated with parameters satisfying
    \begin{equation} \label{eq:finite_size_bound}
    \begin{aligned}
        l &\leq \, n \min_{c^n \in \Omega_\textsc{pe}} \textsc{ca}(\mathrm{freq}_{c^n}) - \min_{\alpha \in (1, 3/2)} \kappa^{\textsc{geat}}\Big(n, \alpha, \frac{\varepsilon_s}{4}, \varepsilon_\mathrm{snd} - \varepsilon_\textsc{kv}, \textsc{ca} \Big) - 2 \log \frac{1}{2 \varepsilon_{\textsc{pa}}} \\
        &\hspace{10pt}- H_{\max}^{\varepsilon_s/4}(C^n|S^n I^n E_\mathrm{fin})_{\rho_{\land \tilde{\Omega}}} - \lambda_\textsc{ec} - \lceil \log \tfrac{1}{\varepsilon_\textsc{kv}} \rceil - 2\log\frac{1}{1-\sqrt{1-(\tfrac{\varepsilon_s}{4})^2}}, \\
        \varepsilon_s &= \frac{1}{2}(\varepsilon_\mathrm{snd} - \varepsilon_\textsc{kv} - \varepsilon_\textsc{pa}), \\
    \end{aligned}
    \end{equation}
    where $\rho$ is the state at the end of the protocol, $\Efin$ is Eve's side information at the end of the protocol,  $\kappa^{\textsc{geat}}$ is as defined in Theorem~\ref{thm:GEAT}, and
    \begin{equation}
    \begin{aligned}
        \tilde{\Omega} \coloneqq& \{ S^n, \hat{S}^n, \hat{C}^n : \mathrm{freq}_{\hat{C}^n} \in \Omega_\textsc{pe} \land (S^n = \hat{S}^n) \}.
    \end{aligned}
    \end{equation}
\end{thm}

\begin{remark}
    The $H_{\max}^\varepsilon$ term quantifies how much information is leaked through the testing registers $C^n$. Many protocols of interest perform infrequent sampling which means that only with probability $\gamma$ a test round is performed and otherwise the register $C_i$ takes some fixed value. One can then bound $H_{\max}^{\varepsilon_s/4}(C^n|S^n I^n \Efin) \lesssim n \gamma$, see Lemma~\ref{lem:hmax_bound}
\end{remark}

\begin{proof}[Proof of Theorem~\ref{thm:sdi_qkd_security}]
Let $\Omega$ denote the event of the protocol accepting (i.e. both key validation and the parameter estimation in Box~\hyperref[box:protocol]{2} accept) and $\Omega_{\textsc{ec}}^\mathrm{corr}$ denote the event that error-correction succeeded, that is $\hat{S}^n = S^n$. Let us denote by $\Efin$ Eve's quantum side information at the end of the protocol (hence her total side information is $O I^n \Efin$).
Then, the state $\rho_{K_A^lK_B^l O I^n \Efin \land \Omega}$ can be decomposed as
    \begin{equation}      
        \rho_{\land \Omega} = \rho_{\land \Omega \land \Omega_{\textsc{ec}}^\mathrm{corr}} + \rho_{\land \Omega \land (\Omega_{\textsc{ec}}^\mathrm{corr})^c}.
    \end{equation}
Note that by the properties of two-universal hashing, we have that $\mathrm{Pr}[\Omega \land (\Omega_{\textsc{ec}}^\mathrm{corr})^c] \leq \varepsilon_{\textsc{kv}}$. We can use the triangle inequality
\begin{align}
    &\frac{1}{2}\norm{\rho_{K_A^l K_B^l O I^n \Efin \land \Omega} - \tau_{K_A^l K_B^l} \otimes \rho_{O I^n \Efin \land \Omega}}_1 \\
    &\leq \frac{1}{2}\norm{\rho_{K_A^l K_B^l O I^n \Efin \land \Omega \land \Omega_{\textsc{ec}}^\mathrm{corr}} - \tau_{K_A^l K_B^l} \otimes \rho_{O I^n \Efin \land \Omega \land \Omega_{\textsc{ec}}^\mathrm{corr}}}_1 \\
    &\hspace{10pt}+ \frac{1}{2}\norm{\rho_{K_A^l K_B^l O I^n \Efin \land \Omega \land (\Omega_{\textsc{ec}}^\mathrm{corr})^c} - \tau_{K_A^l K_B^l} \otimes \rho_{O I^n \Efin \land \Omega \land (\Omega_{\textsc{ec}}^\mathrm{corr})^c}}_1 \\
    &\leq \frac{1}{2}\norm{\rho_{K_A^l K_B^l O I^n \Efin \land \Omega \land \Omega_{\textsc{ec}}^\mathrm{corr}} - \tau_{K_A^l K_B^l} \otimes \rho_{O I^n \Efin \land \Omega \land \Omega_{\textsc{ec}}^\mathrm{corr}}}_1 + \varepsilon_{\textsc{kv}}, \label{eq:trace_dist1}
\end{align}
where we used that $\mathrm{Pr}[\Omega \land (\Omega_{\textsc{ec}}^\mathrm{corr})^c] \leq \varepsilon_\textsc{kv}$ for the second inequality.

It remains to show that the first term is upper bounded by $\varepsilon_\mathrm{sec} \coloneqq \varepsilon_\mathrm{snd} - \varepsilon_\textsc{kv}$. Since the states are conditioned on the event $\Omega_\textsc{ec}^\mathrm{corr}$, we know that $S^n = \hat{S}^n$ and hence $K_A^l=K_B^l$. Therefore, the trace distance does not change when we trace out $K_B^l$. For ease of notation, set $\tilde{\Omega}\coloneqq\Omega\wedge\Omega_\textsc{ec}^\mathrm{corr}$ and decompose $\varepsilon_\mathrm{sec} = 2\varepsilon_s + \varepsilon_\textsc{pa}$. If $\mathrm{Pr}[\tilde{\Omega}] \leq \varepsilon_\mathrm{sec}$, the bound holds trivially. Therefore, we focus on the case $\mathrm{Pr}[\tilde{\Omega}] > \varepsilon_\mathrm{sec}$. From the quantum leftover hashing lemma (Lemma~\ref{lem:leftover_hashing}) we get 
\begin{equation} \label{eq:leftoverhash}
\begin{aligned}
    \frac{1}{2} \norm{\rho_{K_A^l O I^n \Efin \land \tilde{\Omega}} - \tau_{K_A^l} \otimes \rho_{O I^n \Efin \land \tilde{\Omega}}}_1 
    \leq 2\varepsilon_s + \frac{1}{2} \sqrt{2^{l-H_\mathrm{min}^{\varepsilon_s}(S^n|O I^n \Efin)_{\rho_{\land \tilde{\Omega}}}}}.
\end{aligned}
\end{equation}
Note that $H_\mathrm{min}^{\varepsilon_s}(S^n|O I^n \Efin)$ is evaluated on the quantum state right before privacy amplification. To ensure that the right hand side in \eqref{eq:leftoverhash} is upper bounded by $2\varepsilon_s + \varepsilon_\textsc{pa}$, we require
\begin{equation}
    l\le H_\mathrm{min}^{\varepsilon_s}(S^n|O I^n \Efin)_{\rho_{\land \tilde{\Omega}}} - 2\log\frac{1}{2\varepsilon_\textsc{pa}}.
\end{equation}

It is left to bound the entropy $H_\mathrm{min}^{\varepsilon_s}(S^n|O I^n \Efin)_{\rho_{\land \tilde{\Omega}}}$. First, we note that the side information leaked during error correction, $O$, consists of two parts: the number of error-correction bits communicated publicly between the honest parties, $\lambda_\textsc{ec}$, and the number of bits announced during key validation, $\lceil \log \frac{1}{\varepsilon_{\textsc{ec}}}\rceil$ (in principle, $O$ also contains the choice of hash function, but this is uncorrelated to $S^nI^n\Efin$). Using \cite[Lemma 6.8]{Tomamichel2016}, we obtain
\begin{align}
    H_{\min}^{\varepsilon_s}(S^n|O I^n \Efin)_{\rho_{\land \tilde{\Omega}}}
    &\geq H_\mathrm{min}^{\varepsilon_s}(S^n|I^n \Efin)_{\rho_{\land\tilde{\Omega}}} - \log \abs{O} \\
    &\geq H_\mathrm{min}^{\varepsilon_s}(S^n|I^n \Efin)_{\rho_{\land\tilde{\Omega}}} - \lambda_{\textsc{ec}} - \lceil\log\tfrac{1}{\varepsilon_\textsc{ec}}\rceil,
\end{align}
where we used that $\log \abs{O} \leq \lambda_{\textsc{ec}} + \lceil\log\tfrac{1}{\varepsilon_\textsc{ec}}\rceil$.

Next, we would like to apply the GEAT to lower-bound the smooth min-entropy. However, we cannot apply the GEAT directly since the testing registers $\hat{C}_i$ cannot be reconstructed from $S_i$ and Eve's side information in round $i$ (as required by~(\ref{eq:geat_projective_stats})). To overcome this, we want to include $\hat{C}^n$ in Alice's secret register. This can be done via the chain rule in \cite[Theorem 6.1]{Tomamichel2016}:
\begin{equation}
\begin{aligned}
    H_\mathrm{min}^{\varepsilon_s}(S^n|I^n \Efin)_{\rho_{\land\tilde{\Omega}}}
    &\geq H_\mathrm{min}^{\varepsilon_s/4}(S^n\hat{C}^n|I^n \Efin)_{\rho_{\land\tilde{\Omega}}} \\
    &\hspace{10pt}- H_\mathrm{max}^{\varepsilon_s/4}(\hat{C}^n|S^n I^n \Efin)_{\rho_{\land\tilde{\Omega}}} - 2\log\frac{1}{1-\sqrt{1-\left(\tfrac{\varepsilon_s}{4}\right)^2}},
\end{aligned}
\end{equation}
Next, we note that since we consider the state conditioned on $\tilde{\Omega}$ we have $\hat{S}^n = S^n$ and therefore $\hat{C}_i = \textsc{ev}(V_i, I_i, \hat{S}_i) = \textsc{ev}(V_i, I_i, S_i) \eqqcolon C_i$. Therefore,
\begin{equation}
    H_\mathrm{min}^{\varepsilon_s/4}(S^n\hat{C}^n|I^n \Efin)_{\rho_{\land\tilde{\Omega}}}
    = H_\mathrm{min}^{\varepsilon_s/4}(S^nC^n|I^n \Efin)_{\rho_{\land\tilde{\Omega}}}.
\end{equation}
Let us denote by $\Omega_{\hat{C}}$ the event that $\mathrm{freq}_{\hat{C}^n} \in \Omega_{\textsc{pe}}$, by $\Omega_{C}$ the event that $\mathrm{freq}_{C^n} \in \Omega_{\textsc{pe}}$ and by $\Omega_{\textsc{kv}}$ the event that $\textsc{Hash}(S^n) = \textsc{Hash}(\hat{S}^n)$.
Then, we can write $\tilde{\Omega} = \Omega \land \Omega_{\textsc{ec}}^\mathrm{corr} = \Omega_{\hat{C}} \land \Omega_{\textsc{kv}} \land \Omega_{\textsc{ec}}^\mathrm{corr} = \Omega_{C} \land \Omega_{\textsc{ec}}^ \mathrm{corr}$. Hence, applying \cite[Lemma 10]{Tomamichel_2017} twice, we have that
\begin{equation}
    H_\mathrm{min}^{\varepsilon_s/4}(S^nC^n|I^n \Efin)_{\rho_{\land\tilde{\Omega}}}
    \geq H_\mathrm{min}^{\varepsilon_s/4}(S^nC^n|I^n \Efin)_{\rho_{\land\Omega_C}}
    \geq H_\mathrm{min}^{\varepsilon_s/4}(S^nC^n|I^n \Efin)_{\rho_{|\Omega_C}}.
\end{equation}
To apply the GEAT, consider the channels
\begin{equation}
    \cM_{SIC Q_E R_A R_B|R_A R_B Q_E}^{(i)} \coloneqq \textsc{ev} \circ \textsc{rk} \circ \textsc{pd} \circ (\cM_{A R_A|Q_A R_A}^{(i)} \otimes \cN_{B R_B|Q_B R_B}^{(i)}) \circ \cE_{Q_A Q_B Q_E|Q_E}^{(i)},
\end{equation}
for some $\cM^{(i)} \in \Xi^{(A)}$, $\cN^{(i)} \in \Xi^{(B)}$, and $\cE^{(i)} \in \Xi^{(E)}$. Then, by~(\ref{eq:finite_size_bound}), $\textsc{ca}$ is a valid min-tradeoff function for $\cM^{(i)}$ (discarding $C$ can only decrease the entropy) and by (\ref{eq:non_signalling}), the non-signalling condition of the GEAT is satisfied.
Furthermore, since we assumed that $\mathrm{Pr}[\tilde{\Omega}] > \varepsilon_{\mathrm{sec}}$, we have that $\mathrm{Pr}[\Omega_C] \geq \mathrm{Pr}[\tilde{\Omega}] > \varepsilon_{\mathrm{sec}} = \varepsilon_{\mathrm{snd}} - \varepsilon_{\textsc{KV}}$. We can now apply the GEAT to obtain
\begin{equation}
    H_\mathrm{min}^{\varepsilon_s/4}(S^nC^n|I^n \Efin)_{\rho_{|\Omega_C}} \geq
    n \min_{c^n \in \Omega_\textsc{pe}} \textsc{ca}(\mathrm{freq}_{c^n}) - \min_{\alpha \in (1, 3/2)} \kappa^{\textsc{geat}} \Big( n, \alpha, \frac{\varepsilon_s}{4}, \varepsilon_\mathrm{snd} - \varepsilon_\textsc{kv}, \textsc{ca} \Big),
\end{equation}
which implies the claimed bound.
\end{proof}

The following Lemma can be used to bound the $H_{\max}$ term appearing in Theorem~\ref{thm:sdi_qkd_security}. 
\begin{lem} \label{lem:hmax_bound}
    Assume that there is a subset $\cI' \subseteq \cI$ such that for all $i' \in \cI'$, $\textsc{ev}$ is entirely determined by $i'$, i.e., $\textsc{ev}(u, v, i') = \textsc{ev}(u', v', i')$ for all $u, v, u', v'$. Then
    \begin{equation}
        H_{\max}(C^n|I^n S^n \Efin)_{\rho\land\tilde{\Omega}} \leq \log |\cC| \max_{i^n: \rho_{\land\tilde{\Omega}}(i^n) > 0} |\{j : i_j \notin \cI'\}|.
    \end{equation}
    In particular, the above bound also holds for $H_{\max}^\varepsilon$ for any $\varepsilon \geq 0$.
\end{lem}
\begin{proof}

By data-processing, we have that
\begin{equation}
    H_{\max}(C^n|S^n I^n \Efin)_{\rho_{\land \tilde{\Omega}}} \leq H_{\max}(C^n|I^n)_{\rho_{\land\tilde{\Omega}}}.
\end{equation}
Let us denote by $\rho_{\land \tilde{\Omega}}(i^n)$ the probability of $i^n$ in the state $\rho_{\land \tilde{\Omega}}$. We can bound 
\begin{equation}
\begin{aligned}
    H_{\max}(C^n|I^n)_{\rho_{\land \tilde{\Omega}}}
    &= \log \sum_{i^n} 2^{H_{\max}(C^n|I^n=i^n)} \rho_{\land\tilde{\Omega}}(i^n) \\
    &\leq \max_{i^n : \rho_{\land\tilde{\Omega}}(i^n) > 0} H_{\max}(C^n|I=i^n) \\
    &\leq \max_{i^n : \rho_{\land\tilde{\Omega}}(i^n) > 0} \log \abs{\cC} \abs{\{j : i_j \notin \cI'\}} \\
    &= \log \abs{\cC} \max_{i^n : \rho_{\land\tilde{\Omega}}(i^n) > 0} \abs{\{j : i_j \notin \cI'\}}, \\
\end{aligned}
\end{equation}
where for the second inequality we used that for every $i^n$, only the $C_j$ where $i_j \notin \cI'$ are not fixed.
\end{proof}

The Lemma above can, for example, be used when the protocol only performs infrequent test rounds, say with probability $\gamma$. Since for non-test rounds, the register $C_i$ takes a fixed value, the lemma above gives 
    \begin{equation}
        H_{\max}(C^n|I^n S^n \Efin) \leq \log \abs{C} n (\gamma + \delta),
    \end{equation}
where $\delta$ is a small tolerance given in Lemma~\ref{lem:comp_pe}.

\subsubsection{Collective attack bound}
To apply the security proof in Theorem~\ref{thm:sdi_qkd_security}, we need to instantiate the protocol with a valid min-tradeoff function, i.e., an affine function $\textsc{ca}: \mathbb{P}_{\cC} \rightarrow \mathbb{R}$ satisfying
\begin{equation} \label{eq:min_tradeoff}
    \textsc{ca}\big(\nu^{\rho,\cE,\cM,\cN}_C\big) \, \leq \, H(S|I Q_E)_{\nu^{\rho,\cE,\cM,\cN}} \quad \forall \rho,\, \cE \in \Xi^{(E)},\, \cM \in \Xi^{(A)},\, \cN \in \Xi^{(B)}, \\
\end{equation}
where $\nu$ is as given in Theorem~\ref{thm:sdi_qkd_security}. Since $\textsc{ca}$ is affine, we can write it as
\begin{equation} \label{eq:min_tradeoff_expansion}
    \textsc{ca}(p) = h_\lambda + \lambda^T p,
\end{equation}
where $\lambda \in \mathbb{R}^{\abs{\cC}}$ is the gradient of $\textsc{ca}$. Hence, the constraint in (\ref{eq:min_tradeoff}) is satisfied when
\begin{equation} \label{eq:h_lambda}
    h_\lambda \leq \inf_{\nu} \big(H(S|I Q_E)_{\nu} - \lambda^T \nu_C\big).
\end{equation}
Therefore, to obtain a valid min-tradeoff function we can pick an arbitrary $\lambda \in \mathbb{R}^{\abs{\cC}}$, compute $h_\lambda$ according to (\ref{eq:h_lambda}), and then obtain the min-tradeoff function $\textsc{ca}$ from (\ref{eq:min_tradeoff_expansion}). Since this procedure produces a valid min-tradeoff function for any choice of $\lambda$, we can optimize $\lambda$ numerically. Alternatively, one can pick the crude choice $\lambda = \nabla H(S|IV) |_{{p^{\mathrm{exp}}}}$, where $p^\mathrm{exp}$ are the expected statistics, and the gradient can be obtained from the dual optimization problem as outlined in \cite[Appendix E]{brown2024device}.

Using the min-tradeoff function constructed above, we can prove the following corollary which shows that the asymptotic key rate is given by the optimization problem in \eqref{eq:asymptotic_key_rate}.
\begin{cor}[Asymptotic key rate] \label{cor:asymptotic_key_rate}
    Let $r_\infty \coloneqq \lim_{n \rightarrow \infty} \frac{l}{n}$ be the asymptotic key rate of the protocol. Assume that the honest implementation achieves the distribution $\nu^\mathrm{hon}_C$ on the testing registers. Then
    \begin{equation}
        r_\infty = \inf_{\nu} (H(S|I Q_E)_{\nu} - H(S|IB)_{\nu}),
    \end{equation}
    where the infimum is over all states $\nu_{S I C B Q_E}$ of the form given in Theorem~\ref{thm:sdi_qkd_security} subject to the additional constraint that $\nu_C = \nu^{\mathrm{hon}}_C$.
\end{cor}
\begin{proof}
By the bound from Lemma~\ref{lem:comp_ec}, we have that
\begin{equation}
    \lambda_\textsc{ec} = n H(S|IB)_{p^{\mathrm{hon}}} + O(\sqrt{n}).
\end{equation}
Furthermore, by Lemma~\ref{lem:comp_pe} and linearity of $\textsc{ca}$, we have that
\begin{equation}
    \min_{c^n \in \Omega_\textsc{pe}} \textsc{ca}(\mathrm{freq}_{c^n}) \geq \textsc{ca}(p^\mathrm{hon}) - O(1/\sqrt{n})
\end{equation}
Moreover, we can choose the min-tradeoff function $\textsc{ca}$ such that
\begin{equation}
     \textsc{ca}(p^\mathrm{hon}) = \inf_{\nu} H(S|I Q_E)_{\nu},
\end{equation}
where the infimum runs over all states $\nu$ such that $\nu_C = p^{\mathrm{hon}}$.
By picking $\alpha = 1 + 1/\sqrt{n}$ in Theorem~\ref{thm:sdi_qkd_security}, we can compute
\begin{equation}
\begin{aligned}
    r_\infty
    =& \lim_{n \rightarrow \infty} \frac{l}{n} \\
    =& \lim_{n \rightarrow \infty} \frac{1}{n}(n \textsc{ca}(p^\mathrm{hon}) - \lambda_\textsc{ec} - O(\sqrt{n})) \\
    =& \lim_{n \rightarrow \infty} \left(\textsc{ca}(p^\mathrm{hon}) - \frac{\lambda_\textsc{ec}}{n} - O(1/\sqrt{n}) \right) \\
    =& \inf_\nu (H(S|I Q_E)_\nu - H(S|IB)_\nu)
\end{aligned}
\end{equation}
as claimed.
\end{proof}

\section{Relaxation of the conditional von Neumann entropy}
\label{sec:vonNeumann_relaxation}

As mentioned in Section \ref{sec:methods}, two ways have been proposed in the literature to approximate the conditional von Neumann entropy by polynomials such that optimization problems like \eqref{eq:key-rate-opti} can be solved using methods from non-commutative polynomial optimization.

The first one (called BFF in Section \ref{sec:comparison}) can be found in \cite{brown2024device}:

\begin{thm} \label{thm:bff-relax}
    Let $\rho_{ABQ_E} = \sum_{a,b} \ketbra{a,b}_{AB} \otimes \rho_{Q_E \land a, b}$ be a (not necessarily normalized) quantum state which is classical on $A$ and $B$. Then
    \begin{equation}
    \begin{aligned}
        &H(A|BQ_E)_\rho \\
        \geq & \inf_{Z_{Q_E|i,a,b}} \sum_{i = 1}^{n} \frac{w_i}{t_i \ln 2} \Big\{ \Tr[\rho] + \sum_{a,b}
        \begin{aligned}[t]
            &\Tr[\rho_{Q_E \land a, b}\left( Z_{Q_E|i,a,b} + Z^*_{Q_E|i,a,b} + (1 - t_i) Z^*_{Q_E|i,a,b}Z_{Q_E|i,a,b} \right)] \\
            &+ t_i \Tr[\rho_{Q_E \land b} Z_{Q_E|i,a,b}Z^*_{Q_E|i,a,b}] \Big\},
        \end{aligned}
    \end{aligned}
    \end{equation}
    where $\rho_{Q_E \land b} = \sum_{a} \rho_{Q_E \land a, b}$ and $w_i$, $t_i$ are the weights and nodes of an $n$-point Gauss-Radau quadrature on $[0, 1]$ with endpoint $t_n = 1$. Here, the infimum is taken over $Z_{Q_E|i,a,b}$ such that $\|Z_{Q_E|i,a,b}\|_\infty \leq \alpha_i$, where $\alpha_i = \frac{3}{2}\max\{t_i^{-1}, (1-t_i)^{-1}\}$.
\end{thm}

The second one (called KS in Section \ref{sec:comparison}) can be found in \cite{kossmann2025reliableentropyestimationobserved,kossmann2025optimisingrelativeentropysemidefinite}:
\begin{thm} \label{thm:ks_relax}
    Let $\rho_{AB Q_E} = \sum_{a,b} \ketbra{a,b}_{AB} \otimes \rho_{Q_E \land a, b}$ be a (not necessarily normalized) quantum state which is classical on $A$ and $B$. Then
    \begin{equation}
    \begin{aligned}
        &H(A|BQ_E)_\rho \\
        \geq&\ \frac{1}{\ln 2}\left( \Tr[\rho](\abs{A} - 1) + \inf \sum_{k=0}^{r} \sum_{a, b} \Tr[-\left(\alpha_{k} \rho_{Q_E \land a, b} + \beta_{k} \rho_{Q_E \land b} \right) P_{Q_E|k,a,b}] \right),
    \end{aligned}
    \end{equation}
    where the infimum runs over Hermitian projectors $P_{Q_E|k,a,b}$ and $\rho_{Q_E \land b} = \sum_{a} \rho_{Q_E \land a, b}$. Here $r \in \mathbb{N}$ and the real coefficients $\{\alpha_k\}_{k=0}^{r}$ and $\{\beta_k\}_{k=0}^{r}$ are the parameters of the $(r+1)$-term approximation of the logarithm underlying the bound; their explicit values are given in \cite{kossmann2025reliableentropyestimationobserved,kossmann2025optimisingrelativeentropysemidefinite}, and the bound becomes tighter as $r$ increases.
\end{thm}
Both the KS and BFF approximation of the conditional von Neumann entropy can be combined with either the NPA-AC (Section \ref{subsec:npa_hierarchy_with_matrix_constraints}) or NPA-MP (Section \ref{sec:matrix-NPA-methods}) hierarchy.  To illustrate this point, for some of these combinations we will write down the non-commutative polynomial optimization problems approximating $H(A\vert X = x_0,Q_E)_\sigma$, where $Q_E$ is the quantum system held by Eve
and $\sigma_{A Q_E|X=x_0} = \sum_{a} \ketbra{a}_{A} \otimes \tr_{Q_Q}[M_{Q_A}(a|x_0) \rho_{Q_A Q_E}]$ is the state after measuring $\rho$ using $M(a|x_0)$.
In the remaining part of this section, we will first write down the BFF approximation + NPA-MP type constraints and then the KS approximation + NPA-AC type constraints. The remaining combinations, KS approximation + NPA-MP type constraints and BFF approximation + NPA-AC type constraints, can easily be obtained from Theorems \ref{thm:bff-relax} and \ref{thm:ks_relax}. While the resulting problems look quite different, they all approximate \eqref{eq:key-rate-opti}.

First, we will use the BFF approximation in Theorem \ref{thm:bff-relax} to bring \eqref{eq:key-rate-opti} into a form that can subsequently be tackled with the NPA-MP hierarchy. Thus, $H(A\vert X = x_0,Q_E)_\sigma$ is lower bounded by

\begin{align}
    \inf & \sum_{i=1}^{n} \frac{w_i}{t_i \ln{2}} \Big( 1 + \sum_{a \in \mathcal A} \bra{\psi_{Q_AQ_BQ_E}}(M(a|x_0) \otimes (Z_{a,i} + Z_{a,i}^\ast  &\nonumber \\ & + (1-t_i)Z_{a,i}^\ast Z_{a,i}) + t_i(\mathds{1}_{Q_A} \otimes Z_{a,i} Z_{a,i}^\ast)) \ket{\psi_{Q_AQ_BQ_E}} \Big) &\\
    \mathrm{s.t.} &  \quad \sum_{abxy} c_{abxyi} \bra{\psi_{Q_AQ_BQ_E}}M(a|x)\otimes N(b|y) \ket{\psi_{Q_AQ_BQ_E}} \geq q_i & \forall i \in \mathcal C  \nonumber\\
    & \quad \sum_b N(b|y) = \mathds{1}_{Q_BQ_E} & \forall y \in \mathcal Y \nonumber \\
      & \quad N(b|y) \geq 0 & \forall b \in \mathcal B, \forall y \in \mathcal Y \nonumber\\
    & \quad Z_{a,i} Z^\ast_{a,i} \leq \alpha_i^2 \mathds{1}_{Q_BQ_E}  & \forall a \in \mathcal A, \forall i \in [n] \nonumber\\
     & \quad Z^\ast_{a,i} Z_{a,i}  \leq \alpha_i^2 \mathds{1}_{Q_BQ_E} & \forall a \in \mathcal A, \forall i \in [n] \nonumber\\
     & \quad [ Z_{a,i}, N(b|y)] = 0 & \forall a, b, y, \forall i \in [n]\nonumber\\
    & \braket{\psi_{Q_AQ_BQ_E}|\psi_{Q_AQ_BQ_E}}=1 &\nonumber
\end{align}
Here, we optimize over $Z_{a,i}$, $\ket{\psi}$, $N(b|y)$, and the size of $Q_BQ_E$ (the $Q_A$ system and the measurements $M(a|x_0)$ are fixed). We only impose that operators on $Q_B$ and operators on $Q_E$ commute instead of a tensor product structure.

Next, we will use the KS approximation in Theorem \ref{thm:ks_relax} to bring \eqref{eq:key-rate-opti} into a from that can subsequently be tackled with the NPA-AC hierarchy. $H(A\vert X = x_0,Q_E)_\sigma$ is then lower bounded by
\begin{equation}\label{eq:optimization_problem_eta}
\begin{aligned}
    \inf \ \frac{1}{\ln 2}&\left(\vert A\vert - 1
    + \sum_{k=0}^r \sum_{a \in \mathcal A}
       \Big\langle\psi_{Q_A Q_B Q_E}\Big\vert
          -\alpha_k
            \left(\sum_{s,t=1}^m \sum_{b\in \mathcal B}
               \lambda_{s,t}^{(a\vert x_0)} \eta_{s,t,b\vert \tilde{y}}\right) P_k^{(a)} 
          - \beta_k P_k^{(a)}
       \Big\vert\psi_{Q_A Q_B Q_E}\Big\rangle\right) \\
  \mathrm{s.t.}\quad &\sum_{a,b,x,y} c_{abxyi}
      \Big\langle \psi_{Q_A Q_B Q_E}\Big\vert
         \sum_{s,t=1}^m \lambda_{s,t}^{(a\vert x)} \eta_{s,t,b\vert y}
      \Big\vert\psi_{Q_A Q_B Q_E}\Big\rangle
      \geq q_i,\quad  i \in \mathcal C, \\
   &(P_k^{(a)})^2 = P_k^{(a)}, \quad (P_k^{(a)})^\ast = P_k^{(a)},\quad 0\leq k \leq r,\ a \in \mathcal A, \\
   &\eta_{i,j,b\vert y} P_k^{(a)} = P_k^{(a)} \eta_{i,j,b\vert y}
      \quad\forall\, i,j,\ b\in \mathcal B,\ y\in \mathcal Y,\ 0\leq k\leq r,\ a\in \mathcal A, \\
   &\text{and the generators \(\eta_{i,j,b\vert y}\) satisfy the relations}
\end{aligned}
\end{equation}
\begin{align}
\tag{R1}
    &\eta_{i,j,b\vert y}^\ast = \eta_{j,i,b\vert y}
      &&\forall\, i,j,\ b\in \mathcal B,\ y\in \mathcal Y,\\[0.3em]
\tag{R2}
    &\eta_{i,j,b\vert y}\,\eta_{k,\ell,b\vert y}
      = \delta_{j,k}\,\eta_{i,\ell,b\vert y}
      &&\forall\, i,j,k,\ell,\ b\in \mathcal B,\ y\in \mathcal  Y,\\[0.3em]
\tag{R3}
    &\eta_{i,j,b\vert y}\,\eta_{k,\ell,b'\vert y} = 0
      &&\forall\, i,j,k,\ell,\ y\in \mathcal Y,\ b\neq b',\\[0.3em]
\tag{R4}
    &\sum_{i=1}^m \sum_{b\in \mathcal B} \eta_{i,i,b\vert y} = \mathds{1}
      &&\forall\, y\in \mathcal Y,\\[0.3em]
\tag{R5}
    &\sum_{b\in \mathcal  B} \eta_{i,j,b\vert y}
      = \sum_{b\in \mathcal B} \eta_{i,j,b\vert y'}
      &&\forall\, i,j,\ \forall\, y,y'\in \mathcal Y.
\end{align}
Observe that we choose a fixed but arbitrary $\tilde{y} \in \mathcal{Y}$ in the program above.  As mentioned initially, we will not give the KS approximation + NPA-MP type constraints and the BFF approximation + NPA-AC type constraints for \eqref{eq:key-rate-opti} to avoid repetition.

\section{NPA hierarchy with matrix-valued polynomials} \label{sec:matrix-NPA}

In this section, we prove convergence of the hierarchy of SDPs presented in Section~\ref{sec:matrix-NPA-methods}.

\subsection{Proof of convergence}

In this section, we will generalize the SDP hierarchies in \cite{johnston2016extended, escola2025lossy}, which generalize the NPA hierarchy \cite{pironio2010convergent}. The existence of such a generalization was already pointed out in \cite{pironio2010convergent}, because the proof should follow from the Positivstellensatz in \cite{helton2004positivstellensatz}. Here, we will give an explicit proof for the sake of accessibility to researchers in quantum cryptography. We will follow very closely the proof strategy in \cite{pironio2010convergent} and prove convergence directly. We want to solve the following problem:
We assume that $\mathcal Q$, $\mathcal R$ are finite sets of indices, that $m \in \mathbb N$ and $m_j \in \mathbb N$ for all $j \in \mathcal Q$ are fixed dimensions, and that $p$, $r^{(\ell)}$ are Hermitian $\mathcal B(\mathbb C^m)$-valued polynomials in $2g$ non-commutative variables $x\coloneq (x_1, \ldots, x_g, x_1^\ast, \ldots, x_g)$, $g \in \mathbb N$. Being Hermitian in particular implies that $p(X)$, $r^{(\ell)}(X)$ are self-adjoint for any self-adjoint $X$ and any $\ell \in \mathcal R$.
The $q^{(j)}$ are Hermitian $\mathcal B(\mathbb C^{m_j})$-valued polynomials in $2g$ non-commutative variables $x$ for any $j \in \mathcal Q$.

Then, the optimization problem is the following:
\begin{align*}
    \mathrm{minimize} \quad & \quad \bra{\psi} p(X) \ket{\psi} \\
    \mathrm{such~that} \quad & \quad q^{(j)}(X) \geq 0 \qquad \forall j \in \mathcal Q \\
    \quad & \quad \bra{\psi}r^{(\ell)}(X) \ket{\psi} \geq 0  \qquad \forall \ell \in \mathcal R \tag{\textbf{P}} \label{eq:problem-P} \\
    \quad & \quad \langle \psi | \psi \rangle = 1 \\
    \quad & \quad \ket{\psi} \in \mathbb C^m \otimes \mathcal H \\
    \quad & \quad X \in \mathcal B(\mathcal H)^g \\
    \quad & \quad \mathcal H \mathrm{~Hilbert~space}
\end{align*}
Note that the Hilbert space is not restricted to being finite-dimensional.

\begin{remark}
    Note that constraints of the form $a(X)\ket{\psi} = 0$ with $a \in \mathcal B(\mathbb C^m, \mathbb C^n)[x,x^\ast]$ can be imposed as inequality constraints by imposing
    \begin{equation}
        -\bra{\psi} a(X)^\ast a(X) \ket{\psi} \geq 0 \,.
    \end{equation}
    Equality constraints of the form $q^{(j)}(X) = 0$ and $\bra{\psi} r^{(\ell)} \ket{\psi} = 0$ can also be imposed as two inequalities. We refer to \cite[Section  3.5]{pironio2010convergent} for a discussion of these questions and of how equality constraints can be imposed best.
\end{remark}

 Let $\mathcal B(\mathbb C^m)[x, x^\ast] = \mathcal B(\mathbb C^m) \otimes \mathbb C[x,x^\ast]$ be the algebra of of polynomials with coefficients in $\mathcal B(\mathbb C^m)$ and non-commutative variables $x = (x_1, \ldots, x_g)$, $x^\ast = (x_1^\ast, \ldots, x_g^\ast)$. Likewise, let $\mathcal B(\mathbb C^{m}, \mathbb C^n)[x, x^\ast]$ be the polynomials in non-commutative variables $x = (x_1, \ldots, x_g)$, $x^\ast = (x_1^\ast, \ldots, x_g^\ast)$ with coefficients in $\mathcal B(\mathbb C^{m}, \mathbb C^n)$, where $n \in \mathbb N$.  We will often see the variables as $2g$ variables $x_1, \ldots, x_{2g}$ where $x_{g+i}\coloneq x_i^\ast$ for all $i \in [g]$. We can write any $p \in \mathcal B(\mathbb C^m, \mathbb C^n)[x, x^\ast]$ as 
\begin{equation}
    p = \sum_{w} p_w w, 
\end{equation}
where $w$ are words in $2g$ variables and $p_w \in B(\mathbb C^m, \mathbb C^n)[x, x^\ast]$. Note that the sum is finite and can be restricted to run over words of length $|w|\leq \deg{p}$. The evaluation of a polynomial $p \in \mathcal B(\mathbb C^n, \mathbb C^m)[x, x^\ast]$ on a tuple $X \in \mathcal B(\mathcal H)^g$ for some Hilbert space $\mathcal H$ is defined as
\begin{equation}
        p(X) = \sum_{w} p_w \otimes w(X).
\end{equation}
By $w: X \mapsto w(X)$, we mean the map that replaces every $x_i$ in the word $w$ by $X_i$ and every $x_i^*$ in the word $w$ by $X_i^*$.

We identify the $*$-operation with an involution, i.e., if $w=w_1 w_2 \ldots w_n$, then $w^\ast=w_n^\ast \ldots w_2^\ast w_1^\ast$. Hence also $p^\ast \coloneqq \sum_w p_w^\ast w^\ast$, where $p_w^\ast$ is the Hilbert-space adjoint of $p_w$. A polynomial $p$ is \emph{Hermitian} if $p = p^\ast$. We denote the empty word by $\emptyset$. We will write $\mathcal W_d$ for the set words $w$ of length $|w| \leq d$ and write $\mathcal W_\infty$ for the set of all words without length restrictions.  For convenience, we will denote the orthonormal basis of $\mathbb C^{|\mathcal W_d|}$ by $\{\ket{w}\}_{w \in \mathcal W_d}$. We write $\mathcal B(\mathbb C^m, \mathbb C^n)[x, x^\ast]_d$ and $\mathcal B(\mathbb C^m)[x, x^\ast]_d$ for the respective polynomials of degree at most $d \in \mathbb N$. Thus, more formally, $p$, $r^{(\ell)} \in \mathcal B(\mathbb C^m)[x, x^\ast]$ for all $\ell \in \mathcal R$, and $q^{(j)} \in \mathcal B(\mathbb C^{m_j})[x, x^\ast]$ for all $j \in \mathcal Q$, where $p$, $q^{(j)}$, and $r^{(\ell)}$ are Hermitian for all $j \in \mathcal Q$, $\ell \in \mathcal R$.

Moreover, let $y=(y_w)_{|w| \leq d} \in \mathcal B(\mathbb C^m)^{|\mathcal W_d|}$ be a $|\mathcal W_d|$-tuple of matrices of dimension $m$ and let $y^{ij}_w$ be the entry $\bra{i}y_w\ket{j}$ for all $i$, $j \in [m]$. Then, for $n \in \mathbb N$ we can define $L_y: \mathcal B(\mathbb C^n)[x,x^\ast]_d \to \mathcal B(\mathbb C^n) \otimes \mathcal B(\mathbb C^m)$ to be a linear operator on polynomials such that $L_y(p) = \sum_{|w| \leq d} p_w \otimes y_w$. Using this linear operator, for given $y \coloneqq (y_w)_{|w| \leq 2k}$ we can define the \emph{moment matrix} of order $k$, $M_k(y)$, as a block matrix with blocks of size $m \times m$, where
\begin{equation}
    M_k(y)(s,t) = L_{y}(s^* t) = y_{s^*t}
\end{equation}
for $s$, $t \in \mathcal W_k$. Thus, $M_k(y)$ is a matrix of size $m|\mathcal W_k|$.

In a similar way, given a $\mathcal B(\mathbb C^n)$-valued polynomial $q = \sum_{|w|\leq d}q_ww$ of degree $d$ and $y=(y_w)_{|w| \leq 2k+d}$ we define the \emph{localizing matrix} $M_k(qy)$ of order $k$ as a block matrix  with blocks of size $nm \times nm$, where
\begin{equation}
    (M_k(qy))(s,t) = L_{y}(s^\ast q t) = \sum_{|w|\leq d} q_w \otimes y_{s^\ast w t}
\end{equation}
for $s$, $t \in \mathcal W_k$.

Let $k \geq k_0 \coloneq \lceil \max_{j \in \mathcal Q, \ell \in \mathcal R}\{\operatorname{deg}(p), \operatorname{deg}(q^{(j)}), \operatorname{deg}(r^{(\ell)})\}/2\rceil$. We claim that the following optimization problems in terms of moment matrices converge to the solution of the former problem:

\begin{align*}
        \mathrm{minimize} \quad & \quad  \sum_{|w| \leq 2k} \Tr[p_w y_w^T] \\
    \mathrm{such~that} \quad & \quad M_k(y) \geq 0 \tag{$\mathbf{R_k}$} \label{eq:problem-Rk}  \\
    \quad & \quad M_{k-d_j}(q^{(j)}y) \geq 0  \qquad \forall j \in \mathcal Q \\
    \quad & \quad \sum_{|w| \leq 2k} \Tr[r^{(\ell)}_w y_w^T] \geq 0 \qquad \forall \ell \in \mathcal R\\ 
    \quad & \quad \sum_{i=1}^m y_{\emptyset}^{ii} = 1 \\
    \quad & \quad y_w =(y_w^{ij})_{i, j \in [m]} \in \mathcal B(\mathbb C^m) \qquad \forall w \in \mathcal W_{2k}
\end{align*}
Here, $p_w \in \mathcal B(\mathbb C^m)$ are the coefficients of the non-commutative polynomial $p$, the $r^{(\ell)}_w \in \mathcal B(\mathbb C^m)$ the coefficients of the non-commutative polynomials $r^{(\ell)}$, and $d_j \coloneq \lceil \operatorname{deg}(q^{(j)})/2\rceil $ for all $j \in \mathcal Q$. We take the transposition in the canonical basis $\{\ket{i}\}_{i \in [m]}$.

Intuitively, why this should work is the following: The $y_w$ are $m \times m$ matrices whose entries should correspond to
\begin{equation}
    y_w^{ij} = \bra{\psi_i} w(X) \ket{\psi_j}
\end{equation}
where $w$ is a word, $\mathcal H$ is a Hilbert space, $\ket{\psi_i}$, $\ket{\psi_j} \in \mathcal H$ are vectors, $X=(X_1, \ldots, X_g)$ is a tuple of operators in $\mathcal B(\mathcal H)$, and we have expanded the state $\ket{\psi} \in \mathbb C^m \otimes \mathcal H$ as
\begin{equation}
    \ket{\psi} = \sum_{i=1}^m \ket{i} \otimes \ket{\psi_i}.
\end{equation}
Since $\ket{\psi}$ is a quantum state and hence normalized, we obtain
\begin{equation} \label{eq:normalized-trace}
    \sum_{i=1}^m y_{\emptyset}^{ii} = \sum_{i=1}^m \langle \psi_i | \psi_i \rangle = \langle \psi | \psi \rangle = 1 \,.
\end{equation}
Moreover, for $k \geq k_0$,
\begin{align}
    \sum_{|w| \leq 2k}\Tr[p_w y_w^T] &= \sum_{|w| \leq 2k} \sum_{i,j\in [m]} \bra{i} p_w \ket{j} y_w^{ij}\nonumber  \\
    &= \sum_{|w|\leq 2k} \sum_{i,j \in [m]} \bra{i} p_w \ket{j} \bra{\psi_i} w(X) \ket{\psi_j} \nonumber \\
    &= \bra{\psi} p(X) \ket{\psi} \, \label{eq:same-objective-function} 
\end{align}
and similarly for the constraints indexed by $\ell \in \mathcal R$.   Of course, this calculation is not a proof of any kind and should just serve to build intuition concerning the objects appearing in the SDP hierarchy and which roles they are playing.

In the following, we will prove the claim that the optimal values for the problems \eqref{eq:problem-Rk} indeed converge to the optimal value of \eqref{eq:problem-P}. The proof will follow very closely the one in \cite{pironio2010convergent}.

We say that $y$ admits a \emph{moment representation} if there exist a Hilbert space $\mathcal H$, vectors $\ket{\psi_i} \in \mathcal H$ such that 
$\sum_{i \in [m]} \braket{\psi_i| \psi_i} =1$
and operators $X \in \mathcal B(\mathcal H)^g$ such that
\begin{equation}
    y_w^{ij} = \bra{\psi_i} w(X) \ket{\psi_j} \,.
\end{equation}
The next lemma morally states that if \eqref{eq:problem-P} is feasible, then so is \eqref{eq:problem-Rk} for appropriate $k$. Together with  \eqref{eq:same-objective-function}, Lemma \ref{lem:solutions-exist} implies that the optimal value $p^{(k)}$ of \eqref{eq:problem-Rk} is a lower bound for the optimal value $p^\opt$ of \eqref{eq:problem-P}, i.e., $p^{(k)} \leq p^\opt$ for all $k \geq k_0$.

\begin{lem} \label{lem:solutions-exist}
Let $y = (y_w)_{|w| \leq 2k}$ have a moment representation. Then, $\sum_{i=1}^m y_{\emptyset}^{ii}=1$ and $M_k(y) \geq 0$. If the moment representation is such that $q(X) \geq 0$ for some $q \in \mathcal B(\mathbb C^n)[x,x^\ast]$, then in addition $M_{k-d}(qy) \geq 0$ for $d = \lceil \operatorname{deg}(q)/2\rceil$ and $k \geq d$. Finally,  let $r \in \mathcal B(\mathbb C^m)[x, x^\ast]$ and $2k \geq \deg(r)$. Then, if the moment representation is such that $\bra{\psi} r(X) \ket{\psi} \geq 0$, then $\sum_{|w|\leq 2k} \Tr[r_{w} y_w^T] \geq 0$. 
\end{lem}
\begin{proof}
    The first assertion follows directly from  \eqref{eq:normalized-trace}. By definition, if $y$ has a moment representation, for all $s$, $t \in \mathcal W_k$ and all $i$, $j \in [m]$,
    \begin{equation}
        M_k^{ij}(y)(s,t) = y^{ij}_{s^\ast t} = \bra{\psi_i} s^\ast(X) t(X) \ket{\psi_j}.
    \end{equation}
Thus, for $\ket{z} \in \mathbb C^m \otimes \mathbb C^{|\mathcal W_k|}$, $\ket{z}  = \sum_{n = 1}^m \ket{n} \otimes (\sum_{|w| \leq k} z_w^{(n)} \ket{w})$, where $z_w^{(n)} \in \mathbb C$ for all $w \in \mathcal W_k$ and $n \in [m]$, we infer
\begin{align}
    \bra{z} M_k(y) \ket{z} &= \sum_{i,j \in [m]} \sum_{|v|,|w| \leq k} (z^{(i)}_v)^\ast M_k^{ij}(y)(v,w) z^{(j)}_w \\
    &= \sum_{i,j \in [m]} \sum_{|v|,|w| \leq k} (z^{(i)}_v)^\ast\bra{\psi_i} v^\ast(X) w(X) \ket{\psi_j} z^{(j)}_w \\
    &= \sum_{i,j \in [m]} \bra{\psi_i} z_i(X)^\ast z_j(X) \ket{\psi_j} \\
    &= \left\|\sum_{j \in [m]} z_j(X) \ket{\psi_j} \right\|_2^2 \geq 0.
\end{align}
Here, we have defined $z_i \in \mathbb C[x, x^\ast]_k$ as $z_i = \sum_{|w|\leq k} z_w^{(i)} w$ for all $i \in [m]$. This completes the proof of the second assertion.

For the third assertion, in the same vein for all $s$, $t \in \mathcal W_{d-k}$ and $b$, $c \in [n]$, $i$, $j \in [m]$, we have that
\begin{align}
    (\bra{i} \otimes \bra{b}) M_{k-d}(qy)(s,t) (\ket{j} \otimes \ket{c}) &= \sum_{|w| \leq \operatorname{deg}(q)} \bra{b}q_w\ket{c} y^{ij}_{s^\ast w t}\\& = \sum_{|w|\leq \operatorname{deg}(q)} \bra{b}q_w\ket{c} \bra{\psi_i} s^\ast(X) w(X) t(X) \ket{\psi_j}.
\end{align}
Hence, for $\ket{z} \in \mathbb C^n \otimes \mathbb C^m \otimes \mathbb C^{|\mathcal W_k|}$, $\ket{z}  = \sum_{b \in [n]} \sum_{i \in [m]} \ket{b} \otimes \ket{i} \otimes (\sum_{|w| \leq k} z_w^{(b,i)} \ket{w})$, where $z_w^{(b,i)} \in \mathbb C$ for all $w \in \mathcal W_k$ and $b \in [n]$, $i \in [m]$, we obtain
\begin{align}
    &\bra{z}M_{k-d}(qy)\ket{z} \\ &= \sum_{b,c \in [n]}\sum_{i,j \in [m]} \sum_{|v|, |w| \leq k-d} (z^{(b,i)}_v)^\ast (\bra{i} \otimes \bra{b}) M_{k-d}(qy)(s,t) (\ket{j} \otimes \ket{c}) z^{(c,j)}_w \\
    & =\sum_{b,c \in [n]} \sum_{i,j \in [m]} \sum_{|v|, |w| \leq k-d} (z^{(b,i)}_v)^\ast \sum_{|u| \leq \operatorname{deg}(q)} \bra{b}q_u\ket{c} \bra{\psi_i} v^\ast(X) u(X) w(X) \ket{\psi_j} z^{(c,j)}_w \\
    & =\sum_{b,c \in [n]} \sum_{i,j \in [m]} (\bra{b} \otimes \bra{\psi_i} (z^{(b,i)}(X))^\ast) q(X) (\ket{c} \otimes z^{(c,j)}(X) \ket{\psi_j}) \\
    &= \bra{\zeta} q(X) \ket{\zeta} \geq 0,
\end{align}
where $\ket{\zeta} = \sum_{b \in [n]} \sum_{i \in [m]} \ket{b} \otimes z^{(b,i)}(X) \ket{\psi_i} \in \mathbb C^n \otimes \mathcal H$ and $z^{(b,i)} = \sum_{|v|\leq k-d} z^{(b,i)}_v v \in \mathbb C[x,x^\ast]$ for all $b \in [n]$, $i \in [m]$. Finally,
\begin{align}
    \sum_{|w| \leq 2k}  \Tr[r_w y_w^T] & = \sum_{i,j \in [m]}\sum_{|w| \leq \operatorname{deg}(r)} \bra{i}r_w\ket{j}  y^{ij}_w \\
    & = \sum_{i,j \in [m]}\sum_{|w| \leq \operatorname{deg}(r)} \bra{i}r_w\ket{j}  \bra{\psi_i} w(X) \ket{\psi_j} \\
    &= \bra{\psi} r(X) \ket{\psi} \geq 0 \,.
\end{align}
This proves the last assertion.
\end{proof}
Next, we need to show that the constraints are well-behaved. For this, we will need to recall definitions from \cite{klep2007nichtnegativstellensatz, pironio2010convergent}.

Let $Q=\{q^{(j)}: j \in \mathcal Q\}$ be the polynomials appearing as positivity constraints in problem \eqref{eq:problem-P}. Then, the \emph{positivity domain} $\mathcal S_Q$ associated to $Q$ is the set of tuples $X=(X_1, \ldots, X_g)$ of bounded operators on a Hilbert space such that each $q^{(j)}(X)$, $j \in  \mathcal Q$, is a positive operator. The \emph{quadratic module} $\mathcal M_Q$ is the set of all elements of the form $\sum_i f_i^\ast f_i + \sum_i \sum_{j\in \mathcal Q} g_{ij}^\ast q^{(j)} g_{ij}$, where $f_i$, $g_{ij} \in \mathcal B(\mathbb C^\ell)[x, x^\ast]$ and the sums are finite. Here, $\ell = \max_{j \in \mathcal Q} m_j$ and we embed smaller matrices in the upper left corner of larger matrices by adding zeros. The quadratic module is Archimedean if there exists $C \in \mathbb R$ such that $C^2 - \sum_{i \in [2g]} x_i^\ast x_i \in \mathcal M_Q$. If this is the case, then $\mathcal S_Q$ is bounded because it means that for all $X\in \mathcal S_Q$, $C^2 - \sum_{i \in [2g]} X_i^\ast X_i \geq 0$. If $\mathcal S_Q$ is already bounded, we can always find $C$ large enough such that we can make $\mathcal M_Q$ Archimedean by adding $C^2 - \sum_{i \in [2g]} x_i^\ast x_i$ to $Q$ without changing $\mathcal S_Q$. 

\begin{lem} \label{lem:c-bounds-coefficients}
    Let $c= C^2 - \sum_{i \in [2g]} x_i^\ast x_i$ and $k \geq 1$. Let $y$ be a sequence such that $M_{k-1}(cy) \geq 0$, $M_k(y) \geq 0$, and $\sum_i y^{ii}_{\emptyset} = 1$. Then, $|y_w^{ij}| \leq C^{|w|}$ for all $i$, $j \in [m]$ and $|w|\leq 2k$.
\end{lem}
\begin{proof}
    By assumption, the diagonal elements of $M_{k-1}(cy)$ need to be positive, i.e., 
    \begin{equation}
        C^2 y_{w^\ast w}^{ii} - \sum_{j\in [2g]} y_{w^\ast x_j^\ast x_j w}^{ii} \geq 0
    \end{equation}
    for all $i \in [g]$ and  $|w|\leq k-1$. Thus, $\sum_{j\in [2g]} y_{w^\ast x_j^\ast x_j w}^{ii} \leq     C^2 y_{w^\ast w}^{ii}$ for all $i \in [g]$ and  $|w|\leq k-1$. Furthermore, $y_{w^\ast x_j^\ast x_j w}^{ii} \geq 0$ for all $|w|\leq k-1$ and $i$, $j \in [g]$, as $M_k(y) \geq 0$. Hence, it follows that for all $i$, $j \in [m]$ and $|w| \leq k-1$,
    \begin{equation}
        y_{w^\ast x_j^\ast x_j w}^{ii} \leq \sum_{\ell\in [2g]} y_{w^\ast x_\ell^\ast x_\ell w}^{ii} \leq     C^2 y_{w^\ast w}^{ii}.
    \end{equation}
    Moreover,  we infer from $\sum_i y^{ii}_{\emptyset} = 1$ that $y^{ii}_{\emptyset} \leq 1$. Combing with the above, we obtain $y_{x_j^\ast x_j}^{ii} \leq     C^2$ for all $i$, $j \in [m]$. By induction, we infer that $y_{w^\ast w}^{ii} \leq     C^{2|w|}$ for all $|w|\leq k-1$. This proves the assertion for diagonal elements. For off-diagonal elements, we can consider the submatrix
    \begin{equation}
        \begin{pmatrix}
            y^{ii}_{s^\ast s} & y^{ji}_{s^\ast t} \\  
            y^{ij}_{t^\ast s} & y^{jj}_{t^\ast t}   
        \end{pmatrix},
    \end{equation}
    where $i$, $j \in [m]$ and $|s|$, $|t|\leq k$, which is positive semidefinite because $M_k(y)$ is. Hence, $y^{ji}_{s^\ast t} = (y^{ij}_{t^\ast s})^\ast$. As the determinant of this positive semidefinite submatrix must be positive, we obtain $ y^{ii}_{s^\ast s} y^{jj}_{t^\ast t} \geq |y^{ij}_{t^\ast s}|^2$, from which $|y^{ij}_{t^\ast s}| \leq C^{|s|+|t|}$. This proves the assertion.
\end{proof}

Now, we will prove that $\mathcal M_Q$ being Archimedean leads to feasible solutions of the relaxations being bounded. We note that if $\mathcal M_Q$ is Archimedean, then we can express $C^2 - \sum_{i \in [2g]} x_i^\ast x_i = \sum_i f_i^\ast f_i + \sum_i \sum_{j\in \mathcal Q} g_{ij}^\ast q^{(j)} g_{ij}$ for an appropriate $C$, where $f_i \in \mathbb C[x, x^\ast]$, $g_{ij} \in \mathcal B(\mathbb C, \mathbb C^{m_j})[x, x^\ast]$ and we have removed unnecessary zeroes from the embedding into larger matrices necessary to define the quadratic module. Let us define $d_M = \max_{i,j}\{\deg{f_i}, \deg{g_{ij}}+d_j\}$, where $d_j = \lceil \operatorname{deg}(q^{(j)})/2\rceil$ for all $j \in \mathcal Q$.

\begin{lem} \label{lem:bounded-coeff}
    Let $2k \geq \max \{\deg{p^{(i)}}, \deg{q^{(j)}}: i \in \mathcal P, j \in \mathcal Q\}$, $\mathcal M_Q$ be Archimedean, and $y$ a feasible solution of problem $(\mathbf{R_{k-1+d_M}})$. Then, $|y_w^{ij}| \leq C^{|w|}$ for all $i$, $j \in [m]$ and $|w| \leq 2k$.
\end{lem}
\begin{proof}
    We write again $c = C^2 - \sum_{i=1}^{2n} x_i^\ast x_i$. If $\mathcal M_Q$ is Archimedean, then there are $f_i \in \mathbb C[x, x^\ast]$, $g_{ij} \in \mathcal B(\mathbb C, \mathbb C^{m_j})[x, x^\ast]$ such that 
    \begin{equation}
        c = \sum_{i} f_i^\ast f_i + \sum_{i} \sum_{j \in \mathcal Q} g_{ij}^\ast q^{(j)} g_{ij}.
    \end{equation}
    We will make two claims to conclude:
\begin{enumerate}
    \item []{\textbf{Claim 1}:}  Let $f \in \mathbb C[x,x^\ast]_d$ and let $M_{k+d}(y) \geq 0$. Then $M_k(f^\ast f y) \geq 0$.
    \item[] \textbf{Claim 2}: Let $ g \in \mathcal B(\mathbb C, \mathbb C^{m_j})[x, x^\ast]_d$ and let $M_{k+d}(q^{(j)}y) \geq 0$. Then $M_{k}(g^\ast q^{(j)} g y) \geq 0$.
\end{enumerate}
    Now, since $y$ is a feasible solution of of problem $(\mathbf{R_{k-1+d_M}})$, it holds that $M_{k-1+d_M}(y) \geq 0$ and that $M_{k-1+d_M-d_j}(q^{(j)}y) \geq 0$ for all $j \in \mathcal Q$. Assuming that both claims are true, we infer
    \begin{equation}
        M_{k-1}(cy) = \sum_i \underbrace{M_{k-1}(f_i^\ast f_i y)}_{\geq 0 \text{ (Claim 1)}} + \sum_{i} \sum_{j \in \mathcal Q} \underbrace{M_{k-1}(g_{ij}^\ast q^{(j)} g_{ij}y)}_{\geq 0 \text{ (Claim 2)}} \geq 0.
    \end{equation}
    Then, we can use Lemma \ref{lem:c-bounds-coefficients} to conclude that $|y_w^{ij}| \leq C^{|w|}$ for all $i$, $j \in [m]$ and $|w|\leq 2k$.

    It remains hence to prove the claims. We start with Claim 1: Let $\ket{z}  = \sum_{n \in [m]} \sum_{|w|\leq k} \ket{n} \otimes z_w^{(n)} \ket{w}$, where $z_w^{(n)} \in \mathbb C$ for all $w \in \mathcal W_k$ and $n \in [m]$. Then,
    \begin{align}
        \bra{z} M_k(f^\ast f y) \ket{z} & = \sum_{i,j \in [m]} \sum_{|v|, |w| \leq k} (z_v^{(i)})^\ast M_k^{ij}(f^\ast f y)(v,w) z_w^{(j)} \\
        &= \sum_{i,j \in [m]} \sum_{|v|, |w| \leq k}(z_v^{(i)})^\ast z_w^{(j)} \sum_{|u_1|, |u_2|\leq d} f^\ast_{u_1} f_{u_2} \underbrace{y^{ij}_{v^\ast u_1^\ast u_2 w}}_{=M_{k+d}^{ij}(y)(u_1 v, u_2 w)} \\
        &= \bra{\mu} M_{k+d}(y) \ket{\mu} \geq 0 \, ,
    \end{align}
    where $\ket{\mu} = \sum_{n \in [m]} \sum_{|v|\leq k} \sum_{|u| \leq d} \ket{n} \otimes (z^{(n)}_v f_u) \ket{uv} \in  \mathbb C^m \otimes \mathbb C^{|\mathcal W_{k+d}|}$. This proves the first claim.

    For the second claim, we can use the same $\ket{z}$ to infer that
    \begin{align}
        &\bra{z}M_k(g^\ast q^{(\ell)} g y) \ket{z} \\&= \sum_{i,j \in [m]} \sum_{|v|,|w| \leq k} \sum_{|u_1|, |u_2| \leq d} \sum_{r \leq \deg{q^{(j)}}} (z_{v}^{(i)})^\ast z_{w}^{(j)} (g_{u_1})^\ast  q_r^{(\ell)} g_{u_2} y^{ij}_{v^\ast u_1^\ast r u_2 w} \\
        &= \bra{\mu^\prime} M_{k+d}(q_i y) \ket{\mu^\prime} \geq 0 \, ,
    \end{align}
     where $\ket{\mu^\prime} = \sum_{n\in [m]} \sum_{|v| \leq k} \sum_{|u| \leq d} z^{(n)}_v g_u \otimes \ket{n} \otimes \ket{uv} \in \mathbb C^{m_j} \otimes \mathbb C^m \otimes \mathbb C^{|\mathcal W_{k+d}|}$. This proves the second claim and concludes the proof of the lemma.
    \end{proof}

    \begin{lem} \label{lem:limit-exists}
        Let $\mathcal M_Q$ be Archimedean and let $k$ be large enough. Let $p^{(k)}$ be the solution of the optimization problem \eqref{eq:problem-Rk}. Then, the optima $p^{(k)}$ form a non-decreasing bounded sequence and $\hat p = \lim_{p \to \infty} p^{(k)}$ exists.
    \end{lem}
    \begin{proof}
        Since problem \eqref{eq:problem-P} has a solution, Lemma \ref{lem:solutions-exist} implies that the problem $(\mathbf{R_{k_0}})$ is feasible. Positivity of the moment matrix of order $k+1$ implies positivity of the moment matrix of order $k$ and the same holds for the localizing matrix. Moreover, note that the objective function and the constraints indexed by $\ell \in \mathcal R$ only depend on $y_w$ with $|w| \leq \deg(p)$ and $|w| \leq \max_{\ell \in \mathcal R} \deg(r^{(\ell)})$, respectively. Therefore, we have $p^{(k)} \leq p^{(k+1)} $ for $k \geq k_0$.
        
        For $k \geq k_1$, where $k_1\geq k_0$ is large enough, and all $|w| \leq 2 k$, we can use Lemma \ref{lem:bounded-coeff} to infer in addition $|y_w^{ij}| \leq C^{|w|}$. Thus, the objective function
        \begin{equation}
            \sum_{j,\ell \in [m]} \sum_{|w|\leq 2k} \bra{j}p_w\ket{\ell}  y_w^{j\ell}
        \end{equation}
    is bounded on feasible solutions as well. As $\mathcal M_Q$ is Archimedean, the positivity domain $\mathcal S_Q$  is bounded and therefore the solution of the original problem \eqref{eq:problem-P}, which we call $p^\opt$, is also bounded. Hence, for any $k \geq k_1$ large enough
    \begin{equation}
        p^{(k_1)} \leq p^{(k)} \leq p^{(k+1)} \leq p^\opt
    \end{equation}
    and $(p^{(k)})_{k\geq k_1}$ is a bounded non-decreasing sequence. By the monotone convergence theorem, the limit $\hat p$ thus exists.
      \end{proof}

\begin{lem} \label{lem:infinite-sequence}
    Let $\mathcal M_Q$ be Archimedean and let $\hat p = \lim_{k \to \infty} p^{(k)}$ be the limit of the optimal solutions $p^{(k)}$ of the relaxations \eqref{eq:problem-Rk}. Then, there exists an infinite sequence $(\hat y)_{w \in \mathcal W_\infty}$ with values in $\mathcal B(\mathbb C^m)$ for each $w$ such that $|\hat y_w^{ij}|\leq C^{|w|}$ and such that
    \begin{align}
        &\sum_{w} \Tr[p_w^{(i)} \hat y_w^T] = \hat p \nonumber\\
         &\sum_{i=1}^m \hat y_{\emptyset}^{ii} = 1 \\
        &M_k(\hat y) \geq 0  \nonumber\\
     &M_{k-d_j}(q^{(j)}\hat y) \geq 0  \qquad \forall j \in \mathcal Q \nonumber\\
    &\sum_{w}  \Tr[r^{(\ell)}_w \hat y_w^T] \geq 0 \qquad \forall \ell \in \mathcal R \nonumber
\end{align}
for all $k$ large enough.
\end{lem}
\begin{proof}
    In the same way as Lemma 5 in  \cite{pironio2010convergent}, this statement follows from Lemmas \ref{lem:bounded-coeff} and \ref{lem:limit-exists} together with the Banach-Alaoglu theorem.
\end{proof}
     
    \begin{thm}\label{matrix-NPA-converges}
        Let $p^\opt$ be the optimal value of \eqref{eq:problem-P}. If $\mathcal M_Q$ is Archimedean, then $p^\opt = \lim_{k \to \infty} p^{(k)}$.
    \end{thm}
    \begin{proof}
        We will start by constructing a suitable Hilbert space, taking inspiration from the construction of Stinespring dilations. Define a function $\langle \cdot, \cdot \rangle: (\mathbb C^m \otimes \mathbb C[x, x^\ast]) \times  (\mathbb C^m \otimes \mathbb C[x, x^\ast]) \to \mathbb C$ by 
        \begin{equation}
        \langle \varphi \otimes s , \eta \otimes t \rangle = \bra{\varphi} L_{\hat y}(s^\ast t) \ket{\eta}    
        \end{equation}
        for all $\ket{\eta}$, $\ket{\varphi} \in \mathbb C^m$ and $s$, $t \in \mathbb C[x, x^\ast]$, where $\hat y$ is the infinite sequence attaining $\hat p$ from Lemma \ref{lem:infinite-sequence}. By extending this function bi-linearly, we have defined a conjugate symmetric bilinear function on $\mathbb C^m \otimes \mathbb C[x, x^\ast]$. It is also positive semidefinite since $M_k(\hat y)$ is for $k$ large enough,
        because for $d_s = \max_{i \in [m]} s^{(i)}$,
        \begin{align}
            \sum_{i,j \in [m]} \bra{i}L_{\hat y}((s^{(i)})^\ast s^{(j)}) \ket{j} &= \sum_{i,j \in [m]}  \sum_{v, w \in \mathcal W_{d_s}} (s^{(i)}_v)^\ast  \bra{i} L_{\hat y}(v^\ast w) \ket{j} s^{(j)}_w  \\ &= \sum_{v, w \in \mathcal W_{d_s}} (s^{(i)}_v)^\ast  \bra{i} M_{d_s}(\hat y)(v, w)  \ket{j} s^{(j)}_w \label{eq:positivity-for-hat-y} \\
            &= \bra{s} M_{d_s}(\hat y) \ket{s}
        \end{align}
        with $\ket{s} = \sum_{i \in [m]} \sum_{w \in \mathcal W_{d_s}} s^{(i)}_w \ket{i} \otimes \ket{w}$. Therefore, it is a positive semidefinite bilinear form and hence satisfies the Cauchy-Schwarz inequality. Thus,
        \begin{equation}
            \mathcal N \coloneq \{ b \in \mathbb C^m \otimes \mathbb C[x, x^\ast]\, : \, \langle b, b \rangle = 0 \}
        \end{equation}
        is a linear subspace of $\mathbb C^m \otimes \mathbb C[x, x^\ast]$. Consequently, $\langle \cdot, \cdot \rangle$ is an inner product on $\mathbb C^m \otimes \mathbb C[x, x^\ast] \setminus \mathcal N$. 

            We can now define for any $x_i$, $i \in [2g]$, a linear function $\pi(x_i):\mathbb C^m \otimes \mathbb C[x, x^\ast] \to \mathbb C^m \otimes \mathbb C[x, x^\ast]$ as $\pi(x_i) \sum_{i} \ket{\psi_i} \otimes p^{(i)} = \sum_i \ket{\psi_i} \otimes  (x_ip^{(i)})$. As $c = C^2 - \sum_{i=1}^{2n} x_i^\ast x_i \in \mathcal M_Q$, it follows that $M_{k}(c\hat y) \geq 0$ for any sufficiently large $k$ as in Lemma \ref{lem:bounded-coeff}. Therefore, for any $s$, $t \in \mathcal W_k$,
            \begin{equation}
                0 \leq M_k(c\hat y) = C^2   M_k(\hat y) - \sum_{i \in [2g]}  M_k(x_i^\ast x_i\hat y)\, .
            \end{equation}
            As $M_k(x_i^\ast x_i\hat y) \geq 0$ for all $i \in [2g]$ (see Claim 1 in Lemma \ref{lem:bounded-coeff}), we obtain
            \begin{equation}\label{eq:bound-by-moment-matrix}
                0 \leq M_k(x_i^\ast x_i\hat y)\leq C^2  M_k(\hat y).
            \end{equation}
            Hence, for any $k$ large enough, any $b = \sum_{j\in [m]} \ket{j} \otimes b^{(j)} \in \mathcal N$, where $b^{(j)} \in \mathbb C[x,x^\ast]$ for all $j \in [m]$, and any $i \in [2g]$,
            \begin{align}
                \langle \pi(x_i) b, \pi(x_i) b \rangle &= \sum_{j, \ell \in [m]}  \bra{j} L_{\hat y}(b_j^\ast x_i^\ast x_i b_\ell) \ket{\ell} \nonumber\\ & 
                =\sum_{j, \ell \in [m]} 
                \sum_{v, w \in \mathcal W_k} (b_v^{(j)})^\ast b_w^{(\ell)}\bra{j} \underbrace{L_{\hat y}(v^\ast x_i^\ast x_i w)}_{=M_k(x_i^\ast x_i \hat y)(v,w)}\ket{\ell} \nonumber \\
                & = \bra{b} M_k(x_i^\ast x_i \hat y) \ket{b} \nonumber \\
                &\leq C^2 \bra{b} M_k( \hat y) \ket{b} \nonumber
                \\ & =C^2 \langle b, b \rangle \label{eq:bounded-Xi} \,,
            \end{align}
            where we have used  \eqref{eq:bound-by-moment-matrix} and $\ket{b}\coloneq \sum_{\ell \in [m]} \sum_{w \in \mathcal W_k} b_w^{(\ell)} \ket{\ell} \otimes \ket{w}$
        Hence, $\pi(x_i)$ leaves $\mathcal N$ invariant. Therefore, we can define $X_i$ on the inner product space $\mathbb C^m \otimes \mathbb C[x, x^\ast] \setminus \mathcal N$ as
            \begin{equation}
                X_i: b+\mathcal N  \mapsto \pi(x_i) b + \mathcal N
            \end{equation}
            for all $i \in [2g]$ and $b \in \mathbb C^m \otimes \mathbb C[x, x^\ast]$. Moreover, $X_i$ is bounded for all $i \in [2g]$ by the calculation in  \eqref{eq:bounded-Xi}. Hence, we can use the bounded linear extension theorem (or BLT theorem) to obtain a unique bounded operator on the Hilbert space
        \begin{equation}
            \mathcal K = \overline{\mathbb C^m \otimes \mathbb C[x, x^\ast] \setminus \mathcal N}
        \end{equation}
        which we also denote by  $X_i$. We can verify that for $i \in [g]$, it holds that $X_i^\ast = X_{i+g}$, since for any $s = \sum_{j \in [m]} \ket{j} \otimes s^{(j)} \in \mathbb C^m \otimes \mathbb C[x, x^\ast]$, $t= \sum_{j \in [m]} \ket{j} \otimes t^{(j)} \in \mathbb C^m \otimes \mathbb C[x, x^\ast]$, where $s^{(i)}$, $t^{(j)} \in \mathbb C[x, x^\ast]$ for all $i$, $j \in [m]$,
        \begin{align}
            \langle s, X_i^\ast t \rangle &= \langle X_i s, t \rangle \\ &= \sum_{j, \ell \in [m]} \bra{j} L_{\hat y}((s^{(j)})^\ast x_i^\ast t^{(\ell)})\ket{\ell}\\ &=  \sum_{j, \ell \in [m]} \bra{j} L_{\hat y}((s^{(j)})^\ast x_{i+g} t^{(\ell)})\ket{\ell}\\ & = \langle s, X_{i+g} t\rangle \,.
        \end{align}
        We can also define the inclusion of $\mathbb C^m$ into $\mathcal K$ as $V \ket{\psi} = \ket{\psi} \otimes \iota + \mathcal N$, where $\iota$ is the identity polynomial. We can verify that it is a contraction, as 
        \begin{align}
            \langle V \ket{\psi}, V\ket{\psi} \rangle &= \langle \psi \otimes \iota, \psi \otimes \iota \rangle = \bra{\psi} L_{\hat y}(\emptyset) \ket{\psi} \\
            & \implies \|V\|_{\infty} = \|L_{\hat y}(\emptyset)\|_{\infty} \leq 1
        \end{align}
        Now observe that if we define $\pi(s) = s(X)$ for any word $s \in \mathbb C[x, x^\ast]$,
        \begin{equation}
            \bra{\varphi} V^\ast \pi(w) V \ket{\psi}  = \bra{\varphi} L_{\hat y}(w) \ket{\psi} \, .
        \end{equation}
        Now, we define a linear operator
        \begin{align}
            \widetilde V: \mathbb C^m & \to \mathbb C^m \otimes \mathcal K \\
            \ket{i} &\mapsto \ket{i} \otimes V\ket{i} \quad \forall i \in [m]
        \end{align}
        and a map
        \begin{align}
            \widetilde \pi : \mathbb C[x, x^\ast] &\to \mathcal B(\mathbb C^m \otimes \mathcal K) \\
            s & \mapsto \mathds{1}_m \otimes \pi(s).
        \end{align}
        Finally, we can choose 
        \begin{equation}
            \ket{\psi} \coloneq \widetilde V\sum_{i \in [m]} \ket{i} = \sum_i \ket{i} \otimes \ket{\psi_i} \,,
        \end{equation}
        where $\ket{\psi_i} = V\ket{i}$. It holds that
        \begin{align}
            \braket{\psi|\psi} = \sum_{i \in [m]} \bra{i} L_{\hat y}(\emptyset) \ket{i} = \bra{i} \hat y^{ii}_{\emptyset} \ket{i} = 1
        \end{align}
        by Lemma \ref{lem:infinite-sequence}.

        It remains to check that the state $\ket{\psi}$, the Hilbert space $\mathcal K$ and $X$ as defined above form a feasible solution for \eqref{eq:problem-P} with value $\hat p$ (we only need to prove that $\hat p \geq p^\opt$, because $p^{(k)} \leq p^\opt$ always holds by Lemma \ref{lem:solutions-exist} and  \eqref{eq:same-objective-function}). For the objective function, we find that  
        \begin{align}
            \bra{\psi} p(X) \ket{\psi} & = \sum_{|w| \leq \deg{p}} \bra{\psi} p_w \otimes w(X) \ket{\psi} \\
            & = \sum_{|w|\leq \deg{p}} \sum_{j, \ell \in [m]} \bra{j} p_w \ket{\ell} \bra{\psi_j} w(X) \ket{\psi_{\ell}} \, .
        \end{align}
        Moreover, 
        \begin{align}
            \bra{\psi_j} w(X) \ket{\psi_\ell} &= \bra{j} V^\ast \pi(w) V \ket{\ell} \\
            & = \bra{j} L_{\hat y}(w) \ket{\ell}\\& = \hat y_w^{j\ell} \, .
        \end{align}
        Hence, we conclude that 
        \begin{equation}
            \bra{\psi} p(X) \ket{\psi} = \hat p
        \end{equation}
        as desired.

        Likewise, for $j \in \mathcal Q$, let $\ket{z} = \sum_{i \in [m_j]} \sum_{s \in [m]} \sum_{|u| \leq K } \ket{i} \otimes \ket{s} \otimes z^{(i,s)}_u \ket{u} \in \mathbb C^{m_j} \otimes  \mathbb C^m \otimes \mathbb C[x,x^\ast]$, where $z^{(i,s)}_u \in \mathbb C$ for all $i \in [m_j]$, $s \in [m]$, $|u| \leq K$ and $K = \max_{i \in [m_j], s \in [m]} \deg(z^{(i,s)})$. We infer that
        \begin{align}
                    & \bra{z}q^{(j)}(X) \ket{z} \\ &= \sum_{i,\ell \in [m_j]} \sum_{s,t \in [m]} \sum_{|u|,|v| \leq K}  \sum_{|w|\leq \deg(q^{(j)})} \bra{i}q_w\ket{\ell} (z_u^{(i,s)})^\ast z_v^{(\ell,t)} \langle s \otimes u, t \otimes wv \rangle\\
                    &=\sum_{i,\ell \in [m_j]} \sum_{s,t \in [m]} \sum_{|u|,|v| \leq K}  \sum_{|w|\leq \deg(q^{(j)})} \bra{i}q_w\ket{\ell} (z_u^{(i,s)})^\ast z_v^{(\ell,t)}\bra{s} L_{\hat y}(u^\ast w v) \ket{t} \\
                   &=\sum_{i,\ell \in [m_j]} \sum_{s,t \in [m]} \sum_{|u|,|v| \leq K} (z_u^{(i,s)})^\ast z_v^{(\ell,t)} (\bra{i} \otimes \bra{s}) M_k(q\hat y)(u,v) (\ket{j} \otimes \ket{t})\\
                   &=\bra{z} M_k(q\hat y)\ket{z} \geq 0
        \end{align}
        for all $j \in \mathcal Q$ as $M_k(q\hat y) \geq 0$ for $k \geq K$ large enough by Lemma \ref{lem:infinite-sequence}. As it is enough to check positivity on a dense subset, we conclude that $q^{(j)}(X) \geq 0$ for all $j \in \mathcal Q$.

        Finally, for the constraints indexed by $\ell \in \mathcal R$,
        \begin{equation}
            \bra{\psi}r^{(\ell)}(X) \ket{\psi} =   \sum_{|w| \leq \deg(r^{(\ell)})} \sum_{i, j \in [m]} \bra{i} r_w^{(\ell)} \ket{j} \underbrace{\bra{\psi_i} w(X) \ket{\psi_j}}_{\hat y_w^{ij}}\, .
        \end{equation}
    Thus, using Lemma \ref{lem:infinite-sequence} we obtain for all $\ell \in \mathcal R$ and $k \geq \deg(r^{(\ell)})$
    \begin{equation}
    \bra{\psi} r^{(\ell)}(X) \ket{\psi}=  \sum_{|w| \leq \deg(r^{(\ell)})} \Tr[r^{(\ell)}_w  \hat y_w^T] \geq 0 \,. 
    \end{equation}
    All things considered, we have thus shown that $\hat p = p^\opt$, which concludes the proof.
    \end{proof}

\subsection{Sufficient criterion for optimality}

Moreover, we can detect optimality and extract optimizers in very similar way as for the original NPA hierarchy. The following extends Theorem 2 of \cite{pironio2010convergent} and the proof is a straightforward adaptation of the proof of \cite[Theorem 2]{pironio2010convergent}. For convenience, we give it here nonetheless.
\begin{thm}
    Let $k \geq k_0$ and assume that the optimal solution $y^{(k)}$ of \eqref{eq:problem-Rk} satisfies
    \begin{equation} \label{eq:rank-equal}
       \rank M_{k}(y^{(k)}) = \rank M_{k-d}(y^{(k)}),
    \end{equation}
    where $d = \max_{j \in\mathcal Q} d_j$ and $d_j = \lceil \deg(q^{(j)})/2\rceil$ for all $j \in \mathcal Q$. Then, the value $p^{(k)}$ of \eqref{eq:problem-Rk} is equal to the value $p^\opt$ of \eqref{eq:problem-P}. Moreover, there exists an optimizer of $(\mathcal H, X, \ket{\psi})$ of \eqref{eq:problem-P} with $\dim \mathcal H =  \rank M_{k-d}(y^{(k)})$.
 \end{thm}

\begin{proof}
    We only need to show that $p^{(k)}\geq p^\opt$. We will achieve this by constructing an explicit solution to \eqref{eq:problem-P} with value $p^{(k)}$ as follows: Let $b=\rank M_{k}(y^{(k)}) =   \rank M_{k-d}(y^{(k)})$. As $ M_{k}(y^{(k)}) \geq 0$, it can be written as a Gram matrix. Therefore, there exist vectors $\ket{w^{(i)}} \in \mathbb C^{b}$ for all $i \in [m]$, $|w| \leq k$ such that 
    \begin{equation}
        (y^{(k)})^{ij}_{v^\ast w}=\braket{v^{(i)}|w^{(j)}} \qquad \forall i,\, j \in [m], ~\forall |v|,\, |w| \leq k \,.
    \end{equation}
    We define the Hilbert space $\mathcal H \coloneq \{\ket{w^{(i)}}: i \in [m], w \in \mathcal W_k\}$. It follows that $\dim \mathcal H = b$ (see \cite[Theorem 7.2.10]{horn2012matrix}).
    As $M_{k-d}(y^{(k)})$ is a submatrix of $M_{k}(y^{(k)})$ which has the same rank by  \eqref{eq:rank-equal}, \cite[Theorem 7.2.10]{horn2012matrix} implies that
    \begin{equation} \label{eq:Gram-Hilbert-space}
        \mathcal H = \{\ket{w^{(i)}}: i \in [m], w \in \mathcal W_k\}= \{\ket{w^{(i)}}: i \in [m], w \in \mathcal W_{d-k}\} \,.
    \end{equation}
    On any linearly independent set of vectors $\ket{w^{(j)}} \in \mathcal H$ with $|w| \leq k-1$, we can define linear operators $X_i$, $i \in [2g]$ as
    \begin{equation}
        X_i \ket{w^{(j)}} = \ket{(x_iw)^{(j)}}
    \end{equation}
    and extending by linearity. These are well-defined operators on $\mathcal H$, since such vectors span the whole Hilbert space by \eqref{eq:Gram-Hilbert-space}, as $k-d \leq k-1$. The definition does not depend on the subset of vectors $\ket{w^{(j)}}$ that we choose, since for $\ket{y} \in \mathcal H$ and two different compositions $\ket{y} = \sum_{|w|\leq k-1} \sum_{i \in [m]} a^{(i)}_w \ket{w^{(i)}}$ and $\ket{y} = \sum_{|w|\leq k-1} \sum_{i \in [m]} c^{(i)}_w \ket{w^{(i)}}$ with $a^{(i)}_w$, $c^{(i)}_w \in \mathbb C$ for all $i \in [m]$, $w \in \mathcal W_k$, we claim that 
    \begin{equation}
        \sum_{|w|\leq k-1} \sum_{j \in [m]} c^{(j)}_w \ket{(x_iw)^{(j)}} = \sum_{|w|\leq k-1} \sum_{j \in [m]} c^{(j)}_w \ket{(x_iw)^{(j)}}
    \end{equation}
    for all $i \in [2g]$. 
    
    This claim is true, because for any $\ket{v^{(\ell)}}$ where $|v| \leq k-d \leq k-1$ and $\ell \in [m]$, we infer
    \begin{align}
        &\bra{v^{(\ell)}} \sum_{|w|\leq k-1}\sum_{j \in [m]} (a^{(j)}_w - c^{(j)}_w) \ket{(x_iw)^{(j)}} \\
        &= \sum_{|w|\leq k-1}\sum_{j \in [m]} (a^{(j)}_w - c^{(j)}_w) (y^{(k)})^{\ell j}_{v^\ast x_i w} \\
        &= \bra{(x_i^\ast v)^{(\ell)}} \sum_{|w|\leq k-1}\sum_{j \in [m]} (a^{(j)}_w - c^{(j)}_w) \ket{w^{(j)}} \\
        &=0 \,,
    \end{align}
    where in the last line we have used that both sums are decompositions of $\ket{y}$. Since such vectors $\ket{v^{(\ell)}}$ span $\mathcal H$, the claim follows. We can also see that $X_i^\ast = X_{g+i}$ for all $i \in [g]$, as
    \begin{equation}
        \bra{v^{(\ell)}} X_i^\ast \ket{w^{(j)}} = \braket{(x_iv)^{(\ell)}|w^{(j)}}=(y^{(k)})^{\ell j}_{v^\ast x_i^\ast w} = (y^{(k)})^{\ell j}_{v^\ast x_{g+i} w} = \bra{v^{(\ell)}} X_{i+g} \ket{w^{(j)}}
    \end{equation}
    for all $j$, $\ell \in [m]$, $v$, $w \in \mathcal W_{k-1}$.

Having defined $\mathcal H$ and $X$, it remains to fix $\ket{\psi}$. We define $\ket{\psi_i} \coloneq \ket{\emptyset^{(i)}}$ for all $i \in [m]$ and $\ket{\psi}:= \sum_{i \in [m]} \ket{i} \otimes \ket{\psi_i}$. The vector $\ket{\psi}$ is normalized, since 
\begin{equation}
    \braket{\psi|\psi} = \sum_{i \in [m]} \braket{\psi_i | \psi_i} = \sum_{i \in [m]} (y^{(k)})_\emptyset^{ii} =1.
\end{equation}

We are now left with checking that $(\mathcal H, X, \ket{\psi})$ satisfies the constraints in \eqref{eq:problem-P} and has value $p^{(k)}$. We observe that for $w \in \mathcal W_{2k}$ and $w_1$, $w_2 \in \mathcal W_k$ such that $w=w_1w_2$, for all $i$, $j \in [m]$ it holds that
\begin{equation}
    \bra{\psi_i} w(X) \ket{\psi_j} =  \braket{(w_1^\ast)^{(i)} | w_2^{(j)}} = (y^{(k)})^{ij}_w \,.
\end{equation}
Thus,
\begin{align}
    \bra{\psi} p(X) \ket{\psi} &= \sum_{i,j \in [m]} \sum_{|w| \leq \deg(p)} \bra{i} p_w \ket{j} \bra{\psi_i} w(X) \ket{\psi_j}\\ & = \sum_{i,j \in [m]} \sum_{|w| \leq 2k} \bra{i} p_w \ket{j} (y^{(k)})_w^{ij} \\ &= p^{(k)}.
\end{align}
For the constraints, for $j \in \mathcal Q$,  $q^{(j)} \geq 0$ is equivalent to checking positive semidefiniteness of the matrix $M \in \mathcal B(\mathbb C^m \otimes \mathbb C^{|\mathcal W_{k-d}|})$ with entries 
\begin{equation}
M^{i\ell}(v,w) = \bra{i} \otimes \bra{v^{(i)}} q^{(j)}(X) \ket{\ell} \otimes \ket{w^{(\ell)}} = \sum_{|s|\leq \deg(q^{(j)})} \bra{i}q_s\ket{\ell} (y^{(k)})^{i\ell}_{v^\ast s w}  
\end{equation}
for $i$, $\ell \in [m]$ and $v$, $w \in \mathcal W_{k-d}$. Thus, we see that $M=M_{k-d}(q^{(j)} y^{(k)})$, which is the matrix $M_{k-d_j}(q^{(j)} y^{(k)})$ restricted to a subspace and therefore positive semidefinite, as $y^{(k)}$ is a solution of \eqref{eq:problem-Rk}. Finally, for $\ell \in \mathcal R$, 
\begin{align}
    \bra{\psi} r^{(\ell)}(X) \ket{\psi} &= \sum_{|w| \leq \deg{(r^{(\ell)})}} \sum_{i,j \in [m]} \bra{i} r^{(\ell)}_w \ket{j} \bra{\psi_i} w(X)\ket{\psi_j}\\& = \sum_{|w| \leq 2k} \sum_{i,j \in [m]} \bra{i} r^{(\ell)}_w \ket{j} (y^{(k)})^{ij}_{w} \\ &\geq 0\, ,
\end{align}
as $y^{(k)}$ is a solution of \eqref{eq:problem-Rk}.
\end{proof}

\subsection{Connection to the Positivstellensatz by Helton and McCullough}

In this section, we will discuss in how far the NPA hierarchy with matrix-valued polynomials we have described is implied by the Positivstellensatz for matrix-valued polynomial in \cite{helton2004positivstellensatz}, as mentioned in \cite{pironio2010convergent}. The relation between the two that we present here is very similar to link between the Positivstellensatz and the NPA hierarchy described in Section 3.4 of \cite{pironio2010convergent}, see also \cite{doherty2008quantum}.

We consider the following optimization problem for $k \geq k_0$ (see \eqref{eq:problem-Rk}):

\begin{align*}
    \mathrm{maximize} \quad & \quad  \lambda  \\
    \mathrm{such~that} \quad & \quad p-\lambda =  \sum_{i} f_i^\ast f_i + \sum_{i} \sum_{j \in \mathcal Q} g_{ij}^\ast q^{(j)} g_{ij} + \sum_{\ell \in \mathcal R} c_\ell r^{(\ell)} \\
    \quad & \quad \max_i \deg(f_i) \leq k  \tag{$\mathbf{D_k}$} \label{eq:problem-Dk} \\
    \quad & \quad \max_i \deg(g_{ij}) \leq k-d_j \quad \forall j \in \mathcal Q\\
    \quad & \quad c_{\ell} \geq 0 \quad \forall \ell \in \mathcal R\\
    \quad & \quad f_i \in \mathcal B(\mathbb C^m, \mathbb C)[x, x^\ast] \quad \forall i \\
    \quad & \quad g_{ij} \in \mathcal B(\mathbb C^m, \mathbb C^{m_j})[x, x^\ast] \quad \forall i, \, \forall j \in \mathcal Q \\
    \quad & \quad c_{\ell} \in \mathbb R \quad \forall \ell \in \mathcal R
\end{align*}
Here, $d_j = \lceil \deg(d^{(j)})/2\rceil$. Let $\lambda^{(k)}$ be the optimal value of \eqref{eq:problem-Dk}. Using a similar argument as in \cite[Appendix B]{pironio2010convergent}, the problem can be formulated as an SDP, and it can be shown to be the dual SDP of \eqref{eq:problem-Rk}. Therefore, it holds that
\begin{equation}
    \lambda^{(k)} \leq p^{(k)} \qquad \forall k \geq k_0 \,.
\end{equation}
We can also see this directly: Let us consider a feasible solution to \eqref{eq:problem-Dk}. Our aim is to show for $y$ a feasible solution of \eqref{eq:problem-Rk} that 
\begin{equation}
\sum_{|w| \leq 2k} \Tr[p_w y_w^T] - \lambda^{(k)} \geq 0
\end{equation}
where we will make use of the fact that
\begin{equation}
    \bra{\Omega} L_y(p - \lambda^{(k)}) \ket{\Omega} = \sum_{|w| \leq 2k} \Tr[p_w y_w^T] - \lambda^{(k)} \, ,
\end{equation}
where $\ket{\Omega} = \sum_{i \in [m]} \ket{i} \otimes \ket{i}$ is an unnormalized maximally entangled state. Using that $L_y$ is linear, we obtain that
\begin{align}
    &\bra{\Omega} L_y(p - \lambda^{(k)}) \ket{\Omega} \\& = \sum_{i} \bra{\Omega} L_y(f_i^\ast f_i)\ket{\Omega} + \sum_{i} \sum_{j \in \mathcal Q} \bra{\Omega} L_y(g_{ij}^\ast q^{(j)} g_{ij})\ket{\Omega} + \sum_{\ell \in \mathcal R} c_\ell \bra{\Omega} L_y(r^{(\ell)}) \ket{\Omega}.
\end{align}
It is sufficient to show that $\bra{\Omega} L_y(f_i^\ast f_i)\ket{\Omega} \geq 0$ for all $i$, $ \bra{\Omega}L_y(g_{ij}^\ast q^{(j)} g_{ij})\ket{\Omega} \geq 0$ for all $j \in \mathcal Q$ and all $i$, and $ \bra{\Omega} L_y(r^{(\ell)}) \ket{\Omega} \geq 0$ for all $\ell \in \mathcal R$.

For the first assertion, we find that 
\begin{align}
    \bra{\Omega} L_y(f_i^\ast f_i)\ket{\Omega} & = \sum_{|v|, |w| \leq k} \bra{\Omega}(f_i)_v^\ast (f_i)_w \otimes y_{v^\ast w}\ket{\Omega} \\
    & = \sum_{|v|, |w| \leq k}  (f_i)_w \underbrace{y_{v^\ast w}^T}_{M_k(y)^T(v,w)} (f_i)_v^\ast \\
    & = \bra{f_i} M_k(y)^T \ket{f_i} \, ,
\end{align}
where $\ket{f_i} = \sum_{|w| \leq k} (f_i)_w^\ast \otimes \ket{w} \in \mathbb C^m \otimes \mathbb  C^{|\mathcal W_k|}$. Since $y$ is feasible, $M_k(y) \geq 0$ holds and hence $\bra{\Omega} L_y(f_i^\ast f_i)\ket{\Omega} \geq 0$ as claimed.

For the second assertion, it holds that
\begin{align}
    &\bra{\Omega}L_y(g_{ij}^\ast q^{(j)} g_{ij})\ket{\Omega} \\&= \sum_{|v|, |w| \leq k-d_j}  \bra{\Omega}(g_{ij})^\ast_v \otimes \mathds{1} \underbrace{\sum_{|s| \leq \deg(q^{(j)})} q^{(j)}_s \otimes y_{v^\ast s w}}_{M_{k-d_j}(q^{(j)}y)(v,w)} \,(g_{ij})_w \otimes \mathds{1}\ket{\Omega} \\
    &= \bra{g_{ij}} M_{k-d_j}(q^{(j)}y) \ket{g_{ij}} \, ,
\end{align}
where $\ket{g_{ij}} = \sum_{|w| \leq k-d_j} ((g_{ij})_w \otimes \mathds{1}\ket{\Omega}) \otimes \ket{w} \in \mathbb C^{m_j} \otimes \mathbb C^m \otimes \mathbb  C^{|\mathcal W_{k-d_j}|}$. Since $y$ is feasible, $M_{k-d_j}(q^{(j)}y) \geq 0$ holds and hence $\bra{\Omega}L_y(g_{ij}^\ast q^{(j)} g_{ij})\ket{\Omega} \geq 0$ as claimed.

For the third assertion, 
\begin{equation}
    \bra{\Omega} L_y(r^{(\ell)}) \ket{\Omega} = \sum_{|w| \leq 2k} \Tr[r^{(\ell)}_w y_w^T] \, ,
\end{equation}
which is positive because $y$ is feasible.

Hence, we have shown that $\lambda^{(k)} \leq p^{(k)} \leq p^\opt$. To see that $\lambda^{(k)}$ also converges to $p^\opt$, we see that if $\mathcal R = \emptyset$, for any $\varepsilon >0$ it holds that $p-(p^\opt - \varepsilon)$ is strictly positive on $\mathcal S_Q$.  Hence, the Positivstellensatz in \cite{helton2004positivstellensatz}\footnote{To be precise, the result in \cite{helton2004positivstellensatz} only holds for matrix-valued non-commutative polynomials in which the matrices have real entries. It is however folklore that the same proof extends to matrices with complex entries. For scalar-valued non-commutative polynomials, this was explicitly shown in \cite{doherty2008quantum}.} yields that
\begin{equation}
    p-p^\opt + \varepsilon =  \sum_{i} f_i^\ast f_i + \sum_{i} \sum_{j \in \mathcal Q} g_{ij}^\ast q^{(j)} g_{ij}
\end{equation}
for suitable polynomials $f_i \in \mathcal B(\mathbb C^m, \mathbb C)[x, x^\ast]$, $g_{ij} \in \mathcal B(\mathbb C^m, \mathbb C^{m_j})[x, x^\ast]$. Hence, for $k \geq \max_{i,j}\{\deg(f_i), \deg(g_{ij})+d_j\}$, it holds that
\begin{equation}
    p^\opt - \varepsilon \leq \lambda^{(k)} \leq p^{(k)} \leq p^\opt
\end{equation}
and $\lambda^{(k)}$ converges to $p^\opt$ as $\varepsilon$ was arbitrary. In the case of $\mathcal R \neq \emptyset$, it is very plausible that a suitable Positivstellensatz should exist, but we are not aware of a result that is applicable to our situation.

\section{\texorpdfstring{On the $C^\ast$-algebra in one-sided DIQKD}
{}}\label{sec:more_efficient_presentation}

In Section \ref{subsec:npa_hierarchy_with_matrix_constraints} we sketched the resulting relations in a one-sided DIQKD protocol for the algebraic constraints, but we did not verify that they actually suffice in a mathematically rigorous  sense. This is what we rectify in the present section. Let us recall the relations in the following
\begin{enumerate}
    \item[(R1)] \emph{$*$-structure:}
    \begin{align}
        \eta_{i,j,b\vert y}^\ast = \eta_{j,i,b\vert y}
        \quad\text{for all }i \in [m], j \in [m], b \in \mathcal B, y \in \mathcal Y.
    \end{align}
    \item[(R2)] \emph{Matrix-unit relations within a fixed input-output $(b,y)$:}
    \begin{align}
        \eta_{i,j,b\vert y}\,\eta_{k,\ell,b\vert y}
        = \delta_{j,k}\,\eta_{i,\ell,b\vert y}
        \quad\text{for all }i,j,k,\ell \in [m], b \in \mathcal B, y \in \mathcal Y.
    \end{align}
    \item[(R3)] \emph{Orthogonality of different outcomes for the same input on Bob's side $y$:}
    \begin{align}
        \eta_{i,j,b\vert y}\,\eta_{k,\ell,b^\prime\vert y} = 0
        \quad\text{for all }i,j,k,\ell \in [m],y \in \mathcal Y\text{ and }b\neq b' \in \mathcal B.
    \end{align}
    \item[(R4)] \emph{Completeness for each question $y$:}
    \begin{align}
        \sum_{i=1}^m \sum_{b\in \mathcal B} \eta_{i,i,b\vert y} = \mathds{1}
        \quad\text{for all }y\in \mathcal Y.
    \end{align}
    \item[(R5)] \emph{Independence of the input label for Alice's part:}
    For all $i,j \in [m]$ and all $y,y'\in \mathcal Y$,
    \begin{align}
        \sum_{b\in \mathcal B} \eta_{i,j,b\vert y} = \sum_{b\in \mathcal B} \eta_{i,j,b\vert y'}.
    \end{align}
\end{enumerate}
from Section \ref{subsec:npa_hierarchy_with_matrix_constraints}. Indeed, the claim is that these relations generate a $C^\ast$-algebra, which is isomorphic to $\mathbb M_m(\mathbb{C}) \otimes \mathfrak{B}$, if $\mathfrak{B}$ is the algebra generated by Bob's POVM elements. To justify this argument we show in the following an isomorphism of $C^\ast$-algebras.

Let $\mathfrak{C}$ be the universal $C^\ast$-algebra generated by symbols $\{\eta_{i,j,b\vert y}\}_{i \in [m],j \in [m],b \in \mathcal B,y \in \mathcal Y}$ subject to the relations above. The universal $C^\ast$-algebra exists because the relations bound the generators in norm: by (R1) and (R2) the element $\eta_{i,j,b\vert y}^\ast \eta_{i,j,b\vert y} = \eta_{j,i,b\vert y} \eta_{i,j,b\vert y} = \eta_{j,j,b\vert y}$ is a self-adjoint idempotent, so $\lVert \eta_{i,j,b\vert y} \rVert_\infty \leq 1$ in every representation, and the supremum over all admissible representations defines a $C^\ast$-norm (see \cite[Section~II.8.3]{Blackadar2006} and \cite[Chapter~3]{Loring1997} for the general theory of universal $C^\ast$-algebras defined by generators and norm-bounding relations).

Since $\mathbb M_m(\mathbb{C})$ is finite-dimensional and hence nuclear, there is only one
$C^\ast$-norm on the algebraic tensor product $\mathbb M_m(\mathbb{C}) \odot \mathfrak{A}$ for any
$C^\ast$-algebra $\mathfrak{A}$; nuclearity of type~I (in particular finite-dimensional)
$C^\ast$-algebras goes back to Takesaki \cite{Takesaki1964}, see also
\cite[Section~II.9.4]{Blackadar2006} and \cite[Chapter~3]{BrownOzawa2008}. We therefore write
$\mathbb M_m(\mathbb{C}) \otimes \mathfrak{A}$ without specifying a tensor norm and use the
identification $\mathbb M_m(\mathbb{C}) \otimes \mathfrak{A} \cong \mathbb M_m(\mathfrak{A})$ freely. 

\begin{lem}\label{lem:matrix_units}
In $\mathfrak{C}$ let
\begin{align}
    \theta_{i,j} \coloneqq \sum_{b \in \mathcal B} \eta_{i,j,b\vert y}, \qquad
    P(b\vert y) \coloneqq \eta_{1,1,b\vert y},
\end{align}
where $\theta_{i,j}$ is independent of $y$ by \textnormal{(R5)}. Then
\begin{enumerate}
    \item[(1)] $\{\theta_{i,j}\}_{i,j=1}^m$ is a system of matrix units, that is $\theta_{i,j}^\ast = \theta_{j,i}$ and
          $\theta_{i,j}\theta_{k,\ell} = \delta_{j,k}\theta_{i,\ell}$, and $\sum_{i=1}^m \theta_{i,i} = \mathds{1}$;
    \item[(2)] $\eta_{i,j,b\vert y} = \theta_{i,1}\, P(b\vert y)\, \theta_{1,j}$ for all $i,j \in [m], b \in \mathcal B, y \in \mathcal Y$;
    \item[(3)] each $P(b\vert y)$ is a projection with $\sum_{b \in \mathcal B} P(b\vert y) = \theta_{1,1}$, and the
          corner $\theta_{1,1}\mathfrak{C}\theta_{1,1}$ is the $C^\ast$-subalgebra of $\mathfrak{C}$
          generated by $\{P(b\vert y)\}_{b \in \mathcal B,y \in \mathcal Y}$.
\end{enumerate}
\end{lem}

\begin{proof}
(1) By (R1),
\begin{align}
    \theta_{i,j}^\ast = \sum_{b \in \mathcal B} \eta_{i,j,b\vert y}^\ast = \sum_{b \in \mathcal B} \eta_{j,i,b\vert y} = \theta_{j,i}.
\end{align}
For the product, we may compute both factors with the same input $y$, which is allowed by (R5). All
mixed terms vanish by (R3), so only the diagonal terms in the double sum survive, and (R2) gives
\begin{align}
    \theta_{i,j} \theta_{k,\ell}
    = \sum_{b,b' \in \mathcal B} \eta_{i,j,b\vert y}\, \eta_{k,\ell,b'\vert y}
    = \sum_{b\in \mathcal B} \eta_{i,j,b\vert y}\, \eta_{k,\ell,b\vert y}
    = \delta_{j,k} \sum_{b \in \mathcal B} \eta_{i,\ell,b\vert y}
    = \delta_{j,k}\, \theta_{i,\ell}.
\end{align}
Finally $\sum_{i=1}^m \theta_{i,i} = \sum_{i=1}^m \sum_{b \in \mathcal B} \eta_{i,i,b\vert y} = \mathds{1}$ by (R4).

(2) Fix $i,j \in [m]$, $b \in \mathcal B$, $y \in \mathcal Y$ and expand $\theta_{i,1}$ and $\theta_{1,j}$ with the same input $y$. By (R3) only the term
$b' = b$ survives, and by (R2),
\begin{align}
    \theta_{i,1}\, P(b\vert y)
    = \sum_{b' \in \mathcal B} \eta_{i,1,b'\vert y}\, \eta_{1,1,b\vert y}
    = \eta_{i,1,b\vert y}, \qquad
    \eta_{i,1,b\vert y}\, \theta_{1,j}
    = \sum_{b' \in \mathcal B} \eta_{i,1,b\vert y}\, \eta_{1,j,b'\vert y}
    = \eta_{i,j,b\vert y}.
\end{align}
Combining the two identities gives $\theta_{i,1} P(b\vert y) \theta_{1,j} = \eta_{i,j,b\vert y}$.

(3) By (R1) the element $P(b\vert y)$ is self-adjoint, and by (R2) with $i = j = k = \ell = 1$ it
satisfies $P(b\vert y)^2 = P(b\vert y)$, so it is a projection; the identity
$\sum_{b \in \mathcal B} P(b\vert y) = \theta_{1,1}$ is the definition of $\theta_{1,1}$. Let $\mathfrak{D}$ be the
$C^\ast$-subalgebra generated by the $P(b\vert y)$. Since the adjoint of a generator of
$\mathfrak{C}$ is again a generator by (R1), the algebra $\mathfrak{C}$ is the closed linear span of
words in the $\eta_{i,j,b\vert y}$. Compressing such a word by $\theta_{1,1}$ and making use of Lemma \ref{lem:matrix_units} (2) gives
\begin{equation}
\begin{aligned}
    \theta_{1,1}\, \eta_{i_1,j_1,b_1\vert y_1} \cdots \eta_{i_n,j_n,b_n\vert y_n}\, \theta_{1,1}
    &= \theta_{1,1} \theta_{i_1,1} P(b_1\vert y_1) \theta_{1,j_1} \cdots
       \theta_{i_n,1} P(b_n\vert y_n) \theta_{1,j_n} \theta_{1,1} \\
    &= \delta_{1,i_1} \delta_{j_1,i_2} \cdots \delta_{j_{n-1},i_n} \delta_{j_n,1}\,
       P(b_1\vert y_1) \cdots P(b_n\vert y_n),
\end{aligned}
\end{equation}
where we used Lemma \ref{lem:matrix_units} (1) together with $\theta_{1,1}P(b\vert y) = P(b\vert y) = P(b\vert y)\theta_{1,1}$. Hence
$\theta_{1,1}\mathfrak{C}\theta_{1,1}$ is contained in the closed linear span of products of the
$P(b\vert y)$, which is $\mathfrak{D}$. The reverse inclusion is clear.
\end{proof}

\begin{lem}\label{lem:corner}
Let $\mathfrak{A}$ be a unital $C^\ast$-algebra containing a system of matrix units
$\{\theta_{i,j}\}_{i,j=1}^m$ with $\sum_{i=1}^m \theta_{i,i} = \mathds{1}$, and put
$\mathfrak{Z} \coloneqq \theta_{1,1}\mathfrak{A}\theta_{1,1}$. Then
\begin{enumerate}
    \item[(1)] $\mathfrak{Z}$ is a $C^\ast$-subalgebra of $\mathfrak{A}$ with unit $\theta_{1,1}$;
    \item[(2)] there is a $\ast$-isomorphism
          \begin{align}
              \Phi : \mathbb M_m(\mathbb{C}) \otimes \mathfrak{Z} \to \mathfrak{A}, \qquad
              \Phi(\Theta_{i,j} \otimes z) = \theta_{i,1} z \theta_{1,j},
          \end{align}
          where $\Theta_{i,j}$ is a matrix unit.
\end{enumerate}
\end{lem}

\begin{proof}
This is a standard fact, see e.g.\ \cite[Section~3]{AlRawashdeh2011}; we include the short
argument for completeness.

(1) The compression map $x \mapsto \theta_{1,1} x \theta_{1,1}$ is a continuous linear idempotent on
$\mathfrak{A}$, so its range $\mathfrak{Z}$ is a closed subspace; it is clearly closed under
products and adjoints, and $\theta_{1,1} z = z = z \theta_{1,1}$ for every $z \in \mathfrak{Z}$, so
$\mathfrak{Z}$ is a $C^\ast$-subalgebra with unit $\theta_{1,1}$.

(2) Since $\mathbb M_m(\mathbb{C}) \otimes \mathfrak{Z} \cong \mathbb M_m(\mathfrak{Z})$ is already complete, it
suffices to check that the linear extension of $\Phi$ is a $\ast$-homomorphism with a two-sided
inverse. Using $\theta_{1,j} \theta_{k,1} = \delta_{j,k} \theta_{1,1}$ and $z \theta_{1,1} = z$ for $z \in \mathfrak{Z}$,
\begin{align}
    \Phi(\Theta_{i,j} \otimes z)\, \Phi(\Theta_{k,\ell} \otimes w)
    = \theta_{i,1} z \theta_{1,j} \theta_{k,1} w \theta_{1,\ell}
    = \delta_{j,k}\, \theta_{i,1} z w \theta_{1,\ell}
    = \Phi\bigl((\Theta_{i,j} \otimes z)(\Theta_{k,\ell} \otimes w)\bigr),
\end{align}
and $\Phi(\Theta_{i,j} \otimes z)^\ast = \theta_{j,1} z^\ast \theta_{1,i} = \Phi(\Theta_{j,i} \otimes z^\ast)$, so
$\Phi$ is a $\ast$-homomorphism. Define
\begin{align}
    \Psi : \mathfrak{A} \to \mathbb M_m(\mathbb{C}) \otimes \mathfrak{Z}, \qquad
    \Psi(a) = \sum_{i,j=1}^m \Theta_{i,j} \otimes \theta_{1,i}\, a\, \theta_{j,1},
\end{align}
which is well defined since $\theta_{1,i} a \theta_{j,1} = \theta_{1,1}(\theta_{1,i} a \theta_{j,1})\theta_{1,1} \in \mathfrak{Z}$.
Then, using $\sum_{i =1}^m \theta_{i,i} = \mathds{1}$,
\begin{align}
    \Phi(\Psi(a)) = \sum_{i,j \in [m]} \theta_{i,1} \theta_{1,i}\, a\, \theta_{j,1} \theta_{1,j}
    = \Bigl(\sum_{i \in [m]} \theta_{i,i}\Bigr) a \Bigl(\sum_{j \in [m]} \theta_{j,j}\Bigr) = a,
\end{align}
and conversely $\Psi(\Phi(\Theta_{k,\ell} \otimes z)) = \sum_{i,j \in [m]} \Theta_{i,j} \otimes
\theta_{1,i} \theta_{k,1} z \theta_{1,\ell} \theta_{j,1} = \Theta_{k,\ell} \otimes z$. Hence $\Phi$ is bijective and
therefore a $\ast$-isomorphism.
\end{proof}

\begin{thm}\label{thm:isomorphism_theorem}
Let $\mathfrak{C}$ be the universal $C^\ast$-algebra generated by symbols
$\{\eta_{i,j,b\vert y}\}_{i,j,b,y}$ subject to the relations \textnormal{(R1)}--\textnormal{(R5)}
above. Then there is a canonical $\ast$-isomorphism
\begin{align}
    \mathfrak{C} \cong \mathbb M_m(\mathbb{C}) \otimes \mathfrak{B},
\end{align}
sending
\begin{align}
    \eta_{i,j,b\vert y} \longmapsto \Theta_{i,j} \otimes N(b\vert y).
\end{align}
In particular, the abstract generators $\{\eta_{i,j,b\vert y}\}_{i,j,b,y}$ with relations
\textnormal{(R1)}--\textnormal{(R5)} faithfully encode the joint system consisting of Alice's full
matrix algebra and Bob's universal $C^\ast$-algebra of projective measurements.
\end{thm}

\begin{proof}
Write $\theta_{i,j} = \sum_{b \in \mathcal B} \eta_{i,j,b\vert y}$ and $P(b\vert y) = \eta_{1,1,b\vert y}$ as in
Lemma \ref{lem:matrix_units}, and put $\mathfrak{Z} \coloneqq \theta_{1,1}\mathfrak{C}\theta_{1,1}$. By
Lemma \ref{lem:matrix_units} (1) the family $\{\theta_{i,j}\}_{i,j}$ is a system of matrix units in
$\mathfrak{C}$ with $\sum_{i \in [m]} \theta_{i,i} = \mathds{1}$, so Lemma \ref{lem:corner} applies and yields a
$\ast$-isomorphism
\begin{align}
    \Phi : \mathbb M_m(\mathbb{C}) \otimes \mathfrak{Z} \to \mathfrak{C}, \qquad
    \Phi(\Theta_{i,j} \otimes z) = \theta_{i,1} z \theta_{1,j}.
\end{align}
It therefore remains to identify $\mathfrak{Z}$ with $\mathfrak{B}$.

By Lemma \ref{lem:matrix_units} (3) the elements $P(b\vert y)$ are projections in $\mathfrak{Z}$
with $\sum_{b \in \mathcal B} P(b\vert y) = \theta_{1,1} = \mathds{1}_{\mathfrak{Z}}$, so they satisfy the defining relations of
$\mathfrak{B}$. By the universal property of $\mathfrak{B}$ there is a unital
$\ast$-homomorphism
\begin{align}
    \beta : \mathfrak{B} \to \mathfrak{Z}, \qquad \beta(N(b\vert y)) = P(b\vert y).
\end{align}
Its range is a $C^\ast$-subalgebra of $\mathfrak{Z}$ containing all $P(b\vert y)$, hence equals
$\mathfrak{Z}$ by Lemma \ref{lem:matrix_units} (3), and $\beta$ is surjective.

For injectivity we produce a left inverse. The elements
$\Theta_{i,j} \otimes N(b\vert y) \in \mathbb M_m(\mathbb{C}) \otimes \mathfrak{B}$ satisfy
(R1)--(R5): the adjoint of $\Theta_{i,j} \otimes N(b\vert y)$ is $\Theta_{j,i} \otimes N(b\vert y)$, which is
(R1); for a common $b$ and $y$,
\begin{align}
    (\Theta_{i,j} \otimes N(b\vert y))(\Theta_{k,\ell} \otimes N(b\vert y))
    = \delta_{j,k}\, \Theta_{i,\ell} \otimes N(b\vert y)^2
    = \delta_{j,k}\, \Theta_{i,\ell} \otimes N(b\vert y),
\end{align}
which is (R2); for $b \neq b'$ and a common $y$ the product vanishes because
$N(b\vert y)N(b'\vert y) = 0$, which is (R3); moreover
\begin{align}
    \sum_{i=1}^m \sum_{b \in \mathcal B} \Theta_{i,i} \otimes N(b\vert y)
    = \sum_{i=1}^m \Theta_{i,i} \otimes \mathds{1} = \mathds{1},
\end{align}
which is (R4); and $\sum_{b \in \mathcal B} \Theta_{i,j} \otimes N(b\vert y) = \Theta_{i,j} \otimes \mathds{1}$ does not depend on $y$,
which is (R5). By the universal property of $\mathfrak{C}$ there is a $\ast$-homomorphism
\begin{align}
    \pi : \mathfrak{C} \to \mathbb M_m(\mathbb{C}) \otimes \mathfrak{B}, \qquad
    \pi(\eta_{i,j,b\vert y}) = \Theta_{i,j} \otimes N(b\vert y).
\end{align}
It satisfies $\pi(\theta_{i,j}) = \Theta_{i,j} \otimes \mathds{1}$, hence
\begin{align}
    \pi(\mathfrak{Z}) = \pi(\theta_{1,1}\mathfrak{C}\theta_{1,1})
    \subseteq (\Theta_{1,1} \otimes \mathds{1})\, (\mathbb M_m(\mathbb{C}) \otimes \mathfrak{B})\, (\Theta_{1,1} \otimes \mathds{1})
    = \Theta_{1,1} \otimes \mathfrak{B}.
\end{align}
Let $\gamma : \Theta_{1,1} \otimes \mathfrak{B} \to \mathfrak{B}$ be the canonical $\ast$-isomorphism
$\Theta_{1,1} \otimes z \mapsto z$. Then $\gamma \circ \pi \circ \beta$ is a $\ast$-endomorphism of
$\mathfrak{B}$ with
\begin{align}
    (\gamma \circ \pi \circ \beta)(N(b\vert y))
    = \gamma(\pi(\eta_{1,1,b\vert y}))
    = \gamma(\Theta_{1,1} \otimes N(b\vert y))
    = N(b\vert y),
\end{align}
so it agrees with the identity on a generating set and therefore
$\gamma \circ \pi \circ \beta = \mathrm{id}_{\mathfrak{B}}$. In particular $\beta$ is injective,
hence a $\ast$-isomorphism $\mathfrak{B} \cong \mathfrak{Z}$.

Since $\beta$ is a $\ast$-isomorphism, so is $\mathrm{id}_{\mathbb M_m(\mathbb{C})} \otimes \beta$, and
hence also the composition
\begin{align}
    \Lambda \coloneqq \Phi \circ (\mathrm{id}_{\mathbb M_m(\mathbb{C})} \otimes \beta) :
    \mathbb M_m(\mathbb{C}) \otimes \mathfrak{B} \to \mathfrak{C}.
\end{align}
On generators, Lemma \ref{lem:matrix_units} (2) gives
\begin{align}
    \Lambda(\Theta_{i,j} \otimes N(b\vert y))
    = \Phi(\Theta_{i,j} \otimes P(b\vert y))
    = \theta_{i,1} P(b\vert y) \theta_{1,j}
    = \eta_{i,j,b\vert y}.
\end{align}
Thus $\Lambda$ is a $\ast$-isomorphism $\mathbb M_m(\mathbb{C}) \otimes \mathfrak{B} \to \mathfrak{C}$, and
its inverse is $\pi$, since $\pi \circ \Lambda$ fixes the generators
$\Theta_{i,j} \otimes N(b\vert y)$. This is the asserted isomorphism.
\end{proof}

As an alternative to conditions (R1)--(R5), one can also represent the algebra of several projections by commutation relations, as known from the Pauli algebra. In the following, we only show the resulting optimization program for the case of qubit measurements, since this is the situation usually encountered. Indeed, we have the three operators $S_X,S_Y,S_Z$, which essentially obey the same (anti-)commutation relations as the Pauli matrices do. In our bipartite setup, we generalize them, similarly to the $\eta_{i,j,b\vert y}$ operators, to operators $S_{X,b\vert y}$, $S_{Y,b\vert y}$, $S_{Z,b\vert y}$. The resulting program then looks as follows: $H(A\vert X=x_0,Q_E)$ is lower bounded by
\begin{equation}\label{eq:pauli_generator_entropy_program}
\begin{aligned}
 \inf \quad & \ \frac{1}{\ln 2}\left( \vert A\vert -1 + \sum_{k=0}^r\sum_{a\in A}
\bra{\psi} -\alpha_k \left( \sum_{b\in B} \sum_{\mu\in\{0,X,Y,Z\}} \lambda_{\mu}^{(a\mid x_0)} S_{\mu,b\mid \tilde{y}} \right) P_k^{(a)} -\beta_k P_k^{(a)}\ket{\psi} \right) \\
\operatorname{s.th.}\quad &\sum_{a,b,x,y} c_{abxy}^{(t)}\bra{\psi} \sum_{\mu\in\{0,X,Y,Z\}} \lambda_{\mu}^{(a\mid x)} S_{\mu,b\mid y} \ket{\psi} \geq q_t,\quad  t\in \mathcal C,\\
&(P_k^{(a)})^2=P_k^{(a)},\quad (P_k^{(a)})^\ast=P_k^{(a)}, \quad 0\leq k\leq r,\ a\in \mathcal A,\\
&[S_{\mu,b\mid y},P_k^{(a)}] = 0, \quad\mu\in\{0,X,Y,Z\},\ b\in \mathcal B,\ y\in \mathcal Y,\ 0\leq k\leq r,\ a\in \mathcal A, \\
&S_{\mu,b\mid y}^\ast=S_{\mu,b\mid y}, \quad \mu\in\{0,X,Y,Z\},\ b\in \mathcal B,\ y\in \mathcal Y, \\
&S_{0,b\mid y}S_{\mu,b'\mid y}=S_{\mu,b'\mid y}S_{0,b\mid y}= \delta_{b,b'}S_{\mu,b\mid y}, \\
&S_{X,b\mid y}^2=S_{Y,b\mid y}^2=S_{Z,b\mid y}^2=S_{0,b\mid y},\\
&S_{X,b\mid y}S_{Y,b'\mid y}=i\delta_{b,b'}S_{Z,b\mid y},\quad S_{Y,b\mid y}S_{X,b'\mid y}=-i\delta_{b,b'}S_{Z,b\mid y},\\
&S_{Y,b\mid y}S_{Z,b'\mid y} = i\delta_{b,b'}S_{X,b\mid y}, \quad S_{Z,b\mid y}S_{Y,b'\mid y} = -i\delta_{b,b'}S_{X,b\mid y}, \\
&S_{Z,b\mid y}S_{X,b'\mid y}=i\delta_{b,b'}S_{Y,b\mid y}, \quad S_{X,b\mid y}S_{Z,b'\mid y} = -i\delta_{b,b'}S_{Y,b\mid y}, \\
&\sum_{b\in \mathcal B}S_{0,b\mid y}=\mathds{1}, \quad y\in \mathcal Y, \\
& \sum_{b\in \mathcal B}S_{\mu,b\mid y} = \sum_{b\in \mathcal B}S_{\mu,b\mid y'}, \quad \mu\in\{0,X,Y,Z\},\ y,y'\in \mathcal Y .
\end{aligned}
\end{equation}
 Observe that we choose a fixed but arbitrary $\tilde{y} \in \mathcal{Y}$ in the program above. In practice, it needs to be determined case by case whether the representation in \eqref{eq:pauli_generator_entropy_program} or (R1)--(R5) is numerically more appropriate.

\end{document}